\RequirePackage{fix-cm}
\documentclass[smallextended]{svjour3}       % onecolumn (second format)
\smartqed  % flush right qed marks, e.g. at end of proof
\usepackage{graphicx}
\usepackage{mathrsfs}
\usepackage{amsfonts}
\usepackage{graphicx}
\usepackage{amssymb}
\usepackage{amsmath}
\usepackage{color}
\usepackage[misc]{ifsym}

\newtheorem{open problem}{Open Problem}
\begin{document}
	
	\title{Linear codes with few weights from non-weakly regular plateaued functions %\thanks{Grants or other notes
	}
	
	
	\author{Yadi Wei  \and
		Jiaxin Wang \and
		Fang-Wei Fu%etc.
	}
	
	
	\institute{Yadi Wei\at
		Chern Institute of Mathematics and LPMC, Nankai University, Tianjin, 300071, China\\
		\email{wydecho@mail.nankai.edu.cn}           %  \\
		\and 
		Jiaxin Wang\at
		Chern Institute of Mathematics and LPMC, Nankai University, Tianjin, 300071, China\\
		\email{wjiaxin@mail.nankai.edu.cn}
		\and
		Fang-Wei Fu\at
		Chern Institute of Mathematics and LPMC, Nankai University, Tianjin, 300071, China\\
		\email{fwfu@nankai.edu.cn}}
	\date{Received: date / Accepted: date}
	
	\maketitle
	
	\begin{abstract}
		Linear codes with few weights have significant applications in secret sharing schemes, authentication codes, association schemes, and strongly regular graphs. There are a number of  methods to construct linear codes, one of which is based on functions. Furthermore, two generic constructions of linear codes from functions called the first and the second generic constructions, have aroused the research interest of scholars. Recently, in \cite{Nian},  Li and Mesnager proposed two open problems: Based on the first and the second generic constructions, respectively, construct linear codes from non-weakly regular plateaued functions and determine their weight distributions.
		
		Motivated by these two open problems, in this paper, firstly, based on the first generic construction, we construct some three-weight and at most five-weight linear codes from non-weakly regular plateaued functions. When the related functions belong to a special class of non-weakly regular plateaued functions, we determine the weight distributions of the constructed codes.  Next, based on the second generic construction, we construct some three-weight and at most five-weight linear codes from non-weakly regular plateaued functions belonging to $\mathcal{NWRF}$ (defined in this paper). Moreover, when the related functions belong to a special class of $\mathcal{NWRF}$, we determine the weight distributions of the constructed codes. Specially, using non-weakly regular bent functions, that is, non-weakly regular $0$-plateaued functions, we obtain some three-weight, at most four-weight and at most five-weight linear codes. We also give the punctured codes of these codes obtained based on the second generic construction and determine their weight distributions. Meanwhile, we obtain some optimal and almost optimal linear codes. Besides, by the Ashikhmin-Barg condition, we have that the constructed codes are minimal for almost  all cases and obtain some secret sharing schemes with nice access structures based on their dual codes. 
		
	\end{abstract}
	
	\section{Introduction}\label{Section1}
	\quad
	Linear codes  with few weights  have a wide range applications in secret sharing schemes \cite{Anderson,Carlet1,Ding1,Yuan}, authentication codes \cite{Ding2}, association schemes \cite{Calderbank1}, and strongly regular graphs \cite{Calderbank2}. As an important class of linear codes, minimal linear codes have applications in secret sharing schemes \cite{Massey} and secure two-party computation \cite{Chabanne}. Therefore, the constructions of linear codes with few weights and minimal linear codes have attracted much attention in coding theory. There are a number of methods to construct linear codes, one of which is based on functions. Two generic constructions of linear codes from functions, called the first and the second generic constructions, have been distinguished from others and widely used by researchers. Due to the properties of cryptographic functions, it's a good way to construct linear codes from cryptographic functions. By the first and the second generic constructions, several linear codes with few weights have been constructed from cryptographic functions such as bent functions \cite{Mesnager1,Ozbudak,Tang1,Zhou}, perfect nonlinear (PN) functions \cite{Carlet1,Yuan1,Li} and plateaued functions \cite{Mesnager2,Mesnager3,Cheng,Sinak,Yang}. 
	
	The research on constructing linear codes from plateaued functions, to the best of our knowledge, first appeared in \cite{Mesnager2}, where Mesnager et al. constructed some three-weight linear codes from weakly regular plateaued functions based on the first generic construction. Later, in \cite{Mesnager3}, based on the second generic construction, Mesnager and Sinak constructed several new classes of linear codes from weakly regular plateaued functions. In \cite{Cheng}, Cheng and Cao defined some defining sets different from those in \cite{Mesnager3} by using weakly regular plateaued functions, and constructed some linear codes with few weights. In \cite{Sinak}, using weakly regular plateaued balanced functions, Sinak constructed several three-weight and four-weight minimal codes using the second generic construction. Recently, in \cite{Yang}, Yang et al. chose the defining sets from two distinct weakly regular plateaued balanced functions and obtained a family of linear codes
	with at most five weights based on the second generic construction. The previously known research works of constructing linear codes from plateaued functions  all utilize weakly regular plateaued functions, and there is no result on constructing linear codes by using non-weakly regular plateaued functions. In \cite{Nian}, Li and Mesnager proposed the following two open problems:
	
	$Problem\ 1$: Based on the first generic construction, construct linear codes from non-weakly regular plateaued functions and determine their weight distributions.
	
	$Problem\ 2$: Based on the second generic construction, construct linear codes from non-weakly regular plateaued functions and determine their weight distributions.
	
	Motivated by the above two open problems, in this paper, we consider constructing linear codes from non-weakly regular plateaued functions and provide answers to the two open problems. The main results are summarized as follows:
	\begin{itemize}
	\item[$\bullet$] Based on the first generic construction, we construct some three-weight and at most five-weight linear codes from non-weakly regular plateaued function $f$. When the related function $f$ belongs to a special class of  non-weakly regular plateaued functions, that is, the dual of $f$ is bent relative to the Walsh support $\mathrm{Supp}(\widehat{\chi_f})$, we determine the weight distributions of the constructed codes.
		\item[$\bullet$]Based on the second generic construction, using non-weakly regular 
		plateau-ed function $f$ belonging to $\mathcal{NWRF}$ (defined in this paper), we construct some three-weight and at most five-weight linear codes by choosing defining sets $D_f=\{x\in \mathbb{F}_{p}^n\setminus\{0\}:\ f(x)=0\},\ D_{f,SQ}=\{x\in \mathbb{F}_{p}^n\setminus\{0\}:\ f(x)\in SQ\}\ \text{and}\ D_{f,NSQ}=\{x\in \mathbb{F}_{p}^n\setminus\{0\}:\ f(x)\in NSQ\}$, where $SQ$ and $NSQ$ denote, respectively, the set of all squares and non-squares in the multiplicative group $\mathbb{F}_{p}^{\times}$ of $\mathbb{F}_p$. When the related function $f$ in $\mathcal{NWRF}$ belongs to a special class of  non-weakly regular plateaued functions, that is, the dual of $f$ is bent relative to the Walsh support $\mathrm{Supp}(\widehat{\chi_f})$, we determine the weight distributions of the constructed codes. Specially, when $f$ is a non-weakly regular bent function, i.e., non-weakly regular $0$-plateaued function, we obtain some three-weight, at most four-weight and at most five-weight linear codes.
		\item[$\bullet$]We give the punctured codes of the constructed linear codes based on the second generic construction and determine their weight distributions. Meanwhile, we determine the parameters of the dual codes of the constructed codes. Besides, we obtain some optimal and almost optimal linear codes.
	\item[$\bullet$]By the Ashikhmin-Barg condition, we prove that the constructed linear codes are minimal for almost all cases and give some results on the access structures of the secret sharing schemes based on the dual codes of the constructed minimal codes.
	\end{itemize}

	The rest of this paper is arranged as follows. In Section 2, we introduce the main notation and give some necessary preliminaries. In Section 3 and Section 4, based on the first and the second generic constructions, we construct some linear codes with few weights from non-weakly regular plateaued functions and determine the weight distributions of the linear codes obtained from two special classes of non-weakly regular plateaued functions. In Sections 5, we make a conclusion. 
	
	\section{Preliminaries}\label{Section2}
	\quad In this section, we set the basic notation which will be used in the following sections and introduce  some basic results on cyclotomic field, exponential sums, non-weakly regular plateaued functions, linear codes and secret sharing schemes.
	\subsection{Notation}\label{Section2.1}
	\quad Let $p$ be an odd prime, $n$ be a positive integer and $\mathbb{F}_{p^n}$ be the finite field of order $p^n$. Since $\mathbb{F}_{p^n}$ is an $n$-dimensional vector space over $\mathbb{F}_p$, we also use the notation $\mathbb{F}_p^n$ consisting of $n$-tuples over the prime field $\mathbb{F}_p$. Throughout this paper, we set the following notation.
	\begin{itemize}
		\item[$\bullet$] $\#D$: The cardinality of a set $D$;
		\item[$\bullet$] $D^{\times}$: The set of non-zero elements in a set $D$;
		\item[$\bullet$] $SQ$ and $NSQ$: The set of squares and nonsquares in $\mathbb{F}_p^{\times}$, respectively;
		\item[$\bullet$] $\eta$: The quadratic character of $\mathbb{F}_p^{\times}$, that is, $\eta(a)=1$ for $a \in SQ$ and $\eta(a)=-1$ for $a \in NSQ$;
		\item[$\bullet$] $p^*$: $p^*=\eta(-1)p=(-1)^{(p-1)/2}p$;
		\item[$\bullet$] $\xi_p=e^{\frac{2\pi i}{p}} $: The $p$-th complex primitive root of unity, where $i=\sqrt{-1}$.
	\end{itemize}
	
	\subsection{Cyclotomic Field\ $\mathbb{Q}(\xi_p)$}\label{Section2.2}
	\quad The cyclotomic field $\mathbb{Q}(\xi_p)$ is obtained from the rational field $\mathbb{Q}$ by adjoining $\xi_p$. The ring of algebraic integers in $\mathbb{Q}(\xi_p)$ is $\mathcal{O}_{ {\mathbb{Q}_{(\xi_p)}} }=\mathbb{Z}[\xi_p]$ and  the set $\{\xi_p^j\ :\ 1\le j\le p-1\}$ is an integral basis of it. The field extension $\mathbb{Q}(\xi_p)/\mathbb{Q}$ is Galois extension of degree $p-1$, and $Gal(\mathbb{Q}(\xi_p)/\mathbb{Q})=\{\sigma_a:\ a\in \mathbb{F}_p^{\times}\}$ is the Galois group, where the automorphism $\sigma_a$ of $\mathbb{Q}(\xi_p)$ is defined by $\sigma_a(\xi_p)=\xi_p^a.$ The cyclotomic field $\mathbb{Q}(\xi_p)$ has a unique quadratic subfield $\mathbb{Q}(\sqrt{p^{*}})$, and so $Gal(\mathbb{Q}(\sqrt{p^{*}})/\mathbb{Q})=\{1,\ \sigma_\gamma\}$, where $\gamma$ is a nonsquare in $\mathbb{F}_p^{\times}.$ Obviously, for $a\in \mathbb{F}_p^{\times}$ and $b \in \mathbb{F}_p$, we have that $\sigma_a(\xi_p^b)=\xi_p^{ab}$ and $\sigma_a(\sqrt{p^*})=\eta(a)\sqrt{p^*}$. For more information on cyclotomic field, the reader is referred to \cite{Ireland}.
	\subsection{Exponential Sums}\label{Section2.3}
	\quad In this subsection, we give some results about exponential sums over the finite field $\mathbb{F}_p$.
	\begin{lemma}\cite{Lidl}With the defined notation, we have
		\begin{itemize}
			\item[$\mathrm{(1)}$] $\sum_{y\in \mathbb{F}_p^{\times}}\eta(y)=0;$
			\item[$\mathrm{(2)}$] $\sum_{y \in \mathbb{F}_p^{\times}}\xi_p^{zy}=-1\ \text{for any}\ z\in \mathbb{F}_p^{\times};$
			\item[$\mathrm{(3)}$] $\sum_{y\in \mathbb{F}_p^{\times}}\eta(y)\xi_p^y=\sqrt{p^*}=\begin{cases}
			\sqrt{p}&\text{if}\ p\equiv 1\ (\mathrm{mod}\ 4),\\
			i\sqrt{p}&\text{if}\ p\equiv3 \ (\mathrm{mod}\ 4);
			\end{cases}$
			\item[$\mathrm(4)$] $\sum_{y\in SQ}\xi_p^{ya}=\frac{\eta(a)\sqrt{p^*}-1}{2}\ \text {for any}\ a \in \mathbb{F}_p^{\times};$
			
			\item[$\mathrm(5)$] $\sum_{y\in NSQ}\xi_p^{ya}=\frac{-\eta(a)\sqrt{p^*}-1}{2}\  \text {for any}\ a \in \mathbb{F}_p^{\times}.$
		\end{itemize}
	\end{lemma}
	\subsection{Non-Weakly Regular Plateaued Functions}\label{Section2.4}
	\hspace{0.4cm}
	In this subsection, we give the definition of non-weakly regular plateaued functions and introduce some properties of them. 
	
	Let $f:\mathbb{F}_{p}^n\longrightarrow\mathbb{F}_p$ be a $p$-ary function. Its $\mathit{Walsh\  transform}$ is defined by 
	\[\widehat{\chi_f}(\alpha)=\sum\limits_{x\in \mathbb{F}_{p}^n} \xi_p^{f(x)-\alpha\cdot x},\ \alpha \in \mathbb{F}_{p}^n,\]
	where $\cdot$ is the ordinary inner product in $\mathbb{F}_p^n$. The function $f(x)$ is said to be $\mathit{balanced}$ over $\mathbb{F}_p$ if $\widehat{\chi_f}(0)=0$, otherwise it is $\mathit{unbalanced}.$ The function $f(x)$ is called $\mathit{bent}$, if $|\widehat{\chi_f}(\alpha)|=p^{\frac{n}{2}}$ for any $\alpha \in \mathbb{F}_{p}^n$. The Walsh transform of a bent function $f(x)$ at $\alpha\in\mathbb{F}_p^n$ is given as follows \cite{Kumar}:
	\[\widehat{\chi_f}(\alpha)=
	\begin{cases}
	\pm p^{\frac{n}{2}}\xi_{p}^{f^{*}(\alpha)},& \text{if } p^{n}\equiv 1\ (\mathrm{mod}\ 4);\\ 
	\pm ip^{\frac{n}{2}}\xi_{p}^{f^{*}(\alpha)},& \text{ if } p^{n}\equiv 3\ (\mathrm{mod}\ 4),
	\end{cases}\]
	where $f^*(x)$ is a function from $\mathbb{F}_{p}^n$ to $\mathbb{F}_p$ called the $dual$ of $f(x)$.
	A bent function $f(x)$ is called $weakly\ regular$ if for all $\alpha\in\mathbb{F}_{p}^n$, $p^{-\frac{n}{2}}\widehat{\chi_f}(\alpha)=\epsilon\xi_p^{f^*(\alpha)}$ for some fixed $\epsilon\in\{\pm1,\pm i\}$.  If $\epsilon$ changes with respect to $\alpha\in \mathbb{F}_{p}^n$, then $f(x)$ is called $non$-$weakly\ regular$. Specially, if $\epsilon=1$, $f(x)$ is called $regular$. Note that the dual of a weakly regular bent function is also weakly regular bent and satisfies $f^{**}(x)=f(-x)$, where $f^{**}(x)$ is the dual of $f^*(x)$ given by $\widehat{\chi_{f^*}}(\alpha)=\epsilon p^{\frac{n}{2}}\xi_p^{f^{**}(\alpha)}$ for all $\alpha\in \mathbb{F}_p^n$, where $\epsilon\in\{\pm1,\pm i\}$. However, this does not generally hold for non-weakly regular bent functions \cite{Cesmelioglu1,Cesmelioglu2}. In \cite{Ozbudak1}, Özbudak et al. proved that if $f(x)$ is a non-weakly regular bent function such that its dual $f^*(x)$ is bent, then $f^*(x)$ is non-weakly regular bent and satisfies $f^{**}(x)=f(-x)$, where $f^{**}(x)$ is the dual of $f^*(x)$ given by $\widehat{\chi_{f^*}}(\alpha)=\epsilon_{\alpha} p^{\frac{n}{2}}\xi_p^{f^{**}(\alpha)}$ for all $\alpha\in \mathbb{F}_p^n$, where $\epsilon_{\alpha}\in\{\pm1,\pm i\}$.
	
	As an extension of bent functions, plateaued functions were proposed in \cite{Zheng} and generalized to any characteristic $p$ in \cite{Mesnager4}. A function $f(x)$ from $\mathbb{F}_{p}^n$ to $\mathbb{F}_{p}$ is said to be $s$-$plateaued$ if $|\widehat{\chi_f} (\alpha)|\in\{0,p^{\frac{n+s}{2} }\} $ for any $\alpha\in \mathbb{F}_{p}^n$, where $0\le s\le n$. Clearly, when $s=0$, a $0$-plateaued function is actually a bent function. The $Walsh\ support$ of $f(x)$ is defined by $\mathrm{Supp}(\widehat{\chi_f})=\{\alpha\in \mathbb{F}_{p}^n:\  |\widehat{\chi_f}(\alpha)|=p^{\frac{n+s}{2}} \}$. By the Parseval identity $\sum\limits_{\alpha\in\mathbb{F}_p^n}|\widehat{\chi_f}(\alpha)|^2=p^{2n}$, we have $\#\mathrm{Supp}(\widehat{\chi_f})=p^{n-s}$. The Walsh transform of an $s$-plateaued function $f(x)$ at $\alpha\in\mathbb{F}_p^n$ is given as follows \cite{Hyun}:
	\[\widehat{\chi_f}(\alpha)=
	\begin{cases}
	\pm p^{\frac{n+s}{2}}\xi_{p}^{f^{*}(\alpha)}, 0, & \text{if } p^{n+s}\equiv 1\ (\mathrm{mod}\ 4);\\ 
	\pm ip^{\frac{n+s}{2}}\xi_{p}^{f^{*}(\alpha)}, 0,& \text{ if }  p^{n+s}\equiv 3\ (\mathrm{mod}\ 4),
	\end{cases}\] where $f^*(x)$ is a function from $\mathrm {Supp}(\widehat{\chi_f})$ to $\mathbb{F}_p$ called the $dual$ of $f(x)$.
	Similar to the  classification of bent functions,  $f(x)$ is called $weakly\ regular$, if for all $\alpha \in \mathrm {Supp}(\widehat{\chi_f})$,   $p^{-\frac{n+s}{2}}\widehat{\chi_f}(\alpha)=\epsilon\xi_p^{f^*(x)}$ for some fixed $\epsilon\in \{\pm 1,\pm i\}$. If $\epsilon$ changes with respect to $\alpha$, then $f(x)$ is called $non$-$weakly \ regular$.  Specially, if $\epsilon=1$, $f(x)$ is called $regular$.
	
	The following definition is given in \cite{Ozbudak1}.
	
	\begin{definition}
		Let $S$ be a subset of $\mathbb{F}_{p}^n$ with $\#S=N$ and $f(x)$ be a function from $S$ to $\mathbb{F}_p$. If $|\widehat{\chi_f}(\alpha)|=N^{\frac{1}{2} }$ for all $\alpha\in\mathbb{F}_{p}^n$, then $f(x)$ is called $\mathit{bent\ relative}$ to $S$, where $\widehat{\chi_f}(\alpha)=\sum\limits_{x\in S}\xi_p^{f(x)-\alpha \cdot x}.$
	\end{definition}
	\begin{remark}
		In \cite{Mesnager2}, Measnager et al.  proved that the dual of a weakly regular $s$-plateaued function $f(x)$ is bent relative to $\mathrm{Supp}(\widehat{\chi_f})$. Moreover, \cite{Ozbudak1} proved that if $f(x)$ is a non-weakly regular $s$-plateaued function such that its dual $f^*(x)$ is  bent relative to $\mathrm{Supp}(\widehat{\chi_f})$, then $f^*(x)$ is non-weakly regular bent relative to  $\mathrm{Supp}(\widehat{\chi_f})$ and satisfies $f^{**}(x)=f(-x)$, where $f^{**}(x)$ from $\mathbb{F}_{p}^n$ to $\mathbb{F}_p$ is the dual of $f^*(x)$ given by $\widehat{\chi_{f^*}}(\alpha)=\epsilon _{\alpha}p^{\frac{n-s}{2}}\xi_p^{f^{**}(\alpha)}$ for all $\alpha\in \mathbb{F}_p^n$, where $\epsilon_{\alpha}\in\{\pm1,\pm i\}$.
	\end{remark}
	
	Let $\epsilon=1$ if $p^{n+s}\equiv1$ (mod $4$) and  $\epsilon=i$ if $p^{n+s}\equiv3$ (mod $4$). For an $s$-plateaued function $f(x):\mathbb{F}_{p}^n\longrightarrow\mathbb{F}_p$, we define $B_+(f)\ \text{and}\ B_{-}(f)$ as follows.
	\begin{align*}
	B_+(f)&:=\{\alpha\in \mathrm{Supp}(\widehat{\chi_f}):\ \widehat{\chi_f}(\alpha)=\epsilon p^{\frac{n+s}{2}}\xi_p^{f^*(\alpha)}\},\\
	B_-(f)&:=\{\alpha\in \mathrm{Supp}(\widehat{\chi_f}):\ \widehat{\chi_f}(\alpha)=-\epsilon p^{\frac{n+s}{2}}\xi_p^{f^*(\alpha)}\}.
	\end{align*}
	If $f(x)$ is unbalanced, we define the type of $f(x)$ as 
	$f(x)\  \text{is of}\  type\  (+) \ \text{if\  }\  \widehat{\chi_f}(0)\\=\epsilon p^{\frac{n+s}{2}}\xi_p^{f^*(0)}\ \text{and of}\  type\ (-)\ \text{if\ }\  \widehat{\chi_f}(0)=-\epsilon p^{\frac{n+s}{2}}\xi_p^{f^*(0)}.$ Besides, if $f(x)$ is non-weakly regular and the dual $f^*(x)$ of $f(x)$ is bent relative to $\mathrm{Supp}(\widehat{\chi_f})$, then we define $B_+(f^*)$ and $B_-(f^*)$ as follows.
	\begin{align*}
	B_+(f^*)&:=\{\alpha\in \mathbb{F}_{p}^n:\ \widehat{\chi_{f^*}}(\alpha)=\epsilon p^{\frac{n-s}{2}}\xi_p^{f^{**}(\alpha)}\},\\
	B_-(f^*)&:=\{\alpha\in \mathbb{F}_{p}^n:\ \widehat{\chi_{f^*}}(\alpha)=-\epsilon p^{\frac{n-s}{2}}\xi_p^{f^{**}(\alpha)}\}.
	\end{align*}
	Meanwhile, we define the type of $f^*(x)$ as 
	$f^*(x)\  \text{is of}\  type\  (+)\  \text{if\  }\  \widehat{\chi_{f^*}}(0)=\epsilon p^{\frac{n-s}{2}}\xi_p^{f^{**}(0)}\ \text{and of}\  type\ (-)\ \text{if\ }\  \widehat{\chi_{f^*}}(0)=-\epsilon p^{\frac{n-s}{2}}\xi_p^{f^{**}(0)}.$
	
	The following lemma gives the value distributions of unbalanced $p$-ary $s$-plateaued functions.
	
	\begin{lemma} \cite{Ozbudak1}Let $f(x):\mathbb{F}_{p}^n\longrightarrow\mathbb{F}_p$ be an unbalanced $p$-ary $s$-plateaued function with $f^{*}(0)=j_0$. For any $j\in \mathbb{F}_p$, define $N_j(f)=\#\{x\in \mathbb{F}_p^n:\ f(x)=j\}.$                                                                                                                                                                                                                                                                                                                                                                                                                                                                                                                                                                                                                                                                                                                                                                                                                                                                                                                                                                                                                                                                                                                                                                                                                                                                                                                                                                                                                                                                                                                                                                                                                                                                                                                                                                                                                                                                                                                                                                                                                                                                                                                                                                                                                                                                                                                                                                                                                                                                                                                                                                                                                                                                                                                                                                                                                                                                                                                                                                                                                                                                                                                                                                                                                                                                                                                                                                                                                                                                                                                                                                                                                                                                                                                                                                                                                                                                                                                                                                                                                                                                                                                                                                                                                                                                                                                                                                                                                                                                                                                                                                                                                                                                                                                                                                                                                                                                                                                                                                                                                                                                                                                                                                                                                                                                                                                                                                                                                                                                                                                                                                                                                                                                                                                                                                                                                                                                                                                                                                                                                                                                                                                                                                                                                                                                                                                                                                                                                                                                                                                                                                                                                                                                                                                                                                                                                                                                                                                                                                                                                                                                                                                                                                                                                                                                                                                                                                                                                                                                                                                                                                                                                                                                                                                                                                                                                                                                                                                                                                                                                                                                                                                                                                                                                                                                                                                                                                                                                                                                                                                                                                                                                                                                                                                                                                                                                                                                                                                                                                                                                                                                                                                                                                                                                                                                                                                                                                                                                                            
		\begin{itemize}
			\item[$\bullet$]When $n+s$ is even, we have\\$N_{j_0}(f)=p^{n-1}\pm p^{\frac{n+s}{2}}\mp p^{\frac{n+s}{2}-1},$ $N_j(f)=p^{n-1}\mp p^{\frac{n+s}{2}-1},$ for $j\ne j_0\in \mathbb{F}_p$.
			\item[$\bullet$]When $n+s$ is odd, we have\\
			$N_{j_0}(f)=p^{n-1},\ N_{j_0+j}(f)=p^{n-1}\pm\eta(j)p^{\frac{n+s-1}{2}},$ for $ 1\leq j\leq p-1$.
		\end{itemize}

		Here the sign is $+$ (respectively $-$) if and only if the type of $f$ is $(+)$ (respectively $(-)$).
	\end{lemma}	
	
	When the dual $f^*(x)$ of a non-weakly regular $s$-plateaued function $f(x)$ is bent relative to $\mathrm{Supp}(\widehat{\chi_f})$, we give the value distributions of $f^*(x)$ in the following lemma.
	\begin{lemma}
		Let $f(x):\mathbb{F}_{p}^n\longrightarrow\mathbb{F}_p$ be a  non-weakly regular $s$-plateaued function whose dual $f^*(x)$ is bent relative to $\mathrm{Supp}(\widehat{\chi_f})$ and $f(0)=j_0$. For any $j \in \mathbb{F}_p$, define $N_j(f^*)=\#\{x\in \mathrm{Supp}(\widehat{\chi_f}):\ f^*(x)=j\}.$    
		\begin{itemize}
		\item[$\bullet$]When $n+s$ is even, we have\\$N_{j_0}(f^*)=p^{n-s-1}\pm p^{\frac{n-s}{2}}\mp p^{\frac{n-s}{2}-1},$ $N_j(f^*)=p^{n-s-1}\mp p^{\frac{n-s}{2}-1},$ for $j\ne j_0\in \mathbb{F}_p$.
		\item[$\bullet$]When $n+s$ is odd, we have\\
		$N_{j_0}(f^*)=p^{n-s-1},\ N_{j_0+j}(f^*)=p^{n-s-1}\pm\eta(j)p^{\frac{n-s-1}{2}},$ for $ 1\leq j\leq p-1$.
	\end{itemize}

	Here the sign is $+$ (respectively $-$) if and only if the type of $f^*$ is $(+)$ (respectively $(-)$).
\end{lemma}	

\begin{proof} For any $j\in\mathbb{F}_p$, we have
	\begin{align*}
	N_j(f^*)&=p^{-1}\sum\limits_{x\in\mathrm{Supp}(\widehat{\chi_f})}\sum\limits_{y\in\mathbb{F}_p}\xi_p^{y(f^*(x)-j)}\\
	&=p^{-1}\sum\limits_{y\in\mathbb{F}_p^{\times}}\sum\limits_{x\in\mathrm{Supp}(\widehat{\chi_f})}\xi_p^{y(f^*(x)-j)}+p^{n-s-1}\\
	&=p^{-1}\sum\limits_{y\in\mathbb{F}_p^{\times}}\sigma_y(\widehat{\chi_{f^*}}(0))\xi_p^{-yj}+p^{n-s-1}.
	\end{align*}
	Since $f^*(x)$ is bent relative to $\mathrm{Supp}(\widehat{\chi_f})$, then we have $\widehat{\chi_{f^*}}(0)=\epsilon_0\epsilon p^{\frac{n-s}{2}}\xi_p^{f^{**}(0)}$, where $\epsilon\in\{1,i\}$ and $\epsilon_0=1$ (respectively $\epsilon_0=-1$) if $0\in B_+(f^*)$ (respectively $0\in B_-(f^*)$). By Remark 1, we know that $f^{**}(0)=f(0)=j_0$. By Lemma 1, the value of $N_j(f^*)$ can be calculated in the following two cases.
	\begin{itemize}
	\item[$\bullet$]When $n-s$ is even, then we have 
	\begin{align*}
	N_j(f^*)&=\sum\limits_{y\in\mathbb{F}_p^{\times}}\epsilon_0p^{\frac{n-s}{2}-1}\xi_p^{y(j_0-j)}+p^{n-s-1}\\
	&=\begin{cases}
	p^{n-s-1}+\epsilon_0(p-1)p^{\frac{n-s}{2}-1}, &\text{if}\ j=j_0;\\
	p^{n-s-1}-\epsilon_0p^{\frac{n-s}{2}-1},&\text{if}\ j\ne j_0.
	\end{cases}
	\end{align*}
	\item[$\bullet$] When $n-s$ is odd, then we have 
	\begin{align*}
	N_{j_0+j}(f^*)&=\sum\limits_{y\in\mathbb{F}_p^{\times}}\epsilon_0p^{\frac{n-s-3}{2}}\sqrt{p^*}\eta(y)\xi_p^{y(-j)}+p^{n-s-1}\\
	&=\begin{cases}
	p^{n-s-1},&\text{if}\ j=0;\\
	p^{n-s-1}+\epsilon_0\eta(j)p^{\frac{n-s-1}{2}},&\text{if}\ 1\le j\le p-1.
	\end{cases}
	\end{align*}
	\end{itemize}

The proof is now completed.\qed
	\end{proof}
	From the results of Theorem 4.2 in \cite{Ozbudak1}, we have the following lemma.
	
	\begin{lemma} Let $f(x):\mathbb{F}_{p}^n\longrightarrow\mathbb{F}_p$ be a  non-weakly regular $s$-plateaued function whose dual $f^*(x)$ is bent relative to $\mathrm{Supp}(\widehat{\chi_f})$ and $f(0)=j_0$, $\#B_{+}(f)=k\ (k\ne 0\ \text{and}\ k\ne p^{n-s})$. For any $j \in \mathbb{F}_p$, define $c_j(f^*)=\#\{x\in B_+(f):\ f^*(x)=j\}$, $d_j(f^*)=\#\{x\in B_-(f):\ f^*(x)=j\}$ and $e_j(f^*)=c_j(f^*)-d_j(f^*)$.
		\begin{itemize}
			\item[$\bullet$]
			When n+s is even, we have \\
			$c_{j_0}(f^*)=\begin{cases}\frac{k}{p} +(p-1)p^{\frac{n-s}{2}-1},\hspace{1.7cm}         &\text{if}\ 0\in B_+(f^*),\\
			\frac{k}{p},\       &\text{if}\ 0\in B_-(f^*);
			\end{cases}$\\ 
			$d_{j_0}(f^*)=\begin{cases}p^{n-s-1}-\frac{k}{p},\ &\text{if}\ 0\in B_+(f^*),\\
			p^{n-s-1}-(p-1)p^{\frac{n-s}{2}-1}-\frac{k}{p},\ &\text{if}\ 0\in B_-(f^*);\end{cases}$\\
			and for $j\ne j_0\in \mathbb{F}_p$\\
			$c_j(f^*)\ =\begin{cases}
			\frac{k}{p}-p^{\frac{n-s}{2}-1},\hspace{2.8cm} &\text{if}\ 0\in B_+(f^*),\\
			\frac{k}{p},\ &\text{if}\ 0\in B_-(f^*);\end{cases}$\\
			$d_j(f^*)\ =\begin{cases}
			p^{n-s-1}-\frac{k}{p},\hspace{2.8cm} &\text{if}\ 0\in B_+(f^*),\\
			p^{n-s-1}+p^{\frac{n-s}{2}-1}-\frac{k}{p},\ &\text{if}\ 0\in B_-(f^*).\end{cases}$\\
			\item[$\bullet$]When $p\equiv 1$ $\mathrm{(mod\  4)}$ and n+s is odd, we have \\
			$c_{j_0}(f^*)=\frac{k}{p};$\\
			$d_{j_0}(f^*)=p^{n-s-1}-\frac{k}{p};$\\
			and for $1\le j\le p-1$\\
			$c_{j_0+j}(f^*)\ =\begin{cases}
			\frac{k}{p}+\eta(j)p^{\frac{n-s-1}{2}},\hspace{2.3cm} &\text{if}\ 0\in B_+(f^*),\\
			\frac{k}{p},\ &\text{if}\ 0\in B_-(f^*);\end{cases}$\\
			$d_{j_0+j}(f^*)\ =\begin{cases} p^{n-s-1}-\frac{k}{p},\hspace{2.8cm} &\text{if}\ 0\in B_+(f^*),\\
			p^{n-s-1}-\eta(j)p^{\frac{n-s-1}{2}}-\frac{k}{p},\ &\text{if}\ 0\in B_-(f^*).\end{cases}$\\
			\item[$\bullet$]When $p\equiv 3$ $\mathrm{(mod\  4)}$ and n+s is odd, we have\\
			$c_{j_0}(f^*)=\frac{k}{p};$\\
			$d_{j_0}(f^*)=p^{n-s-1}-\frac{k}{p};$\\
			and for $1\le j\le p-1$\\
			$c_{j_0+j}(f^*)\ =\begin{cases}
			\frac{k}{p},\ &\text{if}\ 0\in B_+(f^*),\\
			\frac{k}{p}-\eta(j)p^{\frac{n-s-1}{2}},\hspace{2.3cm} &\text{if}\ 0\in B_-(f^*);\end{cases}$\\
			$d_{j_0+j}(f^*)\ =\begin{cases} p^{n-s-1}+\eta(j)p^{\frac{n-s-1}{2}}-\frac{k}{p},\hspace{0.8cm} &\text{if}\ 0\in B_+(f^*),\\
			p^{n-s-1}-\frac{k}{p},\ &\text{if}\ 0\in B_-(f^*).\end{cases}$\\	
		\end{itemize}
	\end{lemma}
\begin{proof}
	For any $j \in \mathbb{F}_p$, according to the definitions of $c_j(f^*)$ and  $d_j(f^*)$, we have that $c_j(f^*)+d_j(f^*)=N_j(f^*)$.
	From the Theorem 4.2 in \cite{Ozbudak1}, we know that when $n+s$ is even, $e_{j_0}(f^*)=\frac{2k}{p}-p^{\frac{n-s}{2}-1}(p^{\frac{n-s}{2}}-p+1)$ and $e_j(f^*)=\frac{2k}{p}-p^{\frac{n-s}{2}-1}(p^{\frac{n-s}{2}}+1)$ for $j\ne j_0\in\mathbb{F}_p$; when $p\equiv1\ \mathrm{(mod\  4)}$ and $n+s$ is odd, $e_{j_0}(f^*)=\frac{2k}{p}-p^{n-s-1}$ and $e_{j_0+j}(f^*)=\frac{2k}{p}-p^{n-s-1}+\eta
	(j)p^{\frac{n-s-1}{2}}$ for $1\le j\le p-1$; when  $p\equiv3\ \mathrm{(mod\  4)}$ and $n+s$ is odd, $e_{j_0}(f^*)=\frac{2k}{p}-p^{n-s-1}$ and $e_{j_0+j}(f^*)=\frac{2k}{p}-p^{n-s-1}-\eta
	(j)p^{\frac{n-s-1}{2}}$ for $1\le j\le p-1$. By Lemma 3, the proof can be easily completed.\qed
	\end{proof}
	\subsection{Linear codes}
	\quad For a vector $x=(x_1,\ x_2,\cdots, x_n)\in\mathbb{F}_p^n$, the $\mathit{support}$ of $x$ is defined as $\mathrm{supp}(x)=\{j:1\leq j \leq n, x_j\ne 0\}$ and the $\mathit{Hamming\  weight}$ of $x$ denoted by $\mathrm{wt}(x)$ is the size of $\mathrm{supp}(x)$. For two vectors $x,\ y$ $\in\mathbb{F}_p^n$, the $\mathit{Hamming \ distance}$ $\mathrm{d}(x,y)$ between them is defined to be the number of coordinates in which $x$ and $y$ differ. Let $\mathcal{C}$ be an $[n,\ k]$ linear code, that is,  $\mathcal{C}$ is a $k$-dimensional subspace of $\mathbb{F}_p^n$. An element of $\mathcal{C}$ is said to be a $\mathit{codeword}$. The $\mathit{minimum\ distance}$ of $\mathcal{C}$ is defined as the minimal Hamming distance between two distinct codewords.  If the minimum distance of $\mathcal{C}$  is $d$, we call $\mathcal{C}$ an $[n,\ k,\ d]$ linear code. Let $T$ be a set of $t$ coordinate positions in $\mathcal{C}$. The $\mathit{punctured\ code}$ of $\mathcal{C}$, defined by $\widetilde{\mathcal{C}}$, can be obtained by puncturing $\mathcal{C}$ on $T$. The linear code defined by $\mathcal{C}^{\bot}=\{x\in\mathbb{F}^{n}_{p}:\ x\cdot y=0\ \text{for all}\  y\in \mathcal{C}\}$ is called the $\mathit{dual\ code}$ of $\mathcal{C}$. Note that the dual code $\mathcal{C}^{\bot}$ is an $[n,\ n-k,\ d^{\bot}]$ linear code, where $d^{\bot}$ is the minimum distance of $\mathcal{C}^{\bot}$. For $0\le j\le n$, let $A_j$ be the number of codewords with Hamming weight $j$ in $\mathcal{C}$. Then $(1,\ A_1,\cdots,\ A_n)$ is called the $\mathit{weight\  distribution}$ of $\mathcal{C}$ and $1+A_1z+\cdots+A_nz^n$ is the $\mathit{weight\  enumerator}$ of $\mathcal{C}$. The code $\mathcal{C}$ is said to be a $t$-weight code if the number of nonzero $A_j$ ($1\le j\le n$) is equal to $t$. For the dual code $\mathcal{C}^{\bot}$ of $\mathcal{C}$, let $A_j^{\bot}$ be the number of codewords with Hamming weight $j$ in $\mathcal{C}^{\bot}$, where  $0\le j\le n$. The first five Pless power moments are given as follows \cite{Huffman}.
	\begin{flalign*}
	\sum_{j=0}^{n}A_j&=p^k,&&\\
	\sum_{j=0}^{n}jA_j&=p^{k-1}(pn-n-A_1^{\bot}),&&\\
	\sum_{j=0}^{n}j^2A_j&=p^{k-2}((p-1)n(pn-n+1)-(2pn-p-2n+2)A_1^{\bot}+2A_2^{\bot}),&&\\
	\sum_{j=0}^{n}j^3A_j&=p^{k-3}((p-1)n(p^2n^2-2pn^2+3pn-p+n^2-3n+2)-(3p^2n^2&&\\&-3p^2n-6pn^2+12pn+p^2-6p+3n^2-9n+6)A_1^{\bot}+6(pn-p&&\\&-n+2)A_2^{\bot}-6A_3^{\bot}),&&\\
	\sum_{j=0}^{n}j^4A_j&=p^{k-4}((p-1)n(p^3n^3-3p^2n^3+6p^2n^2-4p^2n+p^2+3pn^3-12pn^2&&\\
	&+15pn-6p-n^3+6n^2-11n+6)-(4p^3n^3-6p^3n^2+4p^3n-p^3
	\end{flalign*}
	\begin{flalign*}
	\hspace{1cm}&-12p^2n^3+36p^2n^2-38p^2n+14p^2+12pn^3-54pn^2+78pn-36p&&\\
	&-4n^3+24n^2-44n+24)A_1^{\bot}+(12p^2n^2-24p^2n+14p^2-24pn^2&&\\
	&+84pn-72p+12n^2-60n+72)A_2^{\bot}-(24pn-36p-24n+72)A_3^{\bot}&&\\&+24A_4^{\bot}).
		\end{flalign*}
	
	In coding theory, it is desirable to construct a linear code with $k/n$ and $d$ being as large as possible. However, there is a tradeoff among $n,\ k$, and $d$. In the following, we introduce two bounds on linear codes.
	\begin{lemma}[Singleton Bound]\cite{Huffman} Let $\mathcal{C}$ be an $[n,\ k,\ d]$ linear code over $\mathbb{F}_p$. Then
		\[d\le n-k+1.\]
	\end{lemma}
	\begin{lemma}
		[Sphere Packing Bound]\cite{Huffman}
		Let $\mathcal{C}$ be an $[n,\ k,\ d]$ linear code over $\mathbb{F}_p$. Then 
		\[p^{n-k}\ge\sum_{j=0}^{\lfloor\frac{d-1}{2}\rfloor}\binom{n}{j}(p-1)^j.\]
		
	\end{lemma}
	
	A linear code  is called $\mathit{maximum\ distance\ separable}$, abbreviated MDS if it achieves the Singleton bound. A linear code with parameters $[n,\ k,\ n-k]$ is said to be $\mathit{almost\ maximum}$  $\mathit{
		distance\ separable}$(AMDS).
	For a $p$-ary $[n,\ k,\ d]$ linear code    $\mathcal{C}$, it is said to be (distance) $\mathit{optimal}$ if there does not exist an $[n,\ k,\ d']$ linear code such that $d'>d$ and is called $\mathit{almost\ optimal}$ if there exists an $[n,\ k,\ d+1]$ optimal code.
	
	For codewords $x,\ y$ of linear code $\mathcal{C}$, we say that $x$ covers $y$ if $\mathrm{supp}(y)\subseteq\mathrm{supp}(x)$. A nonzero codeword $x$ of $\mathcal{C}$ is called $\mathit{minimal}$ if $x$  covers only codewords $jx$ with $j\in \mathbb{F}_p$. A linear code $\mathcal{C}$ is said to be $\mathit{minimal}$ if every nonzero codeword of $\mathcal{C}$ is minimal. The following lemma presents a sufficient condition for a linear code to be minimal.
	\begin{lemma}[Ashikhmin-Barg]\cite{Ashikhmin}
		Let $\mathcal{C}$ be a linear code over $\mathbb{F}_p$, if 
		\[\frac{p-1}{p}<\frac{wt_{\mathrm{min}}}{wt_{\mathrm{max}}},\]
		then $\mathcal{C}$ is minimal, where $wt_{\mathrm{min}}$ and $wt_{\mathrm{max}}$ are the minimum and maximum Hamming weights of nonzero codewords in $\mathcal{C}$, respectively.
	\end{lemma}
	
	\subsection{Secret Sharing Schemes}
	\quad Secret sharing, introduced in 1979 by Blakley et al. \cite{Blakley} and Shamir \cite{Shamir} independently, is an interesting topic of cryptography and has a number of applications in real-word security systems. A secret sharing scheme consists of a dealer, a group $\mathcal{P}=\{P_1,\ P_2,\ \cdots,P_{n-1}\}$ of $n-1$ participants, a secret space $\mathcal{S}$, $n-1$ share spaces $\mathcal{S}_1,\ \mathcal{S}_2,\ \cdots,\mathcal{S}_{n-1}$, a share computing procedure and a secret recovering procedure. The dealer randomly chooses an element $s\in \mathcal{S}$ as the secret, then computes $n-1$ shares $s_j\in \mathcal{S}_j$ and finally distributes $s_j$ to participant $\mathcal{P}_j$ for $1\le j\le n-1$. If a set of participants is able to determine $s$, then it is called an $\mathit{access\  set}$. The set of all access sets is called the $\mathit{access\ structure}$. An access set is minimal if any proper subset of it is not an access set. The set of all minimal access sets is called the $\mathit{minimal\ access\ structure}.$
	
	Now, we introduce a construction of secret sharing schemes from linear codes given by Massey \cite{Massey,Massey1}.
	
	Let $\mathcal{C}$ be an $[n,\ k,\ d]$ linear code over $\mathbb{F}_p$ with generator matrix $G=[g_0,\ g_1,\cdots,\ g_{n-1}]$ with $g_i\in \mathbb{F}_p^k$ for $0\le i\le n-1$. The secret space $\mathcal{S}$ of the secret sharing scheme based on $\mathcal{C}$ is $\mathbb{F}_p$ with the uniform probability distribution. The participant set $\mathcal{P}=\{P_1,\ P_2,\ \cdots,P_{n-1}\}$. The secret sharing procedure goes as follows:
	\begin{itemize}
		\item [$\bullet$] The dealer chooses randomly $u=(u_0,\ u_1,\cdots,\ u_{k-1})$ such that $s=u\cdot g_0$.
		\item [$\bullet$] The dealer then computes $t=(t_0,\ t_1,\cdots,\ t_{n-1})=uG$.
		\item [$\bullet$] The dealer finally gives $t_j$ to participant $P_j$ as the share for $j=1,2,\dots,n-1$.
	\end{itemize}
	Note that a set of shares of $\{t_{j_1},\ t_{j_2},\cdots,\ t_{j_m}\}$ can recover the secret if and only if $g_0$ is a linear combination of $g_{j_1},\ g_{j_2},\cdots,\ g_{j_m}$. 
	Assume that $g_0=\sum\limits_{k=1}^mx_kg_{j_k}$, then the secret $s$ is recovered by computing $s=\sum\limits_{k=1}^mx_kt_{j_k}$.
	
	Clearly, the assess structure of the above secret sharing scheme is 
	\begin{align*}
	\mathcal{T}=\{\{P_{j_1},\ P_{j_2},\cdots,\ P_{j_m}\}:&\ g_0\  \text{is\  a\  linear\  combination\  of}\  g_{j_1},\ g_{j_2},\cdots,\ g_{j_m},\\
	&\text{where}\ 1\le j_1<\cdots<j_m\le n-1\}.
	\end{align*}In \cite{Massey,Massey1}, Massey proved that the assess structure can also be given by 
	\[\mathcal{T}=\{\{P_j:\ j\in\mathrm{supp}((c_1,\cdots,c_{n-1}))\}: (1,\ c_1,\cdots,c_{n-1})\in \mathcal{C}^{\bot}
	\}.\]
	
	The following lemma describes the access structure of the secret sharing scheme based on the dual code of a minimal linear code.
	\begin{lemma}\cite{Yuan}
		Let $\mathcal{C}$ be a minimal  $[n,\ k,\ d]$ linear code over $\mathbb{F}_p$ with generator matrix $G=[{g}_0,\ {g}_1,\cdots,\ {g}_{n-1}]$, then in the secret sharing scheme based on $\mathcal{C}^{\bot}$ with respect to the above construction, there are altogether $p^{k-1}$ minimal access sets.
		\begin{itemize}
			\item [$\bullet$] When $d^{\bot}=2$, the access structure is as follows.\\
			If ${g}_j$ is a scalar multiple of ${g}_0$, $1\le j\le n-1$, then participant $P_j$ must be in every minimal access set.\\
			If ${g}_j$ is not a scalar multiple of ${g}_0$, $1\le j\le n-1$, then participant $P_j$ mus
			t be in $(p-1)p^{k-2}$ out of $p^{k-1}$  minimal access sets.
			\item [$\bullet$]When $d^{\bot}\ge 3$, for any fixed $1\le t\le\min\{k-1,d^{\bot}-2\}$, every set of $t$ participants is involved in $(p-1)^tp^{k-(t+1)}$ out of $p^{k-1}$ minimal access sets.
		\end{itemize} 
	\end{lemma}
	\section{Linear codes from non-weakly regular plateaued functions based on the first generic construction}
	\quad In this section, motivated by \cite{Mesnager2}, we construct some linear codes with few weights from non-weakly regular plateaued functions and determine their weight distributions. 
	
	Let $f(x):\mathbb{F}_{p}^n\longrightarrow\mathbb{F}_p$ be a non-weakly regular $s$-plateaued function satisfying $f(0)=0$. For any $a\in\mathrm{supp}(\widehat{\chi_f})$, define $\epsilon_a=1\ (\text{respectively}\ \epsilon_a=-1)$ if $a\in B_+(f)\ (\text{respectively}\ a\in B_-(f)).$ In the following, we consider the linear code over $\mathbb{F}_p$ defined by
	\begin{equation}
	\mathcal{C}_f=\{c(a,\ b)=(af(x)-b\cdot x)_{x\in\mathbb{F}_{p}^{n}\setminus\{0\}}:\ a\in\mathbb{F}_p,\ b\in \mathbb{F}_{p}^n\}.
	\end{equation}
	In the following theorem, we give the weights of codewords of $\mathcal{C}_f$ when $n+s$ is even.
	\begin{theorem}
		Let $n+s$ be an even integer with $0\le s \le n-2$, and $f(x):\mathbb{F}_{p}^n\longrightarrow\mathbb{F}_p$ be a non-weakly regular $s$-plateaued function satisfying $f(0)=0$. Then, $\mathcal{C}_f$ defined by (1) is a $[p^n-1,\ n+1]$ linear code, and the weights of codewords are given as follows.
		\[wt({c}(a,\ b))=
		\begin{cases}
		0,& \text{if}\ a=0,\ b=0;\\
		(p-1)p^{n-1},& \text{if}\  a=0,\ b\ne0\ \text{or}\ a\ne0,\\& a^{-1}b\notin\mathrm{Supp}(\widehat{\chi_f});\\
		(p-1)(p^{n-1}-p^{\frac{n+s}{2}-1}),& \text{if}\ a\ne0,\ a^{-1}b\in B_+(f),\ f^*(a^{-1}b)=0;\\
		(p-1)(p^{n-1}+p^{\frac{n+s}{2}-1}),& \text{if}\ a\ne0,\ a^{-1}b\in B_-(f),\ f^*(a^{-1}b)=0;\\
		(p-1)p^{n-1}+p^{\frac{n+s}{2}-1},& \text{if}\ a\ne0,\ a^{-1}b\in B_+(f),\ f^*(a^{-1}b)\ne0;\\
		(p-1)p^{n-1}-p^{\frac{n+s}{2}-1},& \text{if}\ a\ne0,\ a^{-1}b\in B_-(f),\ f^*(a^{-1}b)\ne0.
		\end{cases}\]
	\end{theorem}
	\begin{proof}
		Since $f(0)=0$, then $af(0)-b \cdot 0=0$, so we have $wt({c}(a,\ b))=p^n-\#\{x\in\mathbb{F}_{p}^n:\ af(x)-b\cdot x=0\}$. Let $N=\#\{x\in\mathbb{F}_{p}^n:\ af(x)-b\cdot x=0\}$, then
		\begin{align*}
		N&=p^{-1}\sum\limits_{x\in\mathbb{F}_{p}^n}\sum\limits_{y\in\mathbb{F}_p}\xi_p^{y(af(x)-b \cdot x)}\\
		&=p^{-1}\sum\limits_{y\in\mathbb{F}_p^{\times}}\sigma_y(\sum\limits_{x\in\mathbb{F}_{p}^n}\xi_p^{af(x)-b \cdot x})+p^{n-1}.
		\end{align*}
		Next, we discuss in two cases.
		\begin{itemize}
			\item[$\bullet$] $a=0$\\
			If $b=0$, then $N=p^n$, so $wt({c}(a,\ b))=0.$
			
			If $b\ne0$, then $N=p^{n-1}$, so $wt({c}(a,\ b))=(p-1)p^{n-1}.$
			\item[$\bullet$] $a\ne0$\\
			By the definition of Walsh transform, we can get
			\begin{align*}
			N&=p^{-1}\sum\limits_{y\in\mathbb{F}_p^{\times}}\sigma_y\sigma_a(\sum\limits_{x\in\mathbb{F}_{p}^n}\xi_p^{f(x)-a^{-1}b\cdot x})+p^{n-1}\\
			&=p^{-1}\sum\limits_{y\in\mathbb{F}_p^{\times}}\sigma_y\sigma_a(\widehat{\chi_f}(a^{-1}b))+p^{n-1}.
			\end{align*}
			When $a^{-1}b\notin\mathrm{Supp}(\widehat{\chi_f})$, we can get $N=p^{n-1}$, then $wt({c}(a,\ b))=(p-1)p^{n-1}.$ When $a^{-1}b\in\mathrm{Supp}(\widehat{\chi_f})$, we can obtain 
			\begin{align*}
			N&=p^{-1}\sum\limits_{y\in\mathbb{F}_p^{\times}}\sigma_y\sigma_a(\epsilon_{a^{-1}b}p^{\frac{n+s}{2}}\xi_p^{f^*(a^{-1}b)})+p^{n-1}\\
			&=p^{-1}\sum\limits_{y\in\mathbb{F}_p^{\times}}\epsilon_{a^{-1}b}p^{\frac{n+s}{2}}\xi_p^{yaf^*(a^{-1}b)}+p^{n-1}.
			\end{align*}
			If $f^*(a^{-1}b)=0$, then $N=\epsilon_{a^{-1}b}(p-1)p^{\frac{n+s}{2}-1}+p^{n-1}$,  so $wt({c}(a,\ b))=(p-1)(p^{n-1}-\epsilon_{a^{-1}b}p^{\frac{n+s}{2}-1}).$
			If $f^*(a^{-1}b)\ne0$, by Lemma 1, we have $N=p^{n-1}-\epsilon_{a^{-1}b}p^{\frac{n+s}{2}-1}$,  so $wt({c}(a,\ b))=(p-1)p^{n-1}+\epsilon_{a^{-1}b}p^{\frac{n+s}{2}-1}.$
		\end{itemize}
		
		From the above discussion, we can get the weights of codewords of $\mathcal{C}_f$. Note that $wt({c}(a,\ b))=0$ if and only $a=0$ and $b=0$, thus the number of codewords of $\mathcal{C}_f$ is $p^{n+1}$, and the dimension  is $n+1$.\qed
	\end{proof}
	\begin{remark}
		By Lemma 7, we have that the linear code $\mathcal{C}_f$ in Theorem 1 is minimal for $0\le s\le n-4$.
	\end{remark}
	\begin{table}
		\caption{The weight distribution of $\mathcal{C}_f$ in Proposition 1 when $n+s$ is even}
		\begin{tabular}{|l|l|l|}
			\hline
			Weight  & Multiplicity  ($0\in B_+(f^*)$)&Multiplicity ($0\in B_-(f^*)$)\\ \hline
			0 & 1&1 \\ \hline
			$(p-1)p^{n-1}$&$p^{n+1}-(p-1)p^{n-s}-1$&$p^{n+1}-(p-1)p^{n-s}-1$\\ \hline
			$(p-1)(p^{n-1}-p^{\frac{n+s}{2}-1})$&$(p-1)(\frac{k}{p}+(p-1)p^{\frac{n-s}{2}-1})$&$(p-1)\frac{k}{p}$\\ \hline
			$(p-1)(p^{n-1}+p^{\frac{n+s}{2}-1})$&$(p-1)(p^{n-s-1}-\frac{k}{p})$&$(p-1)(p^{n-s-1}-\frac{k}{p}$\\ 
			&&$-(p-1)p^{\frac{n-s}{2}-1})$\\\hline
			$(p-1)p^{n-1}+p^{\frac{n+s}{2}-1}$&$(p-1)^2(\frac{k}{p}-p^{\frac{n-s}{2}-1})$&$(p-1)^2\ \frac{k}{p}$\\ \hline
			$(p-1)p^{n-1}-p^{\frac{n+s}{2}-1}$&$(p-1)^2(p^{n-s-1}-\frac{k}{p})$&$(p-1)^2(p^{n-s-1}-\frac{k}{p}$\\
			&&$
			
			+p^{\frac{n-s}{2}-1})$\\\hline
			\end{tabular}
			\end{table}
		
	For a subclass of non-weakly regular $s$-plateaued functions, we determine the weight distributions of linear codes given by Theorem 1 in the following proposition.
	\begin{proposition}
		Let $n+s$ be an even integer with $0\le s\le n-2$, $f(x):\mathbb{F}_{p}^n\longrightarrow\mathbb{F}_p$ be a non-weakly regular $s$-plateaued function satisfying $f(0)=0$, $\#B_+(f)=k\ (k\ne 0\ \text{and}\ k\ne p^{n-s})$, and the dual $f^*(x)$ of $f(x)$ be bent relative to $\mathrm{Supp}(\widehat{\chi_f})$. Then  the weight distribution of $\mathcal{C}_f$ is given by Table 1. The dual code $\mathcal{C}_f^{\bot}$ of $\mathcal{C}_f$ is a $[p^n-1,\ p^n-n-2]$ linear code, and the minimum distance of $\mathcal{C}_f^{\bot}$, denoted by $d^{\bot}$, is $3$ if $s=0$, $n=2$ and $0\in B_-(f^*)$, otherwise $d^{\bot}=2$.	
	\end{proposition}
	\begin{proof}
		By Lemma 4, we can calculate the weight distribution of $\mathcal{C}_f$ in the following six cases.
		\begin{itemize}
			\item[$\bullet$] Obviously $A_0=1$.
			\item[$\bullet$]$wt({c}(a,\ b))=(p-1)p^{n-1}$, that is, $a=0,\ b\ne0\ \text{or}\ a\ne0,\ a^{-1}b\notin\mathrm{Supp}(\widehat{\chi_f})$. \\The number of codewords satisfying $a=0,\ b\ne0$ is $p^n-1$. Since $\#\mathrm{Supp}(\widehat{\chi_f})\\=p^{n-s}$,  and for any fixed $a\in\mathbb{F}_p^{\times}$, $\#\{b:a^{-1}b\notin\mathrm{Supp}(\widehat{\chi_f})\}=\#\{b:b\notin\mathrm{Supp}(\widehat{\chi_f})\}=p^n-p^{n-s}$, then we have that the number of codewords satisfying $a\ne0,\ a^{-1}b\notin\mathrm{Supp}(\widehat{\chi_f})$ is $(p-1)(p^n-p^{n-s})$. Therefore, we obtain  $A_{(p-1)p^{n-1}}=p^{n+1}-(p-1)p^{n-s}-1$.
			\item[$\bullet$]$wt({c}(a,\ b))=(p-1)(p^{n-1}-p^{\frac{n+s}{2}-1})$, that is, $a\ne0,\ a^{-1}b\in B_+(f),\ \text{and}\\  f^*(a^{-1}b)=0.$ \\For any fixed $a\in\mathbb{F}_p^{\times}$, $\#\{b:a^{-1}b\in B_+(f),f^*(a^{-1}b)=0\}=\#\{b:b\in B_+(f),f^*(b)=0\}$, so by Lemma 4, we have : if $0\in B_+(f^*)$, then $A_{(p-1)(p^{n-1}-p^{\frac{n+s}{2}-1})}=(p-1)(\frac{k}{p}+(p-1)p^{\frac{n-s}{2}-1})$; if $0\in B_-(f^*)$, then $A_{(p-1)(p^{n-1}-p^{\frac{n+s}{2}-1})}=(p-1)\frac{k}{p}.$
			\item[$\bullet$]$wt({c}(a,\ b))=(p-1)(p^{n-1}+p^{\frac{n+s}{2}-1})$, that is, $a\ne0,\ a^{-1}b\in B_-(f),\ \text{and}\\  f^*(a^{-1}b)=0.$ \\For any fixed $a\in\mathbb{F}_p^{\times}$, $\#\{b:a^{-1}b\in B_-(f),f^*(a^{-1}b)=0\}=\#\{b:b\in B_-(f),f^*(b)=0\}$, so by Lemma 4, we have : if $0\in B_+(f^*)$,  then $A_{(p-1)(p^{n-1}+p^{\frac{n+s}{2}-1})}=(p-1)(p^{n-s-1}-\frac{k}{p})$; if $0\in B_-(f^*)$, then\\ $A_{(p-1)(p^{n-1}+p^{\frac{n+s}{2}-1})}=(p-1)(p^{n-s-1}-\frac{k}{p}-(p-1)p^{\frac{n-s}{2}-1})$.
			\item[$\bullet$]$wt({c}(a,\ b))=(p-1)p^{n-1}+p^{\frac{n+s}{2}-1}$, that is, $a\ne0,\ a^{-1}b\in B_+(f),\ \text{and}\\  f^*(a^{-1}b)\ne0$. \\For any fixed $a\in\mathbb{F}_p^{\times}$, $\#\{b:a^{-1}b\in B_+(f),f^*(a^{-1}b)\ne0\}=\#\{b:b\in B_+(f),f^*(b)\ne0\}$, so by Lemma 4, we have : if $0\in B_+(f^*)$,  then $A_{(p-1)p^{n-1}+p^{\frac{n+s}{2}-1}}=(p-1)^2(\frac{k}{p}-p^{\frac{n-s}{2}-1})$; if $0\in B_-(f^*)$, then\\ $A_{(p-1)p^{n-1}+p^{\frac{n+s}{2}-1}}=(p-1)^2\ \frac{k}{p}.$
			\item[$\bullet$]$wt({c}(a,\ b))=(p-1)p^{n-1}-p^{\frac{n+s}{2}-1}$, that is, $a\ne0,\ a^{-1}b\in B_-(f),\ \text{and}\\  f^*(a^{-1}b)\ne0$. \\For any fixed $a\in\mathbb{F}_p^{\times}$, $\#\{b:a^{-1}b\in B_-(f), f^*(a^{-1}b)\ne0\}=\#\{b:b\in B_-(f), f^*(b)\ne0\}$, so by Lemma 4, we have : if $0\in B_+(f^*)$,  then $A_{(p-1)p^{n-1}-p^{\frac{n+s}{2}-1}}=(p-1)^2(p^{n-s-1}-\frac{k}{p})$; if $0\in B_-(f^*)$, then\\ $A_{(p-1)p^{n-1}-p^{\frac{n+s}{2}-1}}=(p-1)^2(p^{n-s-1}-\frac{k}{p}+p^{\frac{n-s}{2}-1}).$
		\end{itemize}
		
		By the weight distribution of $\mathcal{C}_f$ and the first four Pless power moments, we have that when $0\in B_+(f^*)$, $A_1^{\bot}=0$ and $A_2^{\bot}=\frac{(p^2-3p+2)(p^n+p^{\frac{n+s}{2}+1}-p^{\frac{n+s}{2}}-p)}{2p}\\>0$, so the minimum distance of $\mathcal{C}_f^{\bot}$ is 2. When $0\in B_-(f^*)$, if $s=0$ and $n=2$, then $A_1^{\bot}=A_2^{\bot}=0$ and $A_3^{\bot}=\frac{(p-1)^2(p-2)(p^3+p^2-3p+2k-3)}{6}>0$, so the minimum distance of $\mathcal{C}_f^{\bot}$ is 3; if $n>2$, then $A_1^{\bot}=0$ and $A_2^{\bot}=\frac{(p^2-3p+2)(p^n+p^{\frac{n+s}{2}}-p-p^{\frac{n+s}{2}+1})}{2p}>0$, so the minimum distance of $\mathcal{C}_f^{\bot}$ is 2.
		Obviously, the dimension of $\mathcal{C}_f^{\bot}$ is $ p^n-n-2$. The proof is now completed.\qed
	\end{proof}
	
			\begin{remark}
			When $s=0,\ n=2$ and $0\in B_-(f^*)$, $\mathcal{C}_f^{\bot}$ is a linear code with parameters $[p^2-1,\ p^2-4,\ 3]$, which is an AMDS code. Specially, according to the Code Table at http://www.codetables.de/, we have that when $p=3,5,7$, the linear codes $\mathcal{C}_f^{\bot}$ with parameters $[8,\ 5,\ 3]$, $[24,\ 21,\ 3]$ and $[48,\ 45,\ 3]$ are optimal, respectively.
			\end{remark}
			
			Next, we verify Proposition 1 by Magma program for the following non-weakly regular plateaued functions.
			\begin{example}
				Consider $f(x): \mathbb{F}_{3}^5\longrightarrow\mathbb{F}_3,\  f(x_1,\ x_2,\ x_3,\ x_4,\ x_5)=2x_1^2x_4^2+2x_1^2+x_2^2+x_3x_4$,  which is a non-weakly regular $1$-plateaued function with $ 0\in B_+(f^*)$, $k=27$, and the dual $f^*(x)$ of $f(x)$ is bent relative to $\mathrm{Supp}(\widehat{\chi_f})$, then $\mathcal{C}_f$ constructed by (1) is a five-weight linear code with parameters $[242,\ 6,\ 144]$, weight enumerator $1+30z^{144}+72z^{153}+566z^{162}+24z^{171}+36z^{180}.$ The dual code $\mathcal{C}_f^{\bot}$ is a $[242,\ 236,\ 2]$ linear code, which is almost optimal according to the Code Table at http://www.codetables.de/.
			\end{example}
			\begin{example}
				Consider $f(x): \mathbb{F}_{3}^5\longrightarrow\mathbb{F}_3,\ f(x_1,\ x_2,\ x_3,\ x_4,\ x_5)=x_1^2x_4^2+x_1^2+x_2^2+x_3x_4$, which is a non-weakly regular $1$-plateaued function with $0\in B_-(f^*)$, $k=54$, and the dual $f^*(x)$ of $f(x)$ is bent relative to $\mathrm{Supp}(\widehat{\chi_f})$, then $\mathcal{C}_f$ constructed by (1) is a five-weight linear code with parameters $[242,\ 6,\ 144]$, weight enumerator $1+36z^{144}+48z^{153}+566z^{162}+72z^{171}+6z^{180}.$ The dual code $\mathcal{C}_f^{\bot}$ is a $[242,\ 236,\ 2]$ linear code, which is almost optimal according to the Code Table at http://www.codetables.de/.
			\end{example}
			
			According to Lemma 8, Remark 2 and Proposition 1, we have the following result on secret sharing scheme.
			\begin{corollary}
				Let $n+s$ be an even integer with $0\le s\le n-4$ and $\mathcal{C}_f$  be the linear code in Proposition 1 with generator matrix $G=[{g}_0,\ {g}_1,\ \cdots, {g}_{m-1}]$, where $m=p^n-1$. Then, in the secret sharing scheme based on $\mathcal{C}_f^{\bot}$ with $d^{\bot}=2$, the number of participants is $m-1$ and there are $p^n$ minimal access sets. Moreover, if ${g}_j$ is a scalar multiple of ${g}_0$, $1\le j\le m-1$, then participant $P_j$ must be in every minimal access set, or else, 
				participant $P_j$ must be in $(p-1)p^{n-1}$ out of $p^{n}$  minimal access sets.
			\end{corollary}

			When $n+s$ is odd, we give the weights of codewords of $\mathcal{C}_f$ in the following theorem.
			\begin{theorem}
				Let $n+s$ be an odd integer with $0\le s\le n-1$, and $f(x):\mathbb{F}_{p}^n\longrightarrow\mathbb{F}_p$ be a non-weakly $s$-plateaued function satisfying $f(0)=0$. Then, $\mathcal{C}_f$ defined by (1) is a $[p^n-1,\ n+1]$ linear code and the weights of codewords  are given as follows.
				\begin{itemize}
					\item[$\bullet$] When $p\equiv1\ (\mathrm{mod}\ 4)$, then
					\[wt({c}(a,\ b))=
					\begin{cases}
					0,& \text{if}\ a=0,\ b=0;\\
					& \text{if}\ a=0,\ b\ne0\ \text{or}\ a\ne0,\\(p-1)p^{n-1},& a^{-1}b\notin\mathrm{Supp}(\widehat{\chi_f})\ \text{or}\ a\ne0,\\& a^{-1}b\in\mathrm{Supp}(\widehat{\chi_f}),\ f^*(a^{-1}b)=0;\\
					&\text{if}\ a\ne0,\ a^{-1}b\in B_+(f),\\(p-1)p^{n-1}-p^{\frac{n+s-1}{2}},& f^*(a^{-1}b)\in SQ\ 
					\text{or}\ a\ne0,\\&a^{-1}b\in B_-(f),\ f^*(a^{-1}b)\in NSQ;\\
					&\text{if}\ a\ne0,\ a^{-1}b\in B_-(f),\\(p-1)p^{n-1}+p^{\frac{n+s-1}{2}},&\ f^*(a^{-1}b)\in SQ\ 
					\text{or}\ a\ne0,\\&a^{-1}b\in B_+(f),\ f^*(a^{-1}b)\in NSQ.
					\end{cases}
					\]
					\item[$\bullet$] When $p\equiv3\ (\mathrm{mod}\ 4)$, then 
					\[wt({c}(a,\ b))=
					\begin{cases}
					0,& \text{if}\ a=0,\ b=0;\\
					& \text{if}\ a=0,\ b\ne0\ \text{or}\ a\ne0,\\(p-1)p^{n-1},& a^{-1}b\notin\mathrm{Supp}(\widehat{\chi_f})\ \text{or}\ a\ne0,\\& a^{-1}b\in\mathrm{Supp}(\widehat{\chi_f}),\ f^*(a^{-1}b)=0;\\
					&\text{if}\ a\ne0,\ a^{-1}b\in B_+(f),\\(p-1)p^{n-1}-p^{\frac{n+s-1}{2}},&\ f^*(a^{-1}b)\in NSQ\ \text{or}\ a\ne0,\\&\ a^{-1}b\in B_-(f),\ f^*(a^{-1}b)\in SQ;\\
					&\text{if}\ a\ne0,\ a^{-1}b\in B_-(f),\\(p-1)p^{n-1}+p^{\frac{n+s-1}{2}},&\ f^*(a^{-1}b)\in NSQ\ \text{or}\ a\ne0,\\&a^{-1}b\in B_+(f),\ f^*(a^{-1}b)\in SQ.
					\end{cases}
					\]
				\end{itemize}
			\end{theorem}
			\begin{proof}
				Since $f(0)=0$, then $af(0)-b\cdot 0=0$, so we have $wt({c}(a,\ b))=p^n-\#\{x\in\mathbb{F}_{p}^n:\ af(x)-b\cdot x=0\}$. Let $N=\#\{x\in\mathbb{F}_{p}^n:\ af(x)-b\cdot x=0\}$, then 
				\begin{align*}
				N&=p^{-1}\sum\limits_{x\in\mathbb{F}_{p}^n}\sum\limits_{y\in\mathbb{F}_p}\xi_p^{y(af(x)-b\cdot x)}\\
				&=p^{-1}\sum\limits_{y\in\mathbb{F}_p^{\times}}\sum\limits_{x\in\mathbb{F}_{p}^n}\xi_p^{y(af(x)-b\cdot x)}+p^{n-1}\\
				&=p^{-1}\sum\limits_{y\in\mathbb{F}_p^{\times}}\sigma_y(\sum\limits_{x\in\mathbb{F}_{p}^n}\xi_p^{af(x)-b\cdot x})+p^{n-1}.
				\end{align*}
				Next, we discuss in two cases.
				\begin{itemize}
					\item[$\bullet$] $a=0$\\
					If $b=0$, then $N=p^n$, so $wt({c}(a,\ b))=0.$\\
					If $b\ne0$, then $N=p^{n-1}$,  so $wt({c}(a,\ b))=(p-1)p^{n-1}.$
					\item[$\bullet$] $a\ne0$\\
					By the definition of Walsh transform, we can get
					\begin{align*}
					N&=p^{-1}\sum\limits_{y\in\mathbb{F}_p^{\times}}\sigma_y\sigma_a(\sum\limits_{x\in\mathbb{F}_{p}^n}\xi_p^{f(x)-a^{-1}b\cdot x})+p^{n-1}\\
					&=p^{-1}\sum\limits_{y\in\mathbb{F}_p^{\times}}\sigma_y\sigma_a(\widehat{\chi_f}(a^{-1}b))+p^{n-1}.
					\end{align*}
					When $a^{-1}b\notin\mathrm{Supp}(\widehat{\chi_f})$, we can get $N=p^{n-1}$,  then $wt({c}(a,\ b))=\\(p-1)p^{n-1}.$ When $a^{-1}b\in\mathrm{Supp}(\widehat{\chi_f})$, we can obtain 
					\begin{align*}
					N&=p^{-1}\sum\limits_{y\in\mathbb{F}_p^{\times}}\sigma_y\sigma_a(\epsilon_{a^{-1}b}\sqrt{p^*}p^{\frac{n+s-1}{2}}\xi_p^{f^*(a^{-1}b)})+p^{n-1}\\
					&=p^{-1}\sum\limits_{y\in\mathbb{F}_p^{\times}}\epsilon_{a^{-1}b}\eta(ya)\sqrt{p^*}p^{\frac{n+s-1}{2}}\xi_p^{yaf^*(a^{-1}b)}+p^{n-1}.
					\end{align*}
					If $f^*(a^{-1}b)=0$, then by Lemma 1, we have $N=p^{n-1}$, so $wt({c}(a,\ b))$\\$=(p-1)p^{n-1}.$
					If $f^*(a^{-1}b)\ne0$, then by Lemma 1, we have $N=\epsilon_{a^{-1}b}$\\$\eta(f^*(a^{-1}b))p^*p^{\frac{n+s-1}{2}-1}+p^{n-1}$, so $wt({c}(a,\ b))=(p-1)p^{n-1}-\epsilon_{a^{-1}b}$\\$\eta(f^*(a^{-1}b))p^*p^{\frac{n+s-1}{2}-1}.$
				\end{itemize}
				
				From the above discussion, we can get the weights of codewords of $\mathcal{C}_f$. Note that $wt({c}(a,\ b))=0$ if and only if $a=0$ and $b=0$, thus the number of codewords of $\mathcal{C}_f$ is $p^{n+1}$ and the dimension is $n+1$. \qed
			\end{proof}
			
			\begin{remark} By Lemma 7, we have that the linear code $\mathcal{C}_f$ in Theorem 2 is minimal for $0\le s\le n-3$.
			\end{remark}
			\begin{table}
				
				\caption{The weight distribution of $\mathcal{C}_f$ in Proposition 2 when $n+s$ is odd}
				\begin{tabular}{|l|l|}
					\hline
					Weight & Multiplicity\\ \hline
					0&1\\ \hline
					$(p-1)p^{n-1}$&$p^{n+1}-(p-1)^2p^{n-s-1}-1$\\ \hline
					$(p-1)p^{n-1}-p^{\frac{n+s-1}{2}}$&$\frac{(p-1)^2}{2}(p^{\frac{n-s-1}{2}}+p^{n-s-1})$\\ \hline
					$(p-1)p^{n-1}+p^{\frac{n+s-1}{2}}$&$\frac{(p-1)^2}{2}(p^{n-s-1}-p^{\frac{n-s-1}{2}})$\\ \hline
				\end{tabular}
			\end{table}
		
			For a subclass of non-weakly regular $s$-plateaued functions, we determine the weight distributions of linear codes given by Theorem 2 in the following proposition.
			\begin{proposition}
				Let $n+s$ be an odd integer with $0\le s\le n-1$, $f(x):\mathbb{F}_{p}^n\longrightarrow\mathbb{F}_p$ be a non-weakly regular $s$-plateaued function satisfying $f(0)=0$, $\#B_+(f)=k\ (k\ne 0\ \text{and}\ k\ne p^{n-s})$, and the dual $f^*(x)$ of $f(x)$ be bent relative to $\mathrm{Supp}(\widehat{\chi_f})$. Then  the weight distribution of $\mathcal{C}_f$ is given by Table 2. The dual code $\mathcal{C}_f^{\bot}$ of $\mathcal{C}_f$ is a $[p^n-1,\ p^n-n-2]$ linear code, and the minimum distance of $\mathcal{C}_f^{\bot}$, denoted by $d^{\bot}$, is $3$ if $s=0$, $n=1$, and $p\ne 3$, or else $d^{\bot}=2$ if $n>1$.
			\end{proposition}
			\begin{proof}
				When $p\equiv1\ (\mathrm{mod}\ 4)$, by Lemma 4, we can calculate the weight distribution of $\mathcal{C}_f$ in the following four cases.
				\begin{itemize}
					\item[$\bullet$]Obviously $A_0=1$.
					\item[$\bullet$]$wt({c}(a,\ b))=(p-1)p^{n-1}$, that is, $a=0,\ b\ne0\ \text{or}\ a\ne0,\ a^{-1}b\notin\mathrm{Supp}(\widehat{\chi_f})\ \text{or}\ a\ne0,\  a^{-1}b\in\mathrm{Supp}(\widehat{\chi_f}),\ \text{and}\ f^*(a^{-1}b)=0$. \\The number of codewords satisfying $a=0,\ b\ne0$ is $p^n-1$. Since $\#\mathrm{Supp}(\widehat{\chi_f})\\=p^{n-s}$, and for any fixed $a\in\mathbb{F}_p^{\times}$, $\#\{b:a^{-1}b\notin\mathrm{Supp}(\widehat{\chi_f})\}=\#\{b:b\notin\mathrm{Supp}(\widehat{\chi_f})\}=p^n-p^{n-s}$, then we have that the number of codewords satisfying $a\ne0,\ a^{-1}b\notin\mathrm{Supp}(\widehat{\chi_f})$ is $(p-1)(p^n-p^{n-s})$. For any fixed $a\in\mathbb{F}_p^{\times}$, by Lemma 4, $\#\{b:a^{-1}b\in\mathrm{Supp}(\widehat{\chi_f}),f^*(a^{-1}b)=0\}=\#\{b:b\in\mathrm{Supp}(\widehat{\chi_f}),f^*(b)=0\}=p^{n-s-1}$, so the number of codewords satisfying $a\ne0,\ a^{-1}b\in\mathrm{Supp}(\widehat{\chi_f}),f^*(a^{-1}b)=0$ is $(p-1)p^{n-s-1}$. Hence, $A_{(p-1)p^{n-1}}=p^{n+1}-(p-1)^2p^{n-s-1}-1$.
					\item[$\bullet$]$wt({c}(a,\ b))=(p-1)p^{n-1}-p^{\frac{n+s-1}{2}}$, that is, $a\ne0,\ a^{-1}b\in B_+(f),\  f^*(a^{-1}b)\in SQ\ 
					\text{or}\ a\ne0,\ a^{-1}b\in B_-(f),\ f^*(a^{-1}b)\in NSQ$. \\For any fixed $a\in\mathbb{F}_p^{\times}$, $\#\{b:a^{-1}b\in B_+(f),f^*(a^{-1}b)\in SQ\}=\#\{b:b\in B_+(f),f^*(b)\in SQ\}$, and $\#\{b:a^{-1}b\in B_-(f),f^*(a^{-1}b)\in NSQ\}=\#\{b:b\in B_-(f),f^*(b)\in NSQ\}$. By Lemma 4, we have $A_{(p-1)p^{n-1}-p^{\frac{n+s-1}{2}}}=\frac{(p-1)^2}{2}(p^{\frac{n-s-1}{2}}+p^{n-s-1}).$
					\item[$\bullet$]$wt({c}(a,\ b))=(p-1)p^{n-1}+p^{\frac{n+s-1}{2}}$, that is, $a\ne0,\ a^{-1}b\in B_-(f),\  f^*(a^{-1}b)\in SQ\ 
					\text{or}\ a\ne0,\ a^{-1}b\in B_+(f),\ f^*(a^{-1}b)\in NSQ$. \\For any fixed $a\in\mathbb{F}_p^{\times}$, $\#\{b:a^{-1}b\in B_+(f),f^*(a^{-1}b)\in NSQ\}=\#\{b:b\in B_+(f),f^*(b)\in NSQ\}$, and $\#\{b:a^{-1}b\in B_-(f),f^*(a^{-1}b)\in SQ\}=\#\{b:b\in B_-(f),f^*(b)\in SQ\}$. By Lemma 4, we have $A_{(p-1)p^{n-1}+p^{\frac{n+s-1}{2}}}=\frac{(p-1)^2}{2}(p^{n-s-1}-p^{\frac{n-s-1}{2}}).$
				\end{itemize}
				
				When $p\equiv3\  (\mathrm{mod}\ 4)$, similar to the above discussion, the weight distribution of $\mathcal{C}_f$ can be obtained.
				
				By the weight distribution of $\mathcal{C}_f$ and the first four Pless power moments, we have if $s=0$, $n=1$, and $p\ne 3$, then $A_1^{\bot}=A_2^{\bot}=0$ and $A_3^{\bot}=\frac{(p-1)^2(p^2-5p+6)}{6}>0$, so the minimum distance of $\mathcal{C}_f^{\bot}$ is 3. If $n>1$, we have $A_1^{\bot}=0$ and $A_2^{\bot}=\frac{(p^n-p)(p^2-3p+2)}{2p}>0$, so the minimal distance of $\mathcal{C}_f^{\bot}$ is 2. Clearly, the dimension of $\mathcal{C}_f^{\bot}$ is $p^n-n-2$. The proof is now completed.\qed
			\end{proof}

			\begin{remark}
					(\romannumeral 1) When $n+s$ is odd, the weight distribution of $\mathcal{C}_f$ given in Proposition 2 is independent of $k$ and is the same as the weight distribution of the linear code from weakly regular $s$-plateaued functions in \cite{Mesnager2}.\\
				(\romannumeral 2)
				When $s=0,\ n=1$ and $p\ne3$, $\mathcal{C}_f$ is a linear code with parameters $[p-1,\ 2,\ p-2]$, which is an MDS code and $\mathcal{C}_f^{\bot}$ is a $[p-1,\ p-3,\ 3]$ linear code, which is also an MDS code.
			\end{remark}
			
			Next, we verify Proposition 2 by Magma program for the following non-weakly regular plateaued function.
			
			\begin{example}
				Consider $f(x):\ \mathbb{F}_{5}^4\longrightarrow\mathbb{F}_5,\  f(x_1,\ x_2,\ x_3,\ x_4)=x_1^2x_3^4+x_1^2+x_2x_3$, which is a non-weakly regular $1$-plateaued function, and the dual $f^*(x)$ of $f(x)$ is bent relative to $\mathrm{Supp}(\widehat{\chi_f})$, then $\mathcal{C}_f$ constructed by (1) is a three-weight linear code with parameters $[624,\ 5,\ 475]$ , weight enumerator $1+240z^{475}+2724z^{500}+160z^{525}.$ The dual code $\mathcal{C}_f^{\bot}$ is a $[624,\ 618,\ 2]$ linear code.
			\end{example}
			
			According to Lemma 8, Remark 4 and Proposition 2, we have the following result on secret sharing scheme.
			\begin{corollary}
				Let $n+s$ be an odd integer with $0\le s\le n-3$ and $\mathcal{C}_f$  be the linear code in Proposition 2 with generator matrix $G=[{g}_0,\ {g}_1,\ \cdots, {g}_{m-1}]$, where $m=p^n-1$. Then, in the secret sharing scheme based on $\mathcal{C}_f^{\bot}$ with $d^{\bot}=2$, the number of participants is $m-1$ and there are $p^n$ minimal access sets. Moreover, if ${g}_j$ is a scalar multiple of ${g}_0$, $1\le j\le m-1$, then participant $P_j$ must be in every minimal access set, or else, 
				participant $P_j$ must be in $(p-1)p^{n-1}$ out of $p^{n}$  minimal access sets.
			\end{corollary}
		
	Let $f(x)$ be a non-weakly regular bent function, i.e., $s=0$, in the following corollary, we give the parameters and weight distributions of $\mathcal{C}_f$ constructed by non-weakly regular bent functions.
		\begin{corollary}
		Let $f(x):\ \mathbb{F}_{p}^n\longrightarrow\mathbb{F}_p$ be a non-weakly regular bent function satisfying $f(0)=0$, $\#B_+(f)=k\ (k\ne0,\text{and}\  k\ne p^{n})$ and the dual $f^*(x)$ of $f(x)$ be bent. Then, when $n$ is even, $\mathcal{C}_f$ defined by (1) is a at most five-weight linear code with parameters $[p^n-1,\ n+1]$ and the weight distribution is given by Table 3; when $n$ is odd, $\mathcal{C}_f$ defined by (1) is a three-weight linear code with parameters $[p^n-1,\ n+1]$ and the weight distribution is given by Table 4.
		\end{corollary}
\begin{table}
	\caption{The weight distribution of $\mathcal{C}_f$ in Corollary 3 when $n$ is even}
	\begin{tabular}{|l|l|l|}
		\hline
		Weight  & Multiplicity  ($0\in B_+(f^*)$)&Multiplicity ($0\in B_-(f^*)$)\\ \hline
		0 & 1&1 \\ \hline
		$(p-1)p^{n-1}$&$p^n-1$&$p^n-1$\\ \hline
		$(p-1)(p^{n-1}-p^{\frac{n}{2}-1})$&$(p-1)(\frac{k}{p}+(p-1)p^{\frac{n}{2}-1})$&$(p-1)\frac{k}{p}$\\ \hline
		$(p-1)(p^{n-1}+p^{\frac{n}{2}-1})$&$(p-1)(p^{n-1}-\frac{k}{p})$&$(p-1)(p^{n-1}-\frac{k}{p}-(p-1)p^{\frac{n}{2}-1})$\\\hline
		$(p-1)p^{n-1}+p^{\frac{n}{2}-1}$&$(p-1)^2(\frac{k}{p}-p^{\frac{n}{2}-1})$&$(p-1)^2\ \frac{k}{p}$\\ \hline
		$(p-1)p^{n-1}-p^{\frac{n}{2}-1}$&$(p-1)^2(p^{n-1}-\frac{k}{p})$&$(p-1)^2(p^{n-1}-\frac{k}{p}
		+p^{\frac{n}{2}-1})$\\\hline
		\end{tabular}
		\end{table}
		
		\begin{table}
			
			\caption{The weight distribution of $\mathcal{C}_f$ in Corollary 3 when $n$ is odd}
			\begin{tabular}{|l|l|}
				\hline
				Weight & Multiplicity\\ \hline
				0&1\\ \hline
				$(p-1)p^{n-1}$&$2p^n-p^{n-1}-1$\\ \hline
				$(p-1)p^{n-1}-p^{\frac{n-1}{2}}$&$\frac{(p-1)^2}{2}(p^{\frac{n-1}{2}}+p^{n-1})$\\ \hline
				$(p-1)p^{n-1}+p^{\frac{n-1}{2}}$&$\frac{(p-1)^2}{2}(p^{n-1}-p^{\frac{n-1}{2}})$\\ \hline
			\end{tabular}
		\end{table}
	
	Next, we verify Corollary 3 by Magma program for the following non-weakly regular bent functions.
	\begin{example}
		Consider $f(x):\ \mathbb{F}_{5}^4\longrightarrow\mathbb{F}_5$, $f(x_1,\ x_2,\ x_3,\ x_4)=4x_1^2x_4^4+4x_1^2x_4^3+3x_1^2x_4+x_1^2+x_2^2+x_3x_4$, which is a non-weakly regular bent function with $0\in B_+(f^*)$, $k=125$, and the dual $f^*(x)$ of $f(x)$ is bent, then $\mathcal{C}_f$ constructed by (1) is a five-weight linear code with parameters $[624,\ 5,\ 480]$, weight enumerator $1+180z^{480}+1600z^{495}+624z^{500}+320z^{505}+400z^{520}$. The dual code $\mathcal{C}_f^{\bot}$ is a $[624,\ 619,\ 2]$ linear code.
		\end{example}
	\begin{example}
		Consider  $f(x):\ \mathbb{F}_{3}^4\longrightarrow\mathbb{F}_3$, $f(x_1,\ x_2,\ x_3,\ x_4)=2x_1^2x_4^2+x_1^2x_4+x_1^2+2x_2^2x_4^2+2x_2^2x_4+x_2^2+x_3x_4$, which is a non-weakly regular bent function with $0\in B_-(f^*)$, $k=54$, and the dual $f^*(x)$ of $f(x)$ is bent, then $\mathcal{C}_f$ constructed by (1) is a five-weight linear code with parameters $[80,\ 5,\ 48]$, weight enumerator $1+36z^{48}+48z^{51}+80z^{54}+72z^{57}+6z^{60}$. The dual code $\mathcal{C}_f^{\bot}$ is a $[80,\ 75,\ 2]$ linear code, which is almost optimal according to the Code Table at http://www.codetables.de/.
		\end{example}
	\begin{example}
		Consider $f(x):\ \mathbb{F}_{5}^5\longrightarrow\mathbb{F}_5$, $f(x_1,\ x_2,\ x_3,\ x_4,\ x_5)=4x_1^2x_5^4+x_1^2x_5^3+2x_1^2x_5+x_1^2+x_2^2+x_3^2+x_4x_5$, which is a non-weakly regular bent function, and the dual $f^*(x)$ of $f(x)$ is bent, then $\mathcal{C}_f$ constructed by (1) is a three-weight linear code with parameters $[3124,\ 6,\ 2475]$, weight enumerator $1+5200z^{2475}+5624z^{2500}+4800z^{2525}$. The dual code $\mathcal{C}_f^{\bot}$ is a $[3124,\ 3118,\ 2]$ linear code.
		\end{example}
			\section{Linear codes from non-weakly regular plateaued functions based on the second generic construction}
			\quad  In this section, motivated by \cite{Mesnager3}, we construct some linear codes with few weights from non-weakly regular $s$-plateaued functions and determine their weight distributions. 
			
			Let $f(x):\mathbb{F}_{p}^n\longrightarrow\mathbb{F}_p$ be a non-weakly regular $s$-plateaued function, for any $a\in\mathrm{Supp}(\widehat{\chi_f})$, define $\epsilon_a=1\ (\text{respectively}\ \epsilon_a=-1)$ if $a\in B_+(f)$  (respectively$\ a\in B_-(f)$).
			Let $\mathcal{NWRF}$ denote the set of non-weakly regular $s$-plateaued functions with  following two properties:\\
			(1) $f(0)$=0;\\
			(2) For any $a\in\mathbb{F}_p^{\times}$, there exists an even positive integer $t$ with $\mathrm{gcd}(t-1,\ p-1)=1$ such that for any $ x\in\mathbb{F}_{p}^n, f(ax)=a^tf(x).$
			
			For $f(x)\in \mathcal{NWRF}$, we have $f(-x)=f(x)$, so $N_0(f)$ is odd, $N_j(f)$ is even ($j\in\mathbb{F}_p,\ j\ne 0$), which implies that $f(x)$ is unbalanced. In the following lemmas, we give some useful results of $f(x)$ which belongs to $\mathcal{NWRF}$.
			\begin{lemma}
				Let $f(x)\in\mathcal{NWRF}$ and  $f^*(x)$ be its dual function. Then there exists an even positive integer $h$ with $\mathrm{gcd}(h-1,\ p-1)=1$ such that $f^*(ax)=a^hf^*(x)$ for any $ a\in\mathbb{F}_p^{\times},\ x\in\mathrm{Supp}(\widehat{\chi_f})$.
				
			\end{lemma}
			\begin{proof}
				Since $\mathrm{gcd}(t-1,\ p-1)=1$,  then there exists a positive integer $l$ such that $l(t-1)\equiv1\ (\mathrm{mod}\ (p-1))$. Let $h=l+1$. For any $a\in\mathbb{F}_p^{\times},\ b\in\mathbb{F}_{p}^n$, we have 
				\begin{align*}
				\widehat{\chi_f}(ab)&=\sum\limits_{x\in\mathbb{F}_{p}^n}\xi_p^{f(x)-ab\cdot x}\\
				&=\sum\limits_{x\in\mathbb{F}_{p}^n}\xi_p^{f(a^lx)-a^{l+1}b \cdot x}\\
				\qquad&=\sum\limits_{x\in\mathbb{F}_{p}^n}\xi_p^{a^hf(x)-a^hb\cdot x}\\
				\qquad&=\sigma_{a^h}(\sum\limits_{x\in\mathbb{F}_{p}^n}\xi_p^{f(x)-b\cdot x})\\
				\qquad&=\sigma_{a^h}(\widehat{\chi_f}(b)).
				\end{align*}
				Thus, for any $ a\in\mathbb{F}_p^{\times},\ b\in\mathbb{F}_{p}^n$, we have  that  $b\in\mathrm{Supp}(\widehat{\chi_f})$ if and only if $ab\in\mathrm{Supp}(\widehat{\chi_f})$. When $b\in\mathrm{Supp}(\widehat{\chi_f})$, we can obtain: \\
				If $n+s$ is even, then  $\epsilon_{ab}p^{\frac{n+s}{2}}\xi_p^{f^*(ab)}=\epsilon_{b}p^{\frac{n+s}{2}}\xi_p^{a^hf^*(b)}$, so $f^*(ab)=a^hf^*(b)$;\\
				If $n+s$ is odd, then $\epsilon_{ab}\sqrt{p^*}p^{\frac{n+s-1}{2}}\xi_p^{f^*(ab)}=\epsilon_{b}\eta(a^h)\sqrt{p^*}p^{\frac{n+s-1}{2}}\xi_p^{a^hf^*(b)}$. Since $h$ is even, then $\eta(a^h)=1$,  so $f^*(ab)=a^hf^*(b)$.\qed
			\end{proof}
			
			\begin{remark}
				By the proof of Lemma 9, for $b\in\mathrm{Supp}(\widehat{\chi_f})$, we have that $b\in B_+(f)$ (respectively$\ B_-(f)$) if and only if $ab\in B_+(f)$ (respectively$\ B_-(f)$) for any $a\in\mathbb{F}_p^{\times}$.
			\end{remark}
			\begin{lemma}
				Let $f(x)\in\mathcal{NWRF}$ and  $f^*(x)$ be its dual function, then $f^*(0)=0$.
			\end{lemma}
			\begin{proof}
				Assume that $f^*(0)=j,\ j\ne 0$. By Lemma 2, we know that $N_j(f)$ is odd. Since $f(x)\in \mathcal{NWRF}$, then $f(x)=f(-x)$ and $f(0)=0$, so $N_j(f)$ should be even, which gives a contradiction. This proof is now completed.\qed
			\end{proof}
			\begin{lemma}
				Let $f(x)\in\mathcal{NWRF}$ with $\#B_+(f)=k\ (k\ne 0\ \text{and}\ k\ne p^{n-s})$, and the dual $f^*(x)$ be bent relative to $\mathrm{Supp}(\widehat{\chi_f})$. Then   the types of $f(x)$ and $f^*(x)$ are the same if $p^{n+s}\equiv1\ (\mathrm{mod}\ 4)$, and different if $p^{n+s}\equiv3\ (\mathrm{mod}\ 4)$.
			\end{lemma}
			\begin{proof}
				Let $f(x)$ be of type $(+)$, then we have $0\in B_+(f)$. Since for any $x\in\mathbb{F}_{p}^n$, if $x\in B_+(f)\ (\text{respectively}\ B_-(f))$, then $-x\in B_+(f)\ (\text{respectively}\ B_-(f))$, thus $k$ is odd. Now, we discuss in three cases.
				\begin{itemize}
					\item[$\bullet$] $n+s$ is even\\
					Assume that $f^*(x)$ is of type $(-)$. By Lemma 4, we have that for any $0\ne j\in\mathbb{F}_p$, $c_j(f^*)=\frac{k}{p},\ d_j(f^*)=p^{n-s-1}+p^{\frac{n-s}{2}-1}-\frac{k}{p}$, hence, $c_j(f^*)$ and $d_j(f^*)$ are odd. By Lemma 9, we can get that for any $x\in\mathrm{Supp}(\widehat{\chi_f})$, $f^*(x)=f^*(-x)$. By Lemma 10, we have $f^*(0)=0$, so $c_j(f^*)$ and $d_j(f^*)$ should be even, which gives a contradiction. Hence, $f^*(x)$ is of type $(+)$.
					\item[$\bullet$] $p\equiv1\ (\mathrm{mod}\ 4)$ and $n+s$ is odd\\
					Assume that $f^*(x)$ is of type $(-)$. By Lemma 4, we have that for any $0\ne j\in\mathbb{F}_p$, $c_j(f^*)=\frac{k}{p},\ d_j(f^*)=p^{n-s-1}-\eta(j)p^{\frac{n-s}{2}-1}-\frac{k}{p}$, hence, $c_j(f^*)$ and $d_j(f^*)$ are odd.  By Lemma 9, we can get that for any $x\in\mathrm{Supp}(\widehat{\chi_f})$, $f^*(x)=f^*(-x)$. By Lemma 10, we have $f^*(0)=0$,  so $c_j(f^*)$ and $d_j(f^*)$ should be even, which gives a contradiction. Hence, $f^*(x)$ is of type $(+)$.
					\item[$\bullet$] $p\equiv3\ (\mathrm{mod}\ 4)$ and $n+s$ is odd\\
					Assume that $f^*(x)$ is of type $(+)$. By Lemma 4, we have that for any $0\ne j\in\mathbb{F}_p$, $c_j(f^*)=\frac{k}{p},\ d_j(f^*)=p^{n-s-1}+\eta(j)p^{\frac{n-s}{2}-1}-\frac{k}{p}$, hence, $c_j(f^*)$ and $d_j(f^*)$ should be odd.  By Lemma 9, we can get that for any $x\in\mathrm{Supp}(\widehat{\chi_f})$, $f^*(x)=f^*(-x)$. By Lemma 10, we have $f^*(0)=0$, so $c_j(f^*)$ and $d_j(f^*)$ should be even, which gives a contradiction. Hence, $f^*(x)$ is of type $(-)$.
					
					When $f(x)$ is of type $(-)$, the proof is similar, so we omit it here.\qed
				\end{itemize}
			\end{proof}
			\subsection{Linear codes with few weights based on $D_f$}
			\quad Let $f(x):\ \mathbb{F}_p^n\longrightarrow\mathbb{F}_p$ be a non-weakly regular $s$-plateaued function belonging to $\mathcal{NWRF}$. Define set $D_f=\{x\in\mathbb{F}_{p}^{n}\setminus\{{0}\}:\ f(x)=0\}=\{x_1,\ x_2,\cdots,\ x_m\}$. In this subsection, we consider the linear code over $\mathbb{F}_p$ defined by  
			\begin{equation} \mathcal{C}_{D_f}=\{{c}(a)=(a\cdot x_1,\ a\cdot x_2,\cdots,\ a\cdot x_m):\ a\in\mathbb{F}_{p}^n\}.
			\end{equation}
			In the following theorem, we give the weights of codewords of $\mathcal{C}_{D_f}$ when $n+s$ is even.
			\begin{theorem}
				Let $n+s$ be an even integer with $0\le s\le n-4$, and $f(x):\ \mathbb{F}_p^n\longrightarrow\mathbb{F}_p$ be a non-weakly regular $s$-plateaued function belonging to $\mathcal{NWRF}$. Then, $\mathcal{C}_{D_f}$ defined by (2) is a $[p^{n-1}+\epsilon_0(p-1)p^{\frac{n+s}{2}-1}-1,\ n]$ linear code, and the weights of codewords are given as follows.
				\[wt({c}(a))=
				\begin{cases}
				0,& \text{if}\ a=0;\\
				(p-1)(p^{n-2}+\epsilon_0(p-1)p^{\frac{n+s}{2}-2}),&\text{if}\  a\notin\mathrm{Supp}(\widehat{\chi_f});\\
				(p-1)(p^{n-2}+(\epsilon_0-\epsilon_a)(p-1)p^{\frac{n+s}{2}-2}),& \text{if}\ a\in \mathrm{Supp}(\widehat{\chi_f}),\\&f^*(a)=0;\\
				(p-1)(p^{n-2}+(\epsilon_0(p-1)+\epsilon_a)p^{\frac{n+s}{2}-2}),& \text{if}\ a\in \mathrm{Supp}(\widehat{\chi_f}),\\&f^*(a)\ne0.
				\end{cases}\]
				The dual code $\mathcal{C}_{D_f}^{\bot}$ is a $[p^{n-1}+\epsilon_0(p-1)p^{\frac{n+s}{2}-1}-1,\ p^{n-1}+\epsilon_0(p-1)p^{\frac{n+s}{2}-1}-1-n,\ 2]$ linear code. 
			\end{theorem}
			\begin{proof}
				By Lemma 10, we have $f^*(0)=0$. Thus, by Lemma 2, we get $\#D_f=p^{n-1}+\epsilon_0(p-1)p^{\frac{n+s}{2}-1}-1$, then $wt({c}(a))=\#D_f+1-\#\{x\in\mathbb{F}_{p}^n:\ f(x)=0,\ a \cdot x=0\}$. Let $N=\#\{x\in\mathbb{F}_{p}^n:\ f(x)=0,\ a\cdot x=0\}$, then
				\begin{flalign*}
				N&=p^{-2}\sum\limits_{x\in\mathbb{F}_{p}^n}\sum\limits_{y\in\mathbb{F}_p}\sum\limits_{z\in\mathbb{F}_p}\xi_p^{yf(x)-za\cdot x}&&\\
				&=p^{-2}(\sum\limits_{z\in\mathbb{F}_p^{\times}}\sum\limits_{y\in\mathbb{F}_p^{\times}}\sum\limits_{x\in\mathbb{F}_{p}^n}\xi_p^{yf(x)-za\cdot x}+\sum\limits_{y\in\mathbb{F}_p^{\times}}\sum\limits_{x\in\mathbb{F}_{p}^n}\xi_p^{yf(x)}+\sum\limits_{z\in\mathbb{F}_p^{\times}}\sum\limits_{x\in\mathbb{F}_{p}^n}\xi_p^{-za\cdot x})&&\\&+p^{n-2}
				&&\\&=p^{-2}(N_1+N_2+N_3)+p^{n-2},
				\end{flalign*}
				where
				\begin{flalign*}
				N_1&=\sum\limits_{z\in\mathbb{F}_p^{\times}}\sum\limits_{y\in\mathbb{F}_p^{\times}}\sum\limits_{x\in\mathbb{F}_{p}^n}\xi_p^{yf(x)-za\cdot x}&&\\
				&=\sum\limits_{z\in\mathbb{F}_p^{\times}}\sum\limits_{y\in\mathbb{F}_p^{\times}}\sigma_y(\sum\limits_{x\in\mathbb{F}_{p}^n}\xi_p^{f(x)-y^{-1}za\cdot x})&&\\
				&=\sum\limits_{z\in\mathbb{F}_p^{\times}}\sum\limits_{y\in\mathbb{F}_p^{\times}}\sigma_y(\widehat{\chi_f}(y^{-1}za)),&&\\
				N_2&=\sum\limits_{y\in\mathbb{F}_p^{\times}}\sum\limits_{x\in\mathbb{F}_{p}^n}\xi_p^{yf(x)}&&\\
				&=\sum\limits_{y\in\mathbb{F}_p^{\times}}\sigma_y(\epsilon_0p^{\frac{n+s}{2}}\xi_p^{f^*(0)})
				=\epsilon_0(p-1)p^{\frac{n+s}{2}},&&\\
				 N_3&=\sum\limits_{z\in\mathbb{F}_p^{\times}}\sum\limits_{x\in\mathbb{F}_{p}^n}\xi_p^{-za\cdot x}=
				\begin{cases}
				(p-1)p^n,&\text{if}\ a=0;\\
				0,&\text{if}\ a\ne0.
				\end{cases}
				\end{flalign*}
				
				Clearly, when $a=0$, $wt({c}(a))=0$. When $a\ne0$, we discuss in two cases.
				\begin{itemize}
					\item[$\bullet$] When $a\notin\mathrm{Supp}(\widehat{\chi_f})$, then for any $y,\ z\in \mathbb{F}_p^{\times}$, we have $ y^{-1}za\notin \mathrm{Supp}(\widehat{\chi_f})$, and  $N_1=0$, so  $wt({c}(a))=(p-1)(p^{n-2}+\epsilon_0(p-1)p^{\frac{n+s}{2}-2})$.
					\item[$\bullet$] When $a\in\mathrm{Supp}(\widehat{\chi_f})$, then for any $y,\ z\in \mathbb{F}_p^{\times}$, we have $y^{-1}za\in \mathrm{Supp}(\widehat{\chi_f})$, and  $N_1=\sum\limits_{z\in\mathbb{F}_p^{\times}}\sum\limits_{y\in\mathbb{F}_p^{\times}}\sigma_y(\epsilon_{y^{-1}za}p^{\frac{n+s}{2}}\xi_p^{f^*(y^{-1}za)}).$
					By Lemma 9 and Remark 6, we have $\epsilon_{y^{-1}za}=\epsilon_a$ and $f^*(y^{-1}za)=y^{-h}z^hf^*(a)$, so $N_1=\sum\limits_{z\in\mathbb{F}_p^{\times}}\sum\limits_{y\in\mathbb{F}_p^{\times}}\epsilon_a\\p^{\frac{n+s}{2}}\xi_p^{y^{-(h-1)}z^hf^*(a)}$. If $f^*(a)=0$,  then $N_1=\epsilon_a(p-1)^2p^{\frac{n+s}{2}}$,  so $wt({c}(a))=(p-1)(p^{n-2}+(\epsilon_0-\epsilon_a)(p-1)p^{\frac{n+s}{2}-2}).$ Note that $\mathrm{gcd}(h-1,\ p-1)=1$, so when $y$ runs through $\mathbb{F}_p^{\times}$, then $y^{-(h-1)}$ runs through $\mathbb{F}_p^{\times}$. By Lemma 1, if $f^*(a)\ne0$, then we have $N_1=-\epsilon_a(p-1)p^{\frac{n+s}{2}}$, so $wt({c}(a))=(p-1)(p^{n-2}+(\epsilon_0(p-1)+\epsilon_a)p^{\frac{n+s}{2}-2}).$
				\end{itemize}
				
				From the above discussion, we can get the weights of codewords of $\mathcal{C}_{D_f}$. Note that  $wt({c}(a))=0$ if and only if $a=0$, thus the number of codewords  of $\mathcal{C}_{D_f}$ is $p^n$, and the dimension of  $\mathcal{C}_{D_f}$ is $n$. Clearly, the dimension of dual code $\mathcal{C}_{D_f}^{\bot}$ is $p^{n-1}+\epsilon_0(p-1)p^{\frac{n+s}{2}-1}-1-n$. Assume that the minimum distance $d^{\bot}$ of $\mathcal{C}_{D_f}^{\bot}$ is 1, then there exist $u\in \mathbb{F}_p^{\times}$ and $x\in D_f$ such that $ux\cdot a=0$ for any $a\in\mathbb{F}_{p}^n$. Since $x\ne 0$, then we arrive at a contradiction. Thus, $d^{\bot}\ne 1$. Due to $f(x)\in\mathcal{NWRF}$, we have $f(x)=f(-x)$, so if $x\in D_f$, then $-x\in D_f$. Since $x\ne-x$ for $x\in D_f$, then we have $x\cdot a+ (-x)\cdot a=0$ for any $a\in\mathbb{F}_{p}^n$, which implies that $d^{\bot}=2$.\qed
			\end{proof}
			\begin{remark}
				By Lemma 7, we have that the linear code $\mathcal{C}_{D_f}$ in Theorem 3 is minimal for $0\le s\le n-6$.
			\end{remark}
			
			For a subclass of non-weakly regular $s$-plateaued functions belonging to $\mathcal{NWRF}$, we determine the weight distributions of linear codes given by Theorem 3 in the following proposition.
			\begin{proposition}
				Let $n+s$ be an even integer with $0\le s\le n-4$,  $f(x):\ \mathbb{F}_p^n\longrightarrow\mathbb{F}_p$ be a non-weakly regular $s$-plateaued function belonging to $\mathcal{NWRF}$ with $\#B_+(f)=k\ (k\ne 0\ \text{and}\ k\ne p^{n-s})$, and the dual $f^*(x)$ of $f(x)$ be  bent relative to $\mathrm{Supp}(\widehat{\chi_f})$. Then the weight distribution of $\mathcal{C}_{D_f}$ is given by Tables 5 and 6. 
			\end{proposition}
			\begin{proof}
				The proof is quite straightforward from Lemma 4 and Theorem 3, so we omit it.\qed
			\end{proof}
			\begin{table}
				\caption{The weight distribution of $\mathcal{C}_{D_f}$ in Proposition 3  when $n+s$ is even and $0\in B_+(f)$}
				\begin{tabular}{|l|l|}
					\hline
					Weight &Multiplicity\\\hline
					0&1\\
					\hline
					$(p-1)(p^{n-2}+(p-1)p^{\frac{n+s}{2}-2})$&$p^n-p^{n-s}$\\
					\hline
					$(p-1)p^{n-2}$&$\frac{k}{p}+(p-1)p^{\frac{n-s}{2}-1}-1$\qquad\qquad\qquad\\
					\hline
					$(p-1)(p^{n-2}+2(p-1)p^{\frac{n+s}{2}-2})$&$p^{n-s-1}-\frac{k}{p}$\\
					\hline
					$(p-1)(p^{n-2}+p^{\frac{n+s}{2}-1})$&$(p-1)(\frac{k}{p}-p^{\frac{n-s}{2}-1})$\\
					\hline
					$(p-1)(p^{n-2}+(p-2)p^{\frac{n+s}{2}-2})$&$ (p-1)(p^{n-s-1}-\frac{k}{p})$\\
					\hline
				\end{tabular}
			\end{table}
			\begin{table}
				\caption{The weight distribution of $\mathcal{C}_{D_f}$ in Proposition 3  when $n+s$ is even and $0\in B_-(f)$}
				\begin{tabular}{|l|l|}
					\hline
					Weight &Multiplicity\\\hline
					0&1\\\hline
					$(p-1)(p^{n-2}-(p-1)p^{\frac{n+s}{2}-2})$&$p^n-p^{n-s}$\\
					\hline
					$(p-1)(p^{n-2}-2(p-1)p^{\frac{n+s}{2}-2})$&$\frac{k}{p}$\\
					\hline
					$(p-1)p^{n-2}$&$p^{n-s-1}-(p-1)p^{\frac{n-s}{2}-1}-\frac{k}{p}-1$\\
					\hline
					$(p-1)(p^{n-2}-(p-2)p^{\frac{n+s}{2}-2})$&$(p-1)\frac{k}{p}$\\
					\hline
					$(p-1)(p^{n-2}-p^{\frac{n+s}{2}-1})$&$(p-1)(p^{n-s-1}+p^{\frac{n-s}{2}-1}-\frac{k}{p})$\\
					\hline
				\end{tabular}
			\end{table}
			
			Let $f(x)\in\mathcal{NWRF}$. Note that for any $x\in \mathbb{F}_{p}^n$, $f(x)=0$ if and only if $f(ax)=f(x)=0$ for any $a\in\mathbb{F}_p^{\times}$. Then we can select a subset $\widetilde{D_f}$ of $D_f$ such that $\bigcup\limits_{a\in \mathbb{F}_p^{\times}}a\widetilde{D_f}$ is just a partition of $D_f$. Let $\widetilde{D}_f=\{\widetilde{x_1},\ \widetilde{x_2},\cdots,\ \widetilde{x_l}\}$.
			Now, we consider the punctured code $\widetilde{\mathcal{C}}_{\widetilde{D_f}}$ defined by
			\begin{equation}
			\widetilde{\mathcal{C}}_{\widetilde{D_f}}=\{\widetilde{{c}}(a)=(a\cdot \widetilde{x_1},\ a\cdot \widetilde{x_2},\cdots,\ a\cdot \widetilde{x_l}):\ a\in\mathbb{F}_{p}^n\}.
			\end{equation}
			 Note that the Hamming weights and length of the punctured code $\widetilde{\mathcal{C}}_{\widetilde{D_f}}$ are given directly from those of $\mathcal{C}_{D_f}$ by dividing them with $(p-1)$. Hence, we have the following corollary.
			\begin{corollary}
				The punctured code $\widetilde{\mathcal{C}}_{\widetilde{D_f}}$ defined by (3)
				of $\mathcal{C}_{D_f}$ in Proposition 3 is a $[\frac{p^{n-1}+\epsilon_0(p-1)p^{\frac{n+s}{2}-1}-1}{p-1},\ n]$ linear code and the weight distribution of $\widetilde{\mathcal{C}}_{\widetilde{D_f}}$ is given by Tables 7 and 8. The dual code $\widetilde{\mathcal{C}}_{\widetilde{D_f}}^{\bot}$ of $\widetilde{\mathcal{C}}_{\widetilde{D_f}}$ is a $[\frac{p^{n-1}+\epsilon_0(p-1)p^{\frac{n+s}{2}-1}-1}{p-1},\\ \frac{p^{n-1}+\epsilon_0(p-1)p^{\frac{n+s}{2}-1}-1}{p-1}-n,\ 3]$ linear code, which is almost optimal according to the sphere packing bound.
			\end{corollary}
			\begin{proof} The weight distribution of $\widetilde{\mathcal{C}}_{\widetilde{D_f}}$ can be easily obtained by Proposition 3 and the dimension of $\widetilde{\mathcal{C}}_{\widetilde{D_f}}^{\bot}$ is clearly $\frac{p^{n-1}+\epsilon_0(p-1)p^{\frac{n+s}{2}-1}-1}{p-1}-n$. By the first four Pless power moments, we have that 
				when $0\in B_+(f)$, $A_1^{\bot}=A_2^{\bot}=0$ and $A_3^{\bot}=\frac{1}{6p^3}(p^{2n}-p^{n+2}(p+1)+p^4-p^{\frac{3n+s}{2}}(p^2-6p+5)+2kp^{\frac{n+3s}{2}}(p^2-3p+2)-p^{\frac{n+s}{2}+2}(p^2-1)+p^{n+s+2}(p-1))>0$, so the minimum distance of $\widetilde{\mathcal{C}}_{\widetilde{D_f}}^{\bot}$ is 3. When $0\in B_-(f)$, note that $k\ne 0$, thus if $s=0$ and $n=4$, then  $A_1^{\bot}=A_2^{\bot}=0$ and $A_3^{\bot}=\frac{2kp^2(p^2-3p+2)}{6p^3}>0$,  so the minimum distance of $\widetilde{\mathcal{C}}_{\widetilde{D_f}}^{\bot}$ is 3; if  $n>4$, then $A_1^{\bot}=A_2^{\bot}=0$ and $A_3^{\bot}=\frac{1}{6p^3}(p^{2n}-p^{n+2}(p+1)+p^4
				-p^{\frac{3n+s}{2}}(p^2-1)+2kp^{\frac{n+3s}{2}}(p^2-3p+2)+p^{\frac{n+s}{2}}(p^4-p^2)+p^{n+s+2}(p-1))>0$, so the minimum distance of $\widetilde{\mathcal{C}}_{\widetilde{D_f}}^{\bot}$ is 3. According to the sphere packing bound, we have that $\widetilde{\mathcal{C}}_{\widetilde{D_f}}^{\bot}$ is almost optimal.\qed
			\end{proof}
			\begin{table}
				\caption{The weight distribution of $\widetilde{\mathcal{C}}_{\widetilde{D_f}}$ in  Corollary 4 when $n+s$ is even and $0\in B_+(f)$}
				\begin{tabular}{|l|l|}
					\hline
					Weight &Multiplicity\\\hline
					0&1\\
					\hline
					$p^{n-2}+(p-1)p^{\frac{n+s}{2}-2}$&$p^n-p^{n-s}$\\
					\hline
					$p^{n-2}$&$\frac{k}{p}+(p-1)p^{\frac{n-s}{2}-1}-1$\qquad\qquad\qquad\\
					\hline
					$p^{n-2}+2(p-1)p^{\frac{n+s}{2}-2}$&$p^{n-s-1}-\frac{k}{p}$\\
					\hline
					$p^{n-2}+p^{\frac{n+s}{2}-1}$&$(p-1)(\frac{k}{p}-p^{\frac{n-s}{2}-1})$\\
					\hline
					$p^{n-2}+(p-2)p^{\frac{n+s}{2}-2}$&$ (p-1)(p^{n-s-1}-\frac{k}{p})$\\
					\hline
				\end{tabular}
			\end{table}
			\begin{table}
				\caption{The weight distribution of $\widetilde{\mathcal{C}}_{\widetilde{D_f}}$  in Corollary 4 when $n+s$ is even and $0\in B_-(f)$}
				\begin{tabular}{|l|l|}
					\hline
					Weight &Multiplicity\\\hline
					0&1\\\hline
					$p^{n-2}-(p-1)p^{\frac{n+s}{2}-2}$&$p^n-p^{n-s}$\\
					\hline
					$p^{n-2}-2(p-1)p^{\frac{n+s}{2}-2}$&$\frac{k}{p}$\\
					\hline
					$p^{n-2}$&$p^{n-s-1}-(p-1)p^{\frac{n-s}{2}-1}-\frac{k}{p}-1$\\
					\hline
					$p^{n-2}-(p-2)p^{\frac{n+s}{2}-2}$&$(p-1)\frac{k}{p}$\\
					\hline
					$p^{n-2}-p^{\frac{n+s}{2}-1}$&$(p-1)(p^{n-s-1}+p^{\frac{n-s}{2}-1}-\frac{k}{p})$\\
					\hline
				\end{tabular}
			\end{table}
			\begin{remark}
				 By Remark 7, we easily have that the linear code $\widetilde{\mathcal{C}}_{\widetilde{D_f}}$ in Corollary 4 is minimal for $0\le s\le n-6$.
			\end{remark}
			
			Next, we verify Proposition 3 and Corollary 4 by Magma program for the following non-weakly regular plateaued functions.
			\begin{example}
				Consider $f(x):\mathbb{F}_{3}^5\longrightarrow\mathbb{F}_3,\ f(x_1,\ x_2,\ x_3,\ x_4,\ x_5)=2x_1^2x_4^2+2x_1^2+x_2^2+x_3x_4$, which is a non-weakly regular $1$-plateaued function belonging to $\mathcal{NWRF}$ with $0\in B_+(f)$, $k=27$, and the dual $f^*(x)$ of $f(x)$ is bent relative to $\mathrm{Supp}(\widehat{\chi_f})$. Then, $\mathcal{C}_{D_f}$ constructed by (2) is a five-weight linear code with parameters $[98,\ 5,\ 54]$, weight enumerator $1+14z^{54}+36z^{60}+162z^{66}+12z^{72}+18z^{78}$. The dual code $\mathcal{C}_{D_f}^{\bot}$ of $\mathcal{C}_{D_f}$ is a $[98,\ 93,\ 2]$ linear code, which is almost optimal according to the Code Table at http://www.codetables.de/. Moreover, the punctured code $\widetilde{\mathcal{C}}_{\widetilde{D_f}}$ constructed by (3) of $\mathcal{C}_{D_f}$ is a $[49,\ 5,\ 27]$ linear code with weight enumerator $1+14z^{27}+36z^{30}+162z^{33}+12z^{36}+18z^{39}$ and the dual code $\widetilde{\mathcal{C}}_{\widetilde{D_f}}^{\bot}$ is a $[49,\ 44,\ 3]$ linear code which is optimal according to the Code Table at http://www.codetables.de/.
			\end{example}
			\begin{example}
				Consider $f(x):\mathbb{F}_{3}^5\longrightarrow\mathbb{F}_3,\ f(x_1,\ x_2,\ x_3,\ x_4,\ x_5)=x_1^2x_4^2+x_1^2+x_2^2+x_3x_4$, which is a non-weakly regular $1$-plateaued function belonging to $\mathcal{NWRF}$ with $0\in B_-(f)$, $k=54$, and the dual $f^*(x)$ of $f(x)$ is bent relative to $\mathrm{Supp}(\widehat{\chi_f})$. Then, $\mathcal{C}_{D_f}$ constructed by (2) is a five-weight linear code with parameters $[62,\ 5,\ 30]$, weight enumerator $1+18z^{30}+24z^{36}+162z^{42}+36z^{48}+2z^{54}$. The dual code  $\mathcal{C}_{D_f}^{\bot}$ of $\mathcal{C}_{D_f}$ is a $[62,\ 57,\ 2]$ linear code, which is almost optimal according to the Code Table at http://www.codetables.de/. Moreover, the punctured code $\widetilde{\mathcal{C}}_{\widetilde{D_f}}$ constructed by (3) of $\mathcal{C}_{D_f}$ is a $[31,\ 5,\ 15]$ linear code with weight enumerator $1+18z^{15}+24z^{18}+162z^{21}+36z^{24}+2z^{27}$ and the dual code $\widetilde{\mathcal{C}}_{\widetilde{D_f}}^{\bot}$ is a $[31,\ 26,\ 3]$ linear code which is optimal according to the Code Table at http://www.codetables.de/.
			\end{example}
			
			According to Lemma 8, Theorem 3, Remark 7 and Corollary 4, we have the following results on secret sharing schemes.
			\begin{corollary}
				Let $n+s$ be an even integer with $0\le s\le n-6$ and $\mathcal{C}_{D_f}$  be the linear code in Theorem 3 with generator matrix $G=[{g}_0,\ {g}_1,\ \cdots, {g}_{m-1}]$, where $m=p^{n-1}+\epsilon_0(p-1)p^{\frac{n+s}{2}-1}-1$. Then, in the secret sharing scheme based on $\mathcal{C}_{D_f}^{\bot}$ with $d^{\bot}=2$, the number of participants is $m-1$ and there are $p^{n-1}$ minimal access sets. Moreover, if ${g}_j$ is a scalar multiple of ${g}_0$, $1\le j\le m-1$, then participant $P_j$ must be in every minimal access set, or else, 
				participant $P_j$ must be in $(p-1)p^{n-2}$ out of $p^{n-1}$  minimal access sets.
			\end{corollary}
			\begin{corollary}
				Let $n+s$ be an even integer with $0\le s\le n-6$ and $\widetilde{\mathcal{C}}_{\widetilde{D_f}}$  be the linear code in Corollary 4 with generator matrix $G=[{g}_0,\ {g}_1,\ \cdots, \ {g}_{m-1}]$, where $m=\frac{p^{n-1}+\epsilon_0(p-1)p^{\frac{n+s}{2}-1}-1}{p-1}$. Then, in the secret sharing scheme based on $\widetilde{\mathcal{C}}_{\widetilde{D_f}}^{\bot}$ with $d^{\bot}=3$, the number of participants is $m-1$ and there are $p^{n-1}$ minimal access sets. Moreover, every participant $P_j$ is involved in $(p-1)p^{n-2}$ out of $p^{n-1}$  minimal access sets.
			\end{corollary}
			
			When $n+s$ is odd, we give the weights of codewords of $\mathcal{C}_{D_f}$ in the following theorem.
			\begin{theorem}
				Let $n+s$ be an odd integer with $0\le s\le n-3$, and $f(x):\ \mathbb{F}_p^n\longrightarrow\mathbb{F}_p$ be a non-weakly regular $s$-plateaued function belonging to $\mathcal{NWRF}$. Then, $\mathcal{C}_{D_f}$ defined by (2) is a $[p^{n-1}-1,\ n]$ linear code, and the weights of codewords are given as follows.  
				\[wt({c}(a))=
				\begin{cases}
				0,&\text{if}\ a=0;\\
				(p-1)p^{n-2},&\text{if}\ a\notin\mathrm{Supp}(\widehat{\chi_f})\ \text{or}\\&a\in\mathrm{Supp}(\widehat{\chi_f}),\ f^*(a)=0;\\
				(p-1)(p^{n-2}-\epsilon_ap^{\frac{n+s-5}{2}}p^*\eta(f^*(a))),&\text{if}\ a\in\mathrm{Supp}(\widehat{\chi_f}),\ f^*(a)\ne 0 .
				\end{cases}\]
				The dual code $\mathcal{C}_{D_f}^{\bot}$ is a $[p^{n-1}-1,\ p^{n-1}-1-n,\ 2]$ linear code.
			\end{theorem}
			\begin{proof}
				By Lemma 10 we have $f^*(0)=0$. Thus, by Lemma 2, we get $\#D_f=p^{n-1}-1$, then $wt({c}(a))=\#D_f+1-\#\{x\in\mathbb{F}_{p}^n:f(x)=0,\ a\cdot x=0\}$. Let $N=\#\{x\in\mathbb{F}_{p}^n:f(x)=0,\ a\cdot x=0\}$, then
				\begin{flalign*}
				N&=p^{-2}\sum\limits_{x\in\mathbb{F}_{p}^n}\sum\limits_{y\in\mathbb{F}_p}\sum\limits_{z\in\mathbb{F}_p}\xi_p^{yf(x)-za\cdot x}&&\\
				&=p^{-2}(\sum\limits_{z\in\mathbb{F}_p^{\times}}\sum\limits_{y\in\mathbb{F}_p^{\times}}\sum\limits_{x\in\mathbb{F}_{p}^n}\xi_p^{yf(x)-za\cdot x}+\sum\limits_{y\in\mathbb{F}_p^{\times}}\sum\limits_{x\in\mathbb{F}_{p}^n}\xi_p^{yf(x)}+\sum\limits_{z\in\mathbb{F}_p^{\times}}\sum\limits_{x\in\mathbb{F}_{p}^n}\xi_p^{-za\cdot x})&&\\&+p^{n-2}&&\\
				&=p^{-2}(N_1+N_2+N_3)+p^{n-2},&&\\
				&\text{where}&&\\
				N_1&=\sum\limits_{z\in\mathbb{F}_p^{\times}}\sum\limits_{y\in\mathbb{F}_p^{\times}}\sum\limits_{x\in\mathbb{F}_{p}^n}\xi_p^{yf(x)-za\cdot x}&&\\
				&=\sum\limits_{z\in\mathbb{F}_p^{\times}}\sum\limits_{y\in\mathbb{F}_p^{\times}}\sigma_y(\sum\limits_{x\in\mathbb{F}_{p}^n}\xi_p^{f(x)-y^{-1}za\cdot x})&&\\
				&=\sum\limits_{z\in\mathbb{F}_p^{\times}}\sum\limits_{y\in\mathbb{F}_p^{\times}}\sigma_y(\widehat{\chi_f}(y^{-1}za)),\end{flalign*}
				\begin{flalign*}
				N_2&=\sum\limits_{y\in\mathbb{F}_p^{\times}}\sum\limits_{x\in\mathbb{F}_{p}^n}\xi_p^{yf(x)}&&\\
				&=\sum\limits_{y\in\mathbb{F}_p^{\times}}\sigma_y(\epsilon_0\sqrt{p^*}p^{\frac{n+s-1}{2}}\xi_p^{f^*(0)})&&\\
				&=\sum\limits_{y\in\mathbb{F}_p^{\times}}\epsilon_0\sqrt{p^*}p^{\frac{n+s-1}{2}}\xi_p^{f^*(0)}\eta(y)=0
				&&\\
				N_3&=\sum\limits_{z\in\mathbb{F}_p^{\times}}\sum\limits_{x\in\mathbb{F}_{p}^n}\xi_p^{-za\cdot x}=
				\begin{cases}
				(p-1)p^n,&\text{if}\ a=0;\\
				0,&\text{if}\ a\ne0.
				\end{cases}
				\end{flalign*}
				
				Clearly, when $a=0$, $wt({c}(a))=0$. When $a\ne0$, we discuss in two cases.
				\begin{itemize}
					\item[$\bullet$] When $a\notin\mathrm{Supp}(\widehat{\chi_f})$, then for any $y,\ z\in \mathbb{F}_p^{\times}$, we have $y^{-1}za\notin \mathrm{Supp}(\widehat{\chi_f})$, $N_1=0$, so $wt({c}(a))=(p-1)p^{n-2}$.
					\item[$\bullet$] When $a\in\mathrm{Supp}(\widehat{\chi_f})$, then for any $y,\ z\in \mathbb{F}_p^{\times}$, we have $y^{-1}za\in \mathrm{Supp}(\widehat{\chi_f})$, $N_1=\sum\limits_{z\in\mathbb{F}_p^{\times}}\sum\limits_{y\in\mathbb{F}_p^{\times}}\sigma_y(\epsilon_{y^{-1}za}\sqrt{p^*}p^{\frac{n+s-1}{2}}\xi_p^{f^*(y^{-1}za)}).$
					By Lemma 9 and Remark 6, we have $\epsilon_{y^{-1}za}=\epsilon_a$, and $f^*(y^{-1}za)=y^{-h}z^hf^*(a)$, so $N_1=\sum\limits_{z\in\mathbb{F}_p^{\times}}\sum\limits_{y\in\mathbb{F}_p^{\times}}\epsilon_a\\\sqrt{p^*}p^{\frac{n+s-1}{2}}\eta(y)\xi_p^{y^{-(h-1)}z^hf^*(a)}$. If $f^*(a)=0$, then $N_1=0$, so $wt({c}(a))=(p-1)p^{n-2}.$ Note that $\mathrm{gcd}(h-1,\ p-1)=1$, so when $y$ runs through $\mathbb{F}_p^{\times}$, $y^{-(h-1)}$ runs through $\mathbb{F}_p^{\times}$. Since $h$ is even, then $\eta(y)=\eta(y^{-(h-1)})$ and $ \eta(z^h)=1$. By Lemma 1, if $f^*(a)\ne0$, we have $N_1=\epsilon_a(p-1)p^{\frac{n+s-1}{2}}p^*\eta(f^*(a))$, so $wt({c}(a))=(p-1)(p^{n-2}-\epsilon_ap^{\frac{n+s-5}{2}}p^*\eta(f^*(a))).$
				\end{itemize}
				
				From the above discussion, we can get the weights of codewords of  $\mathcal{C}_{D_f}$. Note that $wt({c}(a))=0$ if and only if $a=0$, thus the number of codewords of $\mathcal{C}_{D_f}$ is $p^n$ and the dimension of $\mathcal{C}_{D_f}$ is $n$. Clearly, the dimension of dual code $\mathcal{C}_{D_f}^{\bot}$ is $p^{n-1}-1-n$. By the discussion in the proof of Theorem 3, we can easily get that the minimum distance of $\mathcal{C}_{D_f}^{\bot}$ is 2.\qed
			\end{proof}
			\begin{remark}
				By Lemma 7, we have that the linear code $\mathcal{C}_{D_f}$ in Theorem 4 is minimal for $0\le s\le n-5$.
			\end{remark}
			
			For a subclass of non-weakly regular $s$-plateaued functions belonging to $\mathcal{NWRF}$, we determine the weight distributions of linear codes given by Theorem 4 in the following proposition.
			\begin{proposition}
				Let $n+s$ be an odd integer with $0\le s\le n-3$,  $f(x):\ \mathbb{F}_p^n\longrightarrow\mathbb{F}_p$ be a non-weakly regular $s$-plateaued function belonging to $\mathcal{NWRF}$ with $\#B_+(f)=k\ (k\ne 0\ \text{and}\ k\ne p^{n-s})
				$, and the dual $f^*(x)$ of $f(x)$ be bent relative to $\mathrm{Supp}(\widehat{\chi_f})$. Then the weight distribution of $\mathcal{C}_{D_f}$ is given by Table 9. 
			\end{proposition}
			\begin{proof}
				The proof is quite straightforward from Lemma 4 and Theorem 4, so we omit it.\qed
			\end{proof}
			\begin{table}
				\caption{The weight distribution of $\mathcal{C}_{D_f}$ in Proposition 4 when $n+s$ is odd}
				\begin{tabular}{|l|l|}
					\hline
					Weight&Multiplicity\\ \hline
					0&1\\ \hline
					$(p-1)p^{n-2}$&$p^n-(p-1)p^{n-s-1}-1$\\ \hline
					$(p-1)(p^{n-2}-p^{\frac{n+s-3}{2}})$&$\frac{(p-1)}{2}(p^{n-s-1}+p^{\frac{n-s-1}{2}})$\\ \hline
					$(p-1)(p^{n-2}+p^{\frac{n+s-3}{2}})$&$\frac{(p-1)}{2}(p^{n-s-1}-p^{\frac{n-s-1}{2}})$\\ \hline
				\end{tabular}
			\end{table}
			
			Considering the punctured code $\widetilde{\mathcal{C}}_{\widetilde{D_f}}$defined by (3) of $\mathcal{C}_{D_f}$ in Proposition 4, we have the following corollary.
			\begin{corollary}
				The punctured code $\widetilde{\mathcal{C}}_{\widetilde{D_f}}$ defined by (3)
				of $\mathcal{C}_{D_f}$ in Proposition 4 is a $[\frac{p^{n-1}-1}{p-1},\ n]$ linear code and the weight distribution of $\widetilde{\mathcal{C}}_{\widetilde{D_f}}$ is given by Table 10. The dual code $\widetilde{\mathcal{C}}_{\widetilde{D_f}}^{\bot}$ of $\widetilde{\mathcal{C}}_{\widetilde{D_f}}$ is a $[\frac{p^{n-1}-1}{p-1}, \frac{p^{n-1}-1}{p-1}-n]$ linear code and the minimum distance of $\widetilde{\mathcal{C}}_{\widetilde{D_f}}^{\bot}$, denoted by $d^{\bot}$, is 4 if $s=0$ and $n=3$, or else $d^{\bot}=3$ if $n>3$.
			\end{corollary}
			\begin{proof} The weight distribution of $\widetilde{\mathcal{C}}_{\widetilde{D_f}}$ can be easily obtained by Proposition 4 and the dimension of $\widetilde{\mathcal{C}}_{\widetilde{D_f}}^{\bot}$ is clearly $\frac{p^{n-1}-1}{p-1}-n$. By the first five Pless power moments, we have that if $s=0$ and $n=3$, then $A_1^{\bot}=A_2^{\bot}=
				A_3^{\bot}=0$ and $A_4^{\bot}=\frac{p(p-1)^2(p^2-p-2)}{24}>0$, so the minimum distance of $\widetilde{\mathcal{C}}_{\widetilde{D_f}}^{\bot}$ is 4, or else if $n>3$, we have $A_1^{\bot}=A_2^{\bot}=0$ and $A_3^{\bot}=\frac{p^{2n}-p^{n+2}-p^{n+3}+p^{n+s+2}-p^{n+s+1}+p^4}{6p^3}>0$, so the minimum distance of $\widetilde{\mathcal{C}}_{\widetilde{D_f}}^{\bot}$ is 3.\qed
			\end{proof}
			\begin{table}
				\caption{The weight distribution of $\widetilde{\mathcal{C}}_{\widetilde{D_f}}$  in Corollary 7 when $n+s$ is odd}
				\begin{tabular}{|l|l|}
					\hline
					Weight&Multiplicity\\ \hline
					0&1\\ \hline
					$p^{n-2}$&$p^n-(p-1)p^{n-s-1}-1$\\ \hline
					$p^{n-2}-p^{\frac{n+s-3}{2}}$&$\frac{(p-1)}{2}(p^{n-s-1}+p^{\frac{n-s-1}{2}})$\\ \hline
					$p^{n-2}+p^{\frac{n+s-3}{2}}$&$\frac{(p-1)}{2}(p^{n-s-1}-p^{\frac{n-s-1}{2}})$\\ \hline
				\end{tabular}
			\end{table}
			\begin{remark}
				(\romannumeral 1) When $n+s$ is odd, the weight distributions of $\mathcal{C}_{D_f}$ and $\widetilde{\mathcal{C}}_{\widetilde{D_f}}$ in Proposition 4 and Corollary 7 are independent of $k$ and are the same as the weight distributions of linear codes from weakly regular $s$-plateaued functions in \cite{Mesnager3}.\\
				(\romannumeral 2) By Remark 9, we easily have that the linear code $\widetilde{\mathcal{C}}_{\widetilde{D_f}}$ in Corollary 7 is minimal for $0\le s\le n-5$.\\
				(\romannumeral 3) When $s=0$ and $n=3$, $\widetilde{\mathcal{C}}_{\widetilde{D_f}}$ is a linear code with parameters $[p+1,\ 3,\ p-1]$ which is an MDS code and $\widetilde{\mathcal{C}}_{\widetilde{D_f}}^{\bot}$ is a $[p+1,\ p-2,\ 4]$ linear code which is also an MDS code. When $n>3$, $\widetilde{\mathcal{C}}_{\widetilde{D_f}}^{\bot}$ is a $[\frac{p^{n-1}-1}{p-1}, \ \frac{p^{n-1}-1}{p-1}-n,\ 3]$ linear code which is almost optimal according to the sphere packing bound.
			\end{remark}
			
			Next, we verify Proposition 4 and Corollary 7 by Magma program for the following non-weakly regular plateaued function.
			\begin{example}
				Consider $f(x):\mathbb{F}_{5}^4\longrightarrow\mathbb{F}_5,\  f(x_1,\ x_2,\ x_3,\ x_4)=x_1^2x_3^4+x_1^2+x_2x_3$, which is a non-weakly regular $1$-plateaued function belonging to $\mathcal{NWRF}$, and the dual $f^*(x)$ of $f(x)$ is bent relative to $\mathrm{Supp}(\widehat{\chi_f})$. Then, $\mathcal{C}_{D_f}$ constructed by (2) is a three-weight linear code with parameters $[124,\ 4,\ 80]$, weight enumerator $1+60z^{80}+524z^{100}+40z^{120}$. The dual code $\mathcal{C}_{D_f}^{\bot}$ of $\mathcal{C}_{D_f}$ is a $[124,\ 120,\ 2]$ linear code, which is almost optimal according to the Code Table at http://www.codetables.de/. Moreover, the punctured code $\widetilde{\mathcal{C}}_{\widetilde{D_f}}$  constructed by (3) of $\mathcal{C}_{D_f}$ is a $[31,\ 4,\ 20]$ linear code with weight enumerator $1+60z^{20}+524z^{25}+40z^{30}$ and the dual code $\widetilde{\mathcal{C}}_{\widetilde{D_f}}^{\bot}$ is a $[31,\ 27,\ 3]$ linear code which is optimal according to the Code Table at http://www.codetables.de/.
			\end{example}
		
			According to Lemma 8, Theorem 4, Remark 9 and Corollary 7, we have the following results on secret sharing schemes.
			\begin{corollary}
				Let $n+s$ be an odd integer with $0\le s\le n-5$ and $\mathcal{C}_{D_f}$  be the linear code in Theorem 4 with generator matrix $G=[{g}_0,\ {g}_1,\ \cdots, {g}_{m-1}]$, where $m=p^{n-1}-1$. Then, in the secret sharing scheme based on $\mathcal{C}_{D_f}^{\bot}$ with $d^{\bot}=2$, the number of participants is $m-1$ and there are $p^{n-1}$ minimal access sets. Moreover, if ${g}_j$ is a scalar multiple of ${g}_0$, $1\le j\le m-1$, then participant $P_j$ must be in every minimal access set, or else, 
				participant $P_j$ must be in $(p-1)p^{n-2}$ out of $p^{n-1}$  minimal access sets.
			\end{corollary}
			\begin{corollary}
				Let $n+s$ be an odd integer with $0\le s\le n-5$ and $\widetilde{\mathcal{C}}_{\widetilde{D_f}}$  be the linear code in Corollary 7 with generator matrix $G=[{g}_0,\  {g}_1, \cdots, \ {g}_{m-1}]$, where $m=\frac{p^{n-1}-1}{p-1}$. Then, in the secret sharing scheme based on $\widetilde{\mathcal{C}}_{\widetilde{D_f}}^{\bot}$ with $d^{\bot}=3$, the number of participants is $m-1$ and there are $p^{n-1}$ minimal access sets. Moreover, every participant $P_j$ is involved in $(p-1)p^{n-2}$ out of $p^{n-1}$  minimal access sets.
			\end{corollary}
			\begin{table}
				\caption{The weight distribution of $\mathcal{C}_{D_f}$ in Corollary 10  when $n$ is even and $0\in B_+(f)$}
				\begin{tabular}{|l|l|}
					\hline
					Weight &Multiplicity\\\hline
					0&1\\
					\hline
					$(p-1)p^{n-2}$&$\frac{k}{p}+(p-1)p^{\frac{n}{2}-1}-1$\qquad\qquad\qquad\\
					\hline
					$(p-1)(p^{n-2}+2(p-1)p^{\frac{n}{2}-2})$&$p^{n-1}-\frac{k}{p}$\\
					\hline
					$(p-1)(p^{n-2}+p^{\frac{n}{2}-1})$&$(p-1)(\frac{k}{p}-p^{\frac{n}{2}-1})$\\
					\hline
					$(p-1)(p^{n-2}+(p-2)p^{\frac{n}{2}-2})$&$ (p-1)(p^{n-1}-\frac{k}{p})$\\
					\hline
				\end{tabular}
			\end{table}
			
			Let $f(x)$ be a non-weakly regular bent function, i.e., $s=0$, in the following corollary, we give the parameters and weight distributions of $\mathcal{C}_{D_f}$ and $\widetilde{\mathcal{C}}_{\widetilde{D_f}}$ constructed by non-weakly regular bent functions.
			\begin{corollary}
				Let $f(x):\mathbb{F}_{p}^n\longrightarrow\mathbb{F}_p$ be a non-weakly regular bent function belonging to $\mathcal{NWRF}$ with $\#B_+(f)=k\ (k\ne0,\text{and}\  k\ne p^{n})$ and the dual $f^*(x)$ of $f(x)$ be bent. Then, when $n$ is even, $\mathcal{C}_{D_f}$ defined by (2) is a at most four-weight linear code with parameters $[p^{n-1}+\epsilon_0(p-1)p^{\frac{n}{2}-1}-1,\ n]$ and the weight distribution of $\mathcal{C}_{D_f}$ is given by Tables 11 and 12,  $\widetilde{\mathcal{C}}_{\widetilde{D_f}}$ defined by (3) is a at most four-weight linear code with parameters $[\frac{p^{n-1}+\epsilon_0(p-1)p^{\frac{n}{2}-1}-1}{p-1},\ n]$ and the weight distribution of $\widetilde{\mathcal{C}}_{\widetilde{D_f}}$ is given by Tables 13 and 14; when $n$ is odd, $\mathcal{C}_{D_f}$ defined by (2) is a three-weight linear code with parameters $[p^{n-1}-1,\ n]$ and the weight distribution of $\mathcal{C}_{D_f}$ is given by Table 15,  $\widetilde{\mathcal{C}}_{\widetilde{D_f}}$ defined by (3) is a three-weight linear code with parameters $[\frac{p^{n-1}-1}{p-1},\ n]$ and the weight distribution of $\widetilde{\mathcal{C}}_{\widetilde{D_f}}$ is given by Table 16.
			\end{corollary}

Next, we verify Corollary 10 by Magma program for the following non-weakly regular bent functions.
\begin{table}
	\caption{The weight distribution of $\mathcal{C}_{D_f}$ in Corollary 10  when $n$ is even and $0\in B_-(f)$}
	\begin{tabular}{|l|l|}
		\hline
		Weight &Multiplicity\\\hline
		0&1\\\hline
		$(p-1)(p^{n-2}-2(p-1)p^{\frac{n}{2}-2})$&$\frac{k}{p}$\\
		\hline
		$(p-1)p^{n-2}$&$p^{n-1}-(p-1)p^{\frac{n}{2}-1}-\frac{k}{p}-1$\\
		\hline
		$(p-1)(p^{n-2}-(p-2)p^{\frac{n}{2}-2})$&$(p-1)\frac{k}{p}$\\
		\hline
		$(p-1)(p^{n-2}-p^{\frac{n}{2}-1})$&$(p-1)(p^{n-1}+p^{\frac{n}{2}-1}-\frac{k}{p})$\\
		\hline
	\end{tabular}
\end{table}
\begin{table}
	\caption{The weight distribution of $\widetilde{\mathcal{C}}_{\widetilde{D_f}}$ in  Corollary 10 when $n$ is even and $0\in B_+(f)$}
	\begin{tabular}{|l|l|}
		\hline
		Weight &Multiplicity\\\hline
		0&1\\\hline
		$p^{n-2}$&$\frac{k}{p}+(p-1)p^{\frac{n}{2}-1}-1$\qquad\qquad\qquad\\
		\hline
		$p^{n-2}+2(p-1)p^{\frac{n}{2}-2}$&$p^{n-1}-\frac{k}{p}$\\
		\hline
		$p^{n-2}+p^{\frac{n}{2}-1}$&$(p-1)(\frac{k}{p}-p^{\frac{n}{2}-1})$\\
		\hline
		$p^{n-2}+(p-2)p^{\frac{n}{2}-2}$&$ (p-1)(p^{n-1}-\frac{k}{p})$\\
		\hline
	\end{tabular}
\end{table}
\begin{table}
	\caption{The weight distribution of $\widetilde{\mathcal{C}}_{\widetilde{D_f}}$  in Corollary 10 when $n$ is even and $0\in B_-(f)$}
	\begin{tabular}{|l|l|}
		\hline
		Weight &Multiplicity\\\hline
		0&1\\\hline
		$p^{n-2}-2(p-1)p^{\frac{n}{2}-2}$&$\frac{k}{p}$\\
		\hline
		$p^{n-2}$&$p^{n-1}-(p-1)p^{\frac{n}{2}-1}-\frac{k}{p}-1$\\
		\hline
		$p^{n-2}-(p-2)p^{\frac{n}{2}-2}$&$(p-1)\frac{k}{p}$\\
		\hline
		$p^{n-2}-p^{\frac{n}{2}-1}$&$(p-1)(p^{n-1}+p^{\frac{n}{2}-1}-\frac{k}{p})$\\
		\hline
	\end{tabular}
\end{table}

\begin{table}
	\caption{The weight distribution of $\mathcal{C}_{D_f}$ in Crollary 10 when $n$ is odd}
	\begin{tabular}{|l|l|}
		\hline
		Weight&Multiplicity\\ \hline
		0&1\\ \hline
		$(p-1)p^{n-2}$&$p^{n-1}-1$\\ \hline
		$(p-1)(p^{n-2}-p^{\frac{n-3}{2}})$&$\frac{(p-1)}{2}(p^{n-1}+p^{\frac{n-1}{2}})$\\ \hline
		$(p-1)(p^{n-2}+p^{\frac{n-3}{2}})$&$\frac{(p-1)}{2}(p^{n-1}-p^{\frac{n-1}{2}})$\\ \hline
	\end{tabular}
\end{table}

\begin{table}
	\caption{The weight distribution of $\widetilde{\mathcal{C}}_{\widetilde{D_f}}$ in Corollary 10 when $n$ is odd}
	\begin{tabular}{|l|l|}
		\hline
		Weight&Multiplicity\\ \hline
		0&1\\ \hline
		$p^{n-2}$&$p^{n-1}-1$\\ \hline
		$p^{n-2}-p^{\frac{n-3}{2}}$&$\frac{(p-1)}{2}(p^{n-1}+p^{\frac{n-1}{2}})$\\ \hline
		$p^{n-2}+p^{\frac{n-3}{2}}$&$\frac{(p-1)}{2}(p^{n-1}-p^{\frac{n-1}{2}})$\\ \hline
	\end{tabular}
\end{table}

\begin{example}
	Consider $f(x):\mathbb{F}_{3}^4\longrightarrow\mathbb{F}_3$, $f(x_1,\ x_2,\ x_3,\ x_4)=2x_1^2x_4^2+2x_1^2+x_2^2+x_3x_4$, which is a non-weakly regular bent function belonging to $\mathcal{NWRF}$ with $0\in B_+(f)$, $k=27$, and the dual $f^*(x)$ of $f(x)$ is bent. Then $\mathcal{C}_{D_f}$ constructed by (2) is a four-weight linear code with parameters $[32,\ 4,\ 18]$, weight enumerator $1+14z^{18}+36z^{20}+12z^{24}+18z^{26}$. The dual code $\mathcal{C}_{D_f}^{\bot}$ of $\mathcal{C}_{D_f}$ is a $[32,\ 28,\ 2]$ linear code, which is almost optimal according to the Code Table at http://www.codetables.de/. Moreover, the punctured code $\widetilde{\mathcal{C}}_{\widetilde{D_f}}$ constructed by (3) of $\mathcal{C}_{D_f}$ is a $[16,\ 4,\ 9]$ linear code with weight enumerator $1+14z^{9}+36z^{10}+12z^{12}+18z^{13}$ which is optimal according to the Code Table at http://www.codetables.de/ and the dual code $\widetilde{\mathcal{C}}_{\widetilde{D_f}}^{\bot}$ of $\widetilde{\mathcal{C}}_{\widetilde{D_f}}$ is a $[16,\ 12,\ 3]$ linear code which is optimal according to the Code Table at http://www.codetables.de/.
	\end{example}	
\begin{example}
	Consider $f(x):\mathbb{F}_{3}^4\longrightarrow\mathbb{F}_3$, $f(x_1,\ x_2,\ x_3,\ x_4)=x_1^2x_4^2+x_1^2+x_2^2+x_3x_4$, which is a non-weakly bent function belonging to $\mathcal{NWRF}$ with $0\in B_-(f)$, $k=54$, and the dual $f^*(x)$ of $f(x)$ is bent. Then, $\mathcal{C}_{D_f}$ constructed by (2) is a four-weight linear code with parameters $[20,\ 4, 10]$, weight enumerator $1+18z^{10}+24z^{12}+36z^{16}+2z^{18}$. The dual code $\mathcal{C}_{D_f}^{\bot}$ of $\mathcal{C}_{D_f}$ is a $[20,\ 16,\ 2]$ linear code, which is almost optimal according to the Code Table at http://www.codetables.de/. Moreover, the punctured code $\widetilde{\mathcal{C}}_{\widetilde{D_f}}$ constructed by (3) of $\mathcal{C}_{D_f}$ is a $[10,\ 4,\ 5]$ linear code with weight enumerator $1+18z^5+24z^6+36z^8+2z^9$ which is almost optimal according to the Code Table at http://www.codetables.de/ and the dual code $\widetilde{\mathcal{C}}_{\widetilde{D_f}}^{\bot}$ of $\widetilde{\mathcal{C}}_{\widetilde{D_f}}$ is a $[10,\ 6,\ 3]$ linear code which is almost optimal according to the Code Table at http://www.codetables.de/.
\end{example}
\begin{example}
	Consider $f(x):\mathbb{F}_{3}^5\longrightarrow\mathbb{F}_3$, $f(x_1,\ x_2,\ x_3,\ x_4,\ x_5)=x_1^2x_5^2+x_1^2+x_2^2+x_3^2+x_4x_5$, which is a non-weakly bent function belonging to $\mathcal{NWRF}$ and the dual $f^*(x)$ of $f(x)$ is bent. Then, $\mathcal{C}_{D_f}$ constructed by (2) is a three-weight linear code with parameters $[80,\ 5,\ 48 ]$, weight enumerator $1+90z^{48}+80z^{54}+72z^{60}$. The dual code $\mathcal{C}_{D_f}^{\bot}$ of $\mathcal{C}_{D_f}$ is a $[80,\ 75,\ 2]$ linear code, which is almost optimal according to the Code Table at http://www.codetables.de/. Moreover, the punctured code $\widetilde{\mathcal{C}}_{\widetilde{D_f}}$ constructed by (3) of $\mathcal{C}_{D_f}$ is a $[40,\ 5,\ 24]$ linear code with weight enumerator $1+90z^{24}+80z^{27}+72z^{30}$ and the dual code $\widetilde{\mathcal{C}}_{\widetilde{D_f}}^{\bot}$ of $\widetilde{\mathcal{C}}_{\widetilde{D_f}}$ is a $[40,\ 35,\ 3]$ linear code which is  optimal according to the Code Table at http://www.codetables.de/.
	\end{example}

	\subsection{Linear codes with few weights based on $D_{f,SQ}$ and $D_{f,NSQ}$}
			\quad Let $f(x):\ \mathbb{F}_p^n\longrightarrow\mathbb{F}_p$ be a non-weakly regular $s$-plateaued function belonging to $\mathcal{NWRF}$, and $D_{f,SQ}=\{x\in\mathbb{F}_{p}^n\setminus\{0\}:f(x)\in SQ\}=\{x_1,\ x_2,\cdots,\ x_m\}$, $D_{f,NSQ}=\{y\in\mathbb{F}_{p}^{n}\setminus\{0\}:f(y)\in NSQ\}=\{y_1,\ y_2,\cdots,\ y_t\}$. In this subsection, we consider the linear codes over $\mathbb{F}_p$ defined by
			\begin{equation}
			\mathcal{C}_{D_{f,SQ}}=\{{c}(a)=(a\cdot x_1,\ a\cdot x_2,\cdots,\ a\cdot x_m):\ a\in\mathbb{F}_{p}^n\},
			\end{equation}
			\begin{equation}
			\mathcal{C}_{D_{f,NSQ}}=\{{c}(b)=(b\cdot y_1,\ b\cdot y_2,\cdots,\ b\cdot y_t):\ b\in\mathbb{F}_{p}^n\}.
			\end{equation}
			In the following theorem, we give the weights of codewords of $\mathcal{C}_{D_{f,SQ}}$ and $\mathcal{C}_{D_{f,NSQ}}$ when $n+s$ is even.
			
			\begin{theorem}
				Let $n+s$ be an even integer with $0\le s\le n-4$, and $f(x):\ \mathbb{F}_p^n\longrightarrow\mathbb{F}_p$ be a non-weakly regular $s$-plateaued function belonging to $\mathcal{NWRF}$. Then, $\mathcal{C}_{D_{f,SQ}}$ and $\mathcal{C}_{D_{f,NSQ}}$ defined by (4) and (5) are $[\frac{(p-1)}{2}(p^{n-1}-\epsilon_0p^{\frac{n+s}{2}-1}),\ n]$ linear codes, and the weights of codewords are given as follows. 
				\begin{itemize}
					\item[$\bullet$] When the defining set is $D_{f,SQ}$,\\
					$wt({c}(a))=
					\begin{cases}
					0,& \text{if}\ a=0;\\
					\frac{(p-1)^2}{2}(p^{n-2}-\epsilon_0p^{\frac{n+s}{2}-2}),&\text{if}\  a\notin\mathrm{Supp}(\widehat{\chi_f});\\
					\frac{(p-1)^2}{2}(p^{n-2}+(\epsilon_a-\epsilon_0)p^{\frac{n+s}{2}-2}),& \text{if}\ a\in \mathrm{Supp}(\widehat{\chi_f}),\ f^*(a)=0;\\
					\frac{(p-1)^2}{2}(p^{n-2}-\epsilon_0p^{\frac{n+s}{2}-2})-\\\epsilon_a\frac{(p-1)}{2}p^{\frac{n+s}{2}-2}(p\eta(f^*(a))+1)),& \text{if}\ a\in \mathrm{Supp}(\widehat{\chi_f}),\ f^*(a)\ne0.
					\end{cases}$
					
					\item[$\bullet$] When the defining set is $D_{f,NSQ}$,\\
					$wt({c}(b))=
					\begin{cases}
					0,& \text{if}\ b=0;\\
					\frac{(p-1)^2}{2}(p^{n-2}-\epsilon_0p^{\frac{n+s}{2}-2}),&\text{if}\  b\notin\mathrm{Supp}(\widehat{\chi_f});\\
					\frac{(p-1)^2}{2}(p^{n-2}+(\epsilon_b-\epsilon_0)p^{\frac{n+s}{2}-2}),& \text{if}\ b\in \mathrm{Supp}(\widehat{\chi_f}),\ f^*(b)=0;\\
					\frac{(p-1)^2}{2}(p^{n-2}-\epsilon_0p^{\frac{n+s}{2}-2})+\\\epsilon_b\frac{(p-1)}{2}p^{\frac{n+s}{2}-2}(p\eta(f^*(b))-1)),& \text{if}\ b\in \mathrm{Supp}(\widehat{\chi_f}),\ f^*(b)\ne0.
					\end{cases}$
				\end{itemize}
				The dual codes $\mathcal{C}_{D_{f,SQ}}^{\bot}$ and $\mathcal{C}_{D_{f,NSQ}}^{\bot}$ are $[\frac{(p-1)}{2}(p^{n-1}-\epsilon_0p^{\frac{n+s}{2}-1}),\ \frac{(p-1)}{2}(p^{n-1}-\epsilon_0p^{\frac{n+s}{2}-1})-n,\ 2]$ linear codes.
			\end{theorem}
			\begin{proof}
				We only prove the case of $\mathcal{C}_{D_{f,SQ}}$, and the case of $\mathcal{C}_{D_{f,NSQ}}$ is similar.
				
				When the defining set is $D_{f,SQ}$, by Lemma 10, we have $f^*(0)=0$. Thus, by Lemma 2, we get $\#D_{f,SQ}=\frac{(p-1)}{2}(p^{n-1}-\epsilon_0p^{\frac{n+s}{2}-1})$, then $wt({c}(a))=\#D_{f,SQ}-\#\{x\in\mathbb{F}_{p}^n:f(x)\in SQ,\ a\cdot x=0\}$. Let $N=\#\{x\in\mathbb{F}_{p}^n:f(x)\in SQ,\ a\cdot x=0\}$, then
				\begin{flalign*}
				N&=p^{-2}\sum\limits_{j\in SQ}\sum\limits_{x\in\mathbb{F}_{p}^n}\sum\limits_{y\in\mathbb{F}_p}\sum\limits_{z\in\mathbb{F}_p}\xi_p^{y(f(x)-j)-za\cdot x}&&\\
				&=p^{-2}(\sum\limits_{z\in\mathbb{F}_p^{\times}}\sum\limits_{y\in\mathbb{F}_p^{\times}}\sum\limits_{x\in\mathbb{F}_{p}^n}\xi_p^{yf(x)-za\cdot x}\sum\limits_{j\in SQ}\xi_p^{-yj}+\sum\limits_{y\in\mathbb{F}_p^{\times}}\sum\limits_{x\in\mathbb{F}_{p}^n}\xi_p^{yf(x)}\sum\limits_{j\in SQ}\xi_p^{-yj}&&\\
				&+\quad\sum\limits_{j\in SQ}\sum\limits_{z\in\mathbb{F}_p^{\times}}\sum\limits_{x\in\mathbb{F}_{p}^n}\xi_p^{-za\cdot x})+\frac{(p-1)}{2}p^{n-2}&&\\
				&=p^{-2}(N_1+N_2+N_3)+\frac{(p-1)}{2}p^{n-2},
				\end{flalign*}
				where
				\begin{flalign*}
				N_1&=\sum\limits_{z\in\mathbb{F}_p^{\times}}\sum\limits_{y\in\mathbb{F}_p^{\times}}\sum\limits_{x\in\mathbb{F}_{p}^n}\xi_p^{yf(x)-za\cdot x}\sum\limits_{j\in SQ}\xi_p^{-yj}&&
				\end{flalign*}
				\begin{flalign*}
				&=\sum\limits_{z\in\mathbb{F}_p^{\times}}\sum\limits_{y\in\mathbb{F}_p^{\times}}\sigma_y(\sum\limits_{x\in\mathbb{F}_{p}^n}\xi_p^{f(x)-y^{-1}za\cdot x})\sum\limits_{j\in SQ}\xi_p^{-yj}&&\\
				&=\sum\limits_{z\in\mathbb{F}_p^{\times}}\sum\limits_{y\in\mathbb{F}_p^{\times}}\sigma_y(\widehat{\chi_f}(y^{-1}za))\frac{\eta(-y)\sqrt{p^*}-1}{2},&&\\
				N_2&=\sum\limits_{y\in\mathbb{F}_p^{\times}}\sum\limits_{x\in\mathbb{F}_{p}^n}\xi_p^{yf(x)}\sum\limits_{j\in SQ}\xi_p^{-yj}&&\\
				&=\sum\limits_{y\in\mathbb{F}_p^{\times}}\sigma_y(\epsilon_0p^{\frac{n+s}{2}}\xi_p^{f^*(0)})\frac{\eta(-y)\sqrt{p^*}-1}{2}&&\\
				&=-\epsilon_0\frac{(p-1)}{2}p^{\frac{n+s}{2}},&&\\
				N_3&=\sum\limits_{j\in SQ}\sum\limits_{z\in\mathbb{F}_p^{\times}}\sum\limits_{x\in\mathbb{F}_{p}^n}\xi_p^{-za\cdot x}=
				\begin{cases}
				\frac{(p-1)^2}{2}p^n,&\text{if}\ a=0;\\
				0,&\text{if}\ a\ne0.
				\end{cases}
				\end{flalign*}
				Clearly, when $a=0$, $wt({c}(a))=0$. When $a\ne0$, we discuss in two cases.
				\begin{itemize}
					\item[$\bullet$] When $a\notin\mathrm{Supp}(\widehat{\chi_f})$, then for any $y,\ z\in \mathbb{F}_p^{\times}$, we have $ y^{-1}za\notin \mathrm{Supp}(\widehat{\chi_f})$, $N_1=0$, so $wt({c}(a))=\frac{(p-1)^2}{2}(p^{n-2}-\epsilon_0p^{\frac{n+s}{2}-2})$.
					\item[$\bullet$] When $a\in\mathrm{Supp}(\widehat{\chi_f})$, then for any $y,\ z\in \mathbb{F}_p^{\times}$, we have $ y^{-1}za\in \mathrm{Supp}(\widehat{\chi_f})$, $N_1=\sum\limits_{z\in\mathbb{F}_p^{\times}}\sum\limits_{y\in\mathbb{F}_p^{\times}}\sigma_y(\epsilon_{y^{-1}za}p^{\frac{n+s}{2}}\xi_p^{f^*(y^{-1}za)})\frac{\eta(-y)\sqrt{p^*}-1}{2}.$
					By Lemma 9 and Remark 6, we have $\epsilon_{y^{-1}za}=\epsilon_a$,  and $f^*(y^{-1}za)=y^{-h}z^hf^*(a)$, so $N_1=\sum\limits_{z\in\mathbb{F}_p^{\times}}\sum\limits_{y\in\mathbb{F}_p^{\times}}\epsilon_ap^{\frac{n+s}{2}}\xi_p^{y^{-(h-1)}z^hf^*(a)}\frac{\eta(-y)\sqrt{p^*}-1}{2}$. If $f^*(a)=0$, then $N_1=-\epsilon_a\\\frac{(p-1)^2}{2}p^{\frac{n+s}{2}}$, so $wt({c}(a))=\frac{(p-1)^2}{2}(p^{n-2}+(\epsilon_a-\epsilon_0)p^{\frac{n+s}{2}-2}).$ Note that $\mathrm{gcd}(h-1,\ p-1)=1$, so when $y$ runs through $\mathbb{F}_p^{\times}$, $y^{-(h-1)}$ runs through $\mathbb{F}_p^{\times}$. Since $h$ is even, then $\eta(-y)=\eta(-1)\eta(y^{-(h-1)})$, $\eta(z^h)=1$. By Lemma 1, if $f^*(a)\ne0$, we have $N_1=\epsilon_a\frac{(p-1)}{2}p^{\frac{n+s}{2}}(p\eta(f^*(a))+1)$, so $wt({c}(a))=\frac{(p-1)^2}{2}(p^{n-2}-\epsilon_0p^{\frac{n+s}{2}-2})-\frac{(p-1)}{2}\epsilon_ap^{\frac{n+s}{2}-2}(p\eta(f^*(a))+1).$
				\end{itemize}
				
				From the above discussion, we can get the weights of codewords of $\mathcal{C}_{D_{f,SQ}}$. Note that $wt({c}(a))=0$ if and only if $a=0$, thus the number of codewords  of $\mathcal{C}_{D_{f,SQ}}$ is $p^n$ and the dimension  of $\mathcal{C}_{D_{f,SQ}}$ is $n$. Clearly, the dimension of dual code $\mathcal{C}_{D_{f,SQ}}^{\bot}$ is $\frac{(p-1)}{2}(p^{n-1}-\epsilon_0p^{\frac{n+s}{2}-1})-n$. Similar to the discussion in the proof of Theorem 3, we have that the minimum distance  of $\mathcal{C}_{D_{f,SQ}}$ is not 1. Due to $f(x)\in \mathcal{NWRF}$, we have $f(x)=f(-x)$, so if $x\in D_{f,SQ}$, then $-x\in D_{f,SQ}$. Since $x\ne -x$ for $x\in D_{f,SQ}$, then we have $x\cdot a+(-x)\cdot a=0$ for any $a\in\mathbb{F}_{p}^n$ which implies that $d^{\bot}=2.$ \qed
				
			\end{proof}
			\begin{remark}
				By Lemma 7, we have that linear codes $\mathcal{C}_{D_{f,SQ}}$ and $\mathcal{C}_{D_{f,NSQ}}$ are minimal for $0\le s\le n-6$.
			\end{remark}
			\begin{table}
				\caption{The weight distributions of $\mathcal{C}_{D_{f,SQ}}$ and $\mathcal{C}_{D_{f,NSQ}}$  in Proposition 5 when $n+s$ is even and $0\in B_+(f)$}
				\begin{tabular}{|l|l|}
					\hline
					Weight &Multiplicity\\\hline
					0&1\\
					\hline
					$\frac{(p-1)^2}{2}(p^{n-2}-p^{\frac{n+s}{2}-2})$&$p^n-p^{n-s}$\\
					\hline
					$\frac{(p-1)^2}{2}p^{n-2}$&$\frac{(p+1)}{2}\frac{k}{p}+\frac{(p-1)}{2}p^{\frac{n-s}{2}-1}-1$\qquad\qquad\quad\hspace{2cm}\\
					\hline
					$\frac{(p-1)^2}{2}(p^{n-2}-2p^{\frac{n+s}{2}-2} )$&$\frac{(p+1)}{2}(p^{n-s-1}-\frac{k}{p})$\\
					\hline
					$\frac{(p-1)}{2}(p^{n-1}-p^{n-2}-2p^{\frac{n+s}{2}-1})$&$\frac{(p-1)}{2}(\frac{k}{p}-p^{\frac{n-s}{2}-1})$\\
					\hline
					$\frac{(p-1)}{2}(p^{n-1}-p^{n-2}+2p^{\frac{n+s}{2}-2})$&$ \frac{(p-1)}{2}(p^{n-s-1}-\frac{k}{p})$\\
					\hline
				\end{tabular}
			\end{table}

			\begin{table}
				\caption{The weight distributions of $\mathcal{C}_{D_{f,SQ}}$ and $\mathcal{C}_{D_{f,NSQ}}$ in Proposition 5 when $n+s$ is even and $0\in B_-(f)$}
				\begin{tabular}{|l|l|}
					\hline
					Weight &Multiplicity\\\hline
					0&1\\\hline
					$\frac{(p-1)^2}{2}(p^{n-2}+p^{\frac{n+s}{2}-2})$&$p^n-p^{n-s}$\\
					\hline
					$\frac{(p-1)^2}{2}(p^{n-2}+2p^{\frac{n+s}{2}-2})$&$\frac{(p+1)}{2}\frac{k}{p}$\\
					\hline
					$\frac{(p-1)^2}{2}p^{n-2}$&$\frac{(p+1)}{2}(p^{n-s-1}-\frac{k}{p})-\frac{(p-1)}{2}p^{\frac{n-s}{2}-1}-1$\\
					\hline
					$\frac{(p-1)}{2}(p^{n-1}-p^{n-2}-2p^{\frac{n+s}{2}-2})$&$\frac{(p-1)}{2}\frac{k}{p}$\\
					\hline
					$\frac{(p-1)}{2}(p^{n-1}-p^{n-2}+2p^{\frac{n+s}{2}-1})$&$\frac{(p-1)}{2}(p^{n-s-1}-\frac{k}{p}+p^{\frac{n-s}{2}-1})$\\
					\hline
				\end{tabular}
			\end{table}
		
			For a subclass of non-weakly regular $s$-plateaued functions belonging to $\mathcal{NWRF}$, we determine the weight distributions of linear codes given by Theorem 5 in the following proposition.
			
			\begin{proposition}
				Let $n+s$ be an even integer with $0\le s\le n-4$, $f(x):\ \mathbb{F}_p^n\longrightarrow\mathbb{F}_p$ be a non-weakly regular $s$-plateaued function belonging to $\mathcal{NWRF}$ with $\#B_+(f)=k\ (k\ne 0\ \text{and}\ k\ne p^{n-s})
				$, and the dual $f^*(x)$ of $f(x)$ be bent relative to $\mathrm{Supp}(\widehat{\chi_f})$. Then the weight distributions of $\mathcal{C}_{D_{f,SQ}}$ and $\mathcal{C}_{D_{f,NSQ}}$ are given by Tables 17 and 18, respectively. 
			\end{proposition}
			\begin{proof}
				The proof is quite straightforward from Lemma 4 and Theorem 5, so we omit it.\qed
			\end{proof}
			
			Note that when $f(x)\in\mathcal{NWRF}$, since $t$ is even, then for any $x\in\mathbb{F}_{p}^n$, 
			$f(x)\in SQ\ $(respectively $f(x)\in NSQ$) if and only if $f(ax)\in SQ\ $(respectively $f(ax)\in NSQ$) for any $a\in\mathbb{F}_{p}^{\times}$. Then we can select a subset $\widetilde{D}_{f,SQ}$ of $D_{f,SQ}$ such that $\bigcup\limits_{a\in \mathbb{F}_p^{\times}}a\widetilde{D}_{f,SQ}$ is just a partition of $D_{f,SQ}$ and a subset $\widetilde{D}_{f,NSQ}$ of $D_{f,NSQ}$ such that $\bigcup\limits_{a\in \mathbb{F}_p^{\times}}a\widetilde{D}_{f,NSQ}$ is just a partition of $D_{f,NSQ}$. Let $\widetilde{D}_{f,SQ}=\{\widetilde{x_1},\ \widetilde{x_2},\cdots,\ \widetilde{x_l}\}$ and $\widetilde{D}_{f,NSQ}=\{\widetilde{y_1},\ \widetilde{y_2},\cdots,\ \widetilde{y_m}\}$.
			Now, we consider the punctured codes defined by 
			\begin{align}
			\widetilde{{c}}(a)=\{(a\cdot\widetilde{x_1},\ a\cdot\widetilde{x_2},\cdots,\ a\cdot\widetilde{x_l}):a\in \mathbb{F}_{p}^n\};\\
			\widetilde{{c}}(b)=\{(b\cdot \widetilde{y_1},\ b\cdot \widetilde{y_2},\cdots,\ b\cdot \widetilde{y_m}):b\in \mathbb{F}_{p}^n\}.
			\end{align}
			 Note that the Hamming weights and length of the punctured codes $\widetilde{\mathcal{C}}_{\widetilde{D}_{f,SQ}}$, $\widetilde{\mathcal{C}}_{\widetilde{D}_{f,NSQ}}$ are given directly from those of $\mathcal{C}_{D_{f,SQ}}$ and $\mathcal{C}_{D_{f,NSQ}}$ by dividing them with $(p-1)$. Hence, we have the following corollary.
			\begin{corollary}
				The punctured codes $\widetilde{\mathcal{C}}_{\widetilde{D}_{f,SQ}}$ and $\widetilde{\mathcal{C}}_{\widetilde{D}_{f,NSQ}}$ defined by (6)
				and (7) of $\mathcal{C}_{D_{f,SQ}}$ and $\mathcal{C}_{D_{f,NSQ}}$ in Proposition 5 are $[\frac{(p^{n-1}-\epsilon_0p^{\frac{n+s}{2}-1})}{2},\ n]$ linear codes and the weight distributions of $\widetilde{\mathcal{C}}_{\widetilde{D}_{f,SQ}}$ and $\widetilde{\mathcal{C}}_{\widetilde{D}_{f,NSQ}}$ are given by Tables 19 and 20, respectively. The dual codes $\widetilde{\mathcal{C}}_{\widetilde{D}_{f,SQ}}^{\bot}$ and  $\widetilde{\mathcal{C}}_{\widetilde{D}_{f,NSQ}}^{\bot}$ are $[\frac{(p^{n-1}-\epsilon_0p^{\frac{n+s}{2}-1})}{2},\\ \frac{(p^{n-1}-\epsilon_0p^{\frac{n+s}{2}-1})}{2}-n,\ 3]$ linear codes which are almost optimal according to the sphere packing bound.
			\end{corollary}
			\begin{proof}
				The weight distributions of $\widetilde{\mathcal{C}}_{\widetilde{D}_{f,SQ}}$ and $\widetilde{\mathcal{C}}_{\widetilde{D}_{f,NSQ}}$ can be easily obtained by Proposition 5 and the dimensions of $\widetilde{\mathcal{C}}_{\widetilde{D}_{f,SQ}}^{\bot}$ and $\widetilde{\mathcal{C}}_{\widetilde{D}_{f,NSQ}}^{\bot}$ are clearly\\ $\frac{(p^{n-1}-\epsilon_0p^{\frac{n+s}{2}-1})}{2}-n$. By the first four Pless power moments, we have $A_1^{\bot}=A_2^{\bot}=0$, $A_3^{\bot}=\frac{1}{48p^3} (p-1)(4kp^{\frac{3s+n}{2}}(p+1)-(3p^2-4p+5)p^{\frac{n+s}{2}+n}+(p-1)^2p^{2n}-(4p^3-8p^2)(p^n-p^{\frac{n+s}{2}})+(2p^2-6p)p^{n+s})>0$ for $0\in B_+(f)$ and $A_3^{\bot}=\frac{1}{48p^3}(p-1)(p^{\frac{3n+s}{2}}(3p^2-8p+1)+4kp^{\frac{3s+n}{2}}(p+1)+(p-1)^2p^{2n}-(4p^3-8p^2)(p^n+p^{\frac{n+s}{2}})+(2p^2-6p)p^{s+n})>0$ for $0\in B_-(f)$, so the minimum distances of $\widetilde{\mathcal{C}}_{\widetilde{D}_{f,SQ}}^{\bot}$ and $\widetilde{\mathcal{C}}_{\widetilde{D}_{f,NSQ}}^{\bot }$ are 3. According to the sphere packing bound, we have that  $\widetilde{\mathcal{C}}_{\widetilde{D}_{f,SQ}}^{\bot}$ and $\widetilde{\mathcal{C}}_{\widetilde{D}_{f,NSQ}}^{\bot}$ are almost optimal.\qed
			\end{proof}
			\begin{table}
				\caption{The weight distributions of $\widetilde{\mathcal{C}}_{\widetilde{D}_{f,SQ}}$ and $\widetilde{\mathcal{C}}_{\widetilde{D}_{f,NSQ}}$ in Corollary 11  when $n+s$ is even and $0\in B_+(f)$}
				\begin{tabular}{|l|l|}
					\hline
					Weight &Multiplicity\\\hline
					0&1\\
					\hline
					$\frac{(p-1)}{2}(p^{n-2}-p^{\frac{n+s}{2}-2})$&$p^n-p^{n-s}$\\
					\hline
					$\frac{(p-1)}{2}p^{n-2}$&$\frac{(p+1)}{2}\frac{k}{p}+\frac{(p-1)}{2}p^{\frac{n-s}{2}-1}-1$\qquad\qquad\quad\hspace{2cm}\\
					\hline
					$\frac{(p-1)}{2}(p^{n-2}-2p^{\frac{n+s}{2}-2} )$&$\frac{(p+1)}{2}(p^{n-s-1}-\frac{k}{p})$\\
					\hline
					$\frac{p^{n-1}-p^{n-2}-2p^{\frac{n+s}{2}-1}}{2}$&$\frac{(p-1)}{2}(\frac{k}{p}-p^{\frac{n-s}{2}-1})$\\
					\hline
					$\frac{p^{n-1}-p^{n-2}+2p^{\frac{n+s}{2}-2}}{2}$&$ \frac{(p-1)}{2}(p^{n-s-1}-\frac{k}{p})$\\
					\hline
				\end{tabular}
			\end{table}
			\begin{table}
				\caption{The weight distributions of $\widetilde{\mathcal{C}}_{\widetilde{D}_{f,SQ}}$ and $\widetilde{\mathcal{C}}_{\widetilde{D}_{f,NSQ}}$  in Corollary 11 when $n+s$ is even and $0\in B_-(f)$}
				\begin{tabular}{|l|l|}
					\hline
					Weight &Multiplicity\\\hline
					0&1\\\hline
					$\frac{(p-1)}{2}(p^{n-2}+p^{\frac{n+s}{2}-2})$&$p^n-p^{n-s}$\\
					\hline
					$\frac{(p-1)}{2}(p^{n-2}+2p^{\frac{n+s}{2}-2})$&$\frac{(p+1)}{2}\frac{k}{p}$\\
					\hline
					$\frac{(p-1)}{2}p^{n-2}$&$\frac{(p+1)}{2}(p^{n-s-1}-\frac{k}{p})-\frac{(p-1)}{2}p^{\frac{n-s}{2}-1}-1$\\
					\hline
					$\frac{p^{n-1}-p^{n-2}-2p^{\frac{n+s}{2}-2}}{2}$&$\frac{(p-1)}{2}\frac{k}{p}$\\
					\hline
					$\frac{p^{n-1}-p^{n-2}+2p^{\frac{n+s}{2}-1}}{2}$&$\frac{(p-1)}{2}(p^{n-s-1}-\frac{k}{p}+p^{\frac{n-s}{2}-1})$\\
					\hline
				\end{tabular}
			\end{table}
			\begin{remark}
				By Remark 11, we easily have that linear codes $\widetilde{\mathcal{C}}_{\widetilde{D}_{f,SQ}}$ and $\widetilde{\mathcal{C}}_{\widetilde{D}_{f,NSQ}}$ are minimal for $0\le s\le n-6$.
			\end{remark}
			
			Next, we verify Proposition 5 and Corollary 11 by Magma program for the following non-weakly regular plateaued functions.
			\begin{example}
				Consider $f(x)\ :\ \mathbb{F}_{3}^5\longrightarrow\mathbb{F}_3,\ f(x_1,\ x_2,\ x_3,\ x_4,\ x_5)=2x_1^2x_4^2+2x_1^2+x_2^2+x_3x_4$, which is a non-weakly regular $1$-plateaued function belonging to $\mathcal{NWRF}$ with $0\in B_+(f)$, $k=27$, and the dual $f^*(x)$ of $f(x)$ is bent relative to $\mathrm{Supp}(\widehat{\chi_f})$. Then, $\mathcal{C}_{D_{f,SQ}}$ constructed by (4) and $\mathcal{C}_{D_{f,NSQ}}$ constructed by (5) are five-weight linear codes with parameters $[72,\ 5,\ 36]$, weight enumerator $1+6z^{36}+36z^{42}+162z^{48}+20z^{54}+18z^{60}.$ The dual code $\mathcal{C}_{D_{f,SQ}}^{\bot}$ and $\mathcal{C}_{D_{f,NSQ}}^{\bot}$ are $[72,\ 67,\ 2]$ linear codes which are almost optimal according to the Code Table at http://www.codetables.de/. Moreover, the punctured codes $\widetilde{\mathcal{C}}_{\widetilde{D}_{f,SQ}}$ constructed by (6) and  $\widetilde{\mathcal{C}}_{\widetilde{D}_{f,NSQ}}$ constructed by (7) are $[36,\ 5,\ 18]$ linear codes with weight enumerator $1+6z^{18}+36z^{21}+162z^{24}+20z^{27}+18z^{30}$. The dual codes $\widetilde{\mathcal{C}}_{\widetilde{D}_{f,SQ}}^{\bot}$ and $\widetilde{\mathcal{C}}_{\widetilde{D}_{f,NSQ}}^{\bot}$ are $[36,\ 31,\ 3]$ linear codes which are optimal according to the Code Table at http://www.codetables.de/.
			\end{example}
			\begin{example}
				Consider $f(x)\ :\ \mathbb{F}_{3}^5\longrightarrow\mathbb{F}_3,\ f(x_1,\ x_2,\ x_3,\ x_4,\ x_5)=x_1^2x_4^2+x_1^2+x_2^2+x_3x_4$, which is a non-weakly regular $1$-plateaued function belonging to $\mathcal{NWRF}$ with $0\in B_-(f)$, $k=54$, and the dual $f^*(x)$ of $f(x)$ is bent relative to $\mathrm{Supp}(\widehat{\chi_f})$. Then,  $\mathcal{C}_{D_{f,SQ}}$ constructed by (4) and $\mathcal{C}_{D_{f,NSQ}}$ constructed by (5) are five-weight linear codes with parameters $[90,\ 5,\ 48]$, weight enumerator $1+18z^{48}+14z^{54}+162z^{60}+36z^{66}+12z^{72}.$ The dual codes $\mathcal{C}_{D_{f,SQ}}^{\bot}$ and $\mathcal{C}_{D_{f,NSQ}}^{\bot}$ are $[90,\ 85,\ 2]$ linear codes, which are almost optimal according to the Code Table at http://www.codetables.de/. Moreover, the punctured codes  $\widetilde{\mathcal{C}}_{\widetilde{D}_{f,SQ}}$ constructed by (6) and $\widetilde{\mathcal{C}}_{\widetilde{D}_{f,NSQ}}$ constructed by (7) are $[45,\ 5,\ 24]$ linear codes with weight enumerator $1+18z^{24}+14z^{27}+162z^{30}+36z^{33}+12z^{36}$. The dual codes $\widetilde{\mathcal{C}}_{\widetilde{D}_{f,SQ}}^{\bot}$ and $\widetilde{\mathcal{C}}_{\widetilde{D}_{f,NSQ}}^{\bot}$ are $[45,\ 40,\ 3]$ linear codes which are optimal according to the Code Table at http://www.codetables.de/.
			\end{example}
			
			According to Lemma 8, Theorem 5, Remark 11 and Corollary 11, we have the following results on secret sharing schemes. 
			\begin{corollary}
				Let $n+s$ be an even integer with $0\le s\le n-6$ and $\mathcal{C}_{D_{f,SQ}}$ and $\mathcal{C}_{D_{f,NSQ}}$ be  linear codes in Theorem 5 with generator matrix $G=[{g}_0,\ {g}_1,\ \cdots, \\{g}_{m-1}]$, where $m=\frac{(p-1)}{2}(p^{n-1}-\epsilon_0p^{\frac{n+s}{2}-1})$. Then, in the secret sharing schemes based on $\mathcal{C}_{D_{f,SQ}}^{\bot}$ and $\mathcal{C}_{D_{f,NSQ}}^{\bot}$ with $d^{\bot}=2$, the number of participants is $m-1$ and there are $p^{n-1}$ minimal access sets. Moreover, if ${g}_j$ is a scalar multiple of ${g}_0$, $1\le j\le m-1$, then participant $P_j$ must be in every minimal access set, or else, 
				participant $P_j$ must be in $(p-1)p^{n-2}$ out of $p^{n-1}$  minimal access sets.
			\end{corollary}
			\begin{corollary}
				Let $n+s$ be an even integer with $0\le s\le n-6$ and $\widetilde{\mathcal{C}}_{\widetilde{D}_{f,SQ}}$ and $\widetilde{\mathcal{C}}_{\widetilde{D}_{f,NSQ}}$ be  linear codes in Corollary 11 with generator matrix $G=[{g}_0,\ {g}_1,\ \cdots, \\{g}_{m-1}]$, where $m=\frac{(p^{n-1}-\epsilon_0p^{\frac{n+s}{2}-1})}{2}$. Then, in the secret sharing schemes based on $\widetilde{\mathcal{C}}_{\widetilde{D}_{f,SQ}}^{\bot}$ and $\widetilde{\mathcal{C}}_{\widetilde{D}_{f,NSQ}}^{\bot}$ with $d^{\bot}=3$, the number of participants is $m-1$ and there are $p^{n-1}$ minimal access sets. Moreover, every participant $P_j$ is involved in $(p-1)p^{n-2}$ out of $p^{n-1}$  minimal access sets.
			\end{corollary}
		
			When $n+s$ is odd, we give the weights of codewords of $\mathcal{C}_{D_{f,SQ}}$ and $\mathcal{C}_{D_{f,NSQ}}$ in the following theorem.
			\begin{theorem}
				Let $n+s$ be an odd integer with $0\le s\le n-3$, and $f(x):\ \mathbb{F}_p^n\longrightarrow\mathbb{F}_p$ be a non-weakly regular $s$-plateaued function belonging to $\mathcal{NWRF}$. Then, $\mathcal{C}_{D_{f,SQ}}$ defined by (4) is a $[\frac{(p-1)}{2}(p^{n-1}+\epsilon_0p^{\frac{n+s-1}{2}}),\ n]$ linear code, $\mathcal{C}_{D_{f,NSQ}}$ defined by (5) is a $[\frac{(p-1)}{2}(p^{n-1}-\epsilon_0p^{\frac{n+s-1}{2}}),\ n]$ linear code, and the weights of codewords are given as follows. 
				\begin{itemize}
					\item[$\bullet$] 
					When the defining set is $D_{f,SQ}$,\\
					$wt({c}(a))=
					\begin{cases}
					0,&\text{if}\ a=0;\\
					\frac{(p-1)^2}{2}(p^{n-2}+\epsilon_0p^{\frac{n+s-3}{2}}),&\text{if}\ a\notin\mathrm{Supp}(\widehat{\chi_f});\\
					\frac{(p-1)^2}{2}(p^{n-2}+(\epsilon_0-\epsilon_a)p^{\frac{n+s-3}{2}}),&\text{if}\ a\in\mathrm{Supp}(\widehat{\chi_f}),\ f^*(a)=0;\\
					\frac{(p-1)^2}{2}(p^{n-2}+\epsilon_0p^{\frac{n+s-3}{2}})+\\\epsilon_a\frac{(p-1)}{2}p^{\frac{n+s-5}{2}}(p+p^*\eta(f^*(a)),&\text{if}\ a\in\mathrm{Supp}(\widehat{\chi_f}),\ f^*(a)\ne0.
					\end{cases}$
					
					\item[$\bullet$] 
					When the defining set is $D_{f,NSQ}$,\\
					$wt({c}(b))=
					\begin{cases}
					0,&\text{if}\ b=0;\\
					\frac{(p-1)^2}{2}(p^{n-2}-\epsilon_0p^{\frac{n+s-3}{2}}),&\text{if}\ b\notin\mathrm{Supp}(\widehat{\chi_f});\\
					\frac{(p-1)^2}{2}(p^{n-2}-(\epsilon_0-\epsilon_b)p^{\frac{n+s-3}{2}}),&\text{if}\ b\in\mathrm{Supp}(\widehat{\chi_f}),\ f^*(b)=0;\\
					\frac{(p-1)^2}{2}(p^{n-2}-\epsilon_0p^{\frac{n+s-3}{2}})-\\\epsilon_b\frac{(p-1)}{2}p^{\frac{n+s-5}{2}}(p-p^*\eta(f^*(b)),&\text{if}\ b\in\mathrm{Supp}(\widehat{\chi_f}),\ f^*(b)\ne0.
					\end{cases}$
					
				\end{itemize}
				The dual code $\mathcal{C}_{D_{f,SQ}}^{\bot}$ is a $[\frac{(p-1)}{2}(p^{n-1}+\epsilon_0p^{\frac{n+s-1}{2}}),\ \frac{(p-1)}{2}(p^{n-1}+\epsilon_0p^{\frac{n+s-1}{2}})-n,\ 2]$ linear code and $\mathcal{C}_{D_{f,NSQ}}^{\bot}$ is a $[\frac{(p-1)}{2}(p^{n-1}-\epsilon_0p^{\frac{n+s-1}{2}}),\ \frac{(p-1)}{2}(p^{n-1}-\epsilon_0p^{\frac{n+s-1}{2}})-n,\ 2]$ linear code.
			\end{theorem}
			\begin{proof}
				We only prove the case of $\mathcal{C}_{D_{f,SQ}}$, and the case of  $\mathcal{C}_{D_{f,NSQ}}$ is similar.
				
				When the defining set is $D_{f,SQ}$, by Lemma 10, we have $f^*(0)=0$. Thus, by Lemma 2, we get $\#D_{f,SQ}=\frac{(p-1)}{2}(p^{n-1}+\epsilon_0p^{\frac{n+s-1}{2}})$, then $wt({c}(a))=\#D_{f,SQ}-\#\{x\in\mathbb{F}_{p}^n:f(x)\in SQ,\ a\cdot x=0\}$. Let $N=\#\{x\in\mathbb{F}_{p}^n:f(x)\in SQ,\ a\cdot x=0\}$, then
				\begin{flalign*}
				N&=p^{-2}\sum\limits_{j\in SQ}\sum\limits_{x\in\mathbb{F}_{p}^n}\sum\limits_{y\in\mathbb{F}_p}\sum\limits_{z\in\mathbb{F}_p}\xi_p^{y(f(x)-j)-za\cdot x}&&\\
				&=p^{-2}(\sum\limits_{z\in\mathbb{F}_p^{\times}}\sum\limits_{y\in\mathbb{F}_p^{\times}}\sum\limits_{x\in\mathbb{F}_{p}^n}\xi_p^{yf(x)-za\cdot x}\sum\limits_{j\in SQ}\xi_p^{-yj}+\sum\limits_{y\in\mathbb{F}_p^{\times}}\sum\limits_{x\in\mathbb{F}_{p}^n}\xi_p^{yf(x)}\sum\limits_{j\in SQ}\xi_p^{-yj}&&\\
				&+\quad\sum\limits_{j\in SQ}\sum\limits_{z\in\mathbb{F}_p^{\times}}\sum\limits_{x\in\mathbb{F}_{p}^n}\xi_p^{-za\cdot x})+\frac{(p-1)}{2}p^{n-2}&&\\
				&=p^{-2}(N_1+N_2+N_3)+\frac{(p-1)}{2}p^{n-2},
				\end{flalign*}
				where
				\begin{flalign*}
				N_1&=\sum\limits_{z\in\mathbb{F}_p^{\times}}\sum\limits_{y\in\mathbb{F}_p^{\times}}\sum\limits_{x\in\mathbb{F}_{p}^n}\xi_p^{yf(x)-za\cdot x}\sum\limits_{j\in SQ}\xi_p^{-yj}&&\\
				&=\sum\limits_{z\in\mathbb{F}_p^{\times}}\sum\limits_{y\in\mathbb{F}_p^{\times}}\sigma_y(\sum\limits_{x\in\mathbb{F}_{p}^n}\xi_p^{f(x)-y^{-1}za\cdot x})\sum\limits_{j\in SQ}\xi_p^{-yj}&&\\
				&=\sum\limits_{z\in\mathbb{F}_p^{\times}}\sum\limits_{y\in\mathbb{F}_p^{\times}}\sigma_y(\widehat{\chi_f}(y^{-1}za))\frac{\eta(-y)\sqrt{p^*}-1}{2},&&\\
				N_2&=\sum\limits_{y\in\mathbb{F}_p^{\times}}\sum\limits_{x\in\mathbb{F}_{p}^n}\xi_p^{yf(x)}\sum\limits_{j\in SQ}\xi_p^{-yj}&&\\
				&=\sum\limits_{y\in\mathbb{F}_p^{\times}}\sigma_y(\epsilon_0\sqrt{p^*}p^{\frac{n+s-1}{2}}\xi_p^{f^*(0)})\frac{\eta(-y)\sqrt{p^*}-1}{2}&&\\
				&=\epsilon_0\frac{(p-1)}{2}p^{\frac{n+s+1}{2}}\end{flalign*}
				\begin{flalign*}
				N_3&=\sum\limits_{j\in SQ}\sum\limits_{z\in\mathbb{F}_p^{\times}}\sum\limits_{x\in\mathbb{F}_{p}^n}\xi_p^{-za\cdot x}=
				\begin{cases}
				\frac{(p-1)^2}{2}p^n,&\text{if}\ a=0;\\
				0,&\text{if}\ a\ne0.
				\end{cases}&&
				\end{flalign*}
				
				Clearly, when $a=0$, $wt({c}(a))=0$. When $a\ne0$, we discuss in two cases.
				\begin{itemize}
					\item[$\bullet$] When $a\notin\mathrm{Supp}(\widehat{\chi_f})$, then for any $y,\ z\in \mathbb{F}_p^{\times}$, we have $ y^{-1}za\notin \mathrm{Supp}(\widehat{\chi_f})$, $N_1=0$, so $wt({c}(a))=\frac{(p-1)^2}{2}(p^{n-2}+\epsilon_0p^{\frac{n+s-3}{2}})$.
					\item[$\bullet$] When $a\in\mathrm{Supp}(\widehat{\chi_f})$, then for any $y,\ z\in \mathbb{F}_p^{\times}$, we have $y^{-1}za\in \mathrm{Supp}(\widehat{\chi_f})$, $N_1=\sum\limits_{z\in\mathbb{F}_p^{\times}}\sum\limits_{y\in\mathbb{F}_p^{\times}}\sigma_y(\epsilon_{y^{-1}za}\sqrt{p^*}p^{\frac{n+s-1}{2}}\xi_p^{f^*(y^{-1}za)})\frac{\eta(-y)\sqrt{p^*}-1}{2}.$
					By Lemma 9 and Remark 6, we have $\epsilon_{y^{-1}za}=\epsilon_a$, and $f^*(y^{-1}za)=y^{-h}z^hf^*(a)$, so $N_1=\sum\limits_{z\in\mathbb{F}_p^{\times}}\sum\limits_{y\in\mathbb{F}_p^{\times}}\epsilon_ap^{\frac{n+s-1}{2}}\xi_p^{y^{-(h-1)}z^hf^*(a)}\frac{p-\eta(y)\sqrt{p^*}}{2}$. If $f^*(a)=0$, then $N_1=\epsilon_a\frac{(p-1)^2}{2}p^{\frac{n+s+1}{2}}$, so $wt({c}(a))=\frac{(p-1)^2}{2}(p^{n-2}+(\epsilon_0-\epsilon_a)p^{\frac{n+s-3}{2}}).$ Note that $\mathrm{gcd}(h-1,\ p-1)=1$, so when $y$ runs through $\mathbb{F}_p^{\times}$, $y^{-(h-1)}$ runs through $\mathbb{F}_p^{\times}$. Since $h$ is even, then $\eta(y)=\eta(y^{-(h-1)}),\ \eta(z^h)=1$. By Lemma 1, if $f^*(a)\ne 0$, we have $N_1=-\epsilon_a\frac{(p-1)}{2}p^{\frac{n+s-1}{2}}(p+p^*\eta(f^*(a)))$, so $wt({c}(a))=\frac{(p-1)^2}{2}(p^{n-2}+\epsilon_0p^{\frac{n+s-3}{2}})+\epsilon_a\frac{(p-1)}{2}p^{\frac{n+s-5}{2}}(p+p^*\eta(f^*(a)).$
				\end{itemize}
				
				From the above discussion, we can get the weights of codewords of $\mathcal{C}_{D_{f,SQ}}$. Note that  $wt({c}(a))=0$ if and only if $a=0$, thus the number of codewords of $\mathcal{C}_{D_{f,SQ}}$ is $p^n$, the dimension of $\mathcal{C}_{D_{f,SQ}}$ is $n$. By the discussion in the proof of Theorem 5, we can easily get that the minimum distance  of $\mathcal{C}_{D_{f,SQ}}$ is 2.\qed
			\end{proof}
			\begin{remark}
				By Lemma 7, we have that linear codes $\mathcal{C}_{D_{f,SQ}}$ and $\mathcal{C}_{D_{f,NSQ}}$ are minimal for $0\le s\le n-5$.
			\end{remark}
			\begin{table}
				\caption{The weight distribution of $\mathcal{C}_{D_{f,SQ}}$ in Proposition 6 when $n+s$ is odd and $0\in B_+(f)$ and the weight distribution of $\mathcal{C}_{D_{f,NSQ}}$ in Proposition 6 when $n+s$ is odd and $0\in B_-(f)$}
				\begin{tabular}{|l|l|l|}
					\hline
					Weight & Multiplicity of $\mathcal{C}_{D_{f,SQ}}$ &Multiplicity of $\mathcal{C}_{D_{f,NSQ}}$ \\ 
					&($0\in B_+(f))$&($0\in B_-(f))$\\\hline
					0&1&1\\ \hline
					$\frac{(p-1)^2}{2}(p^{n-2}+p^{\frac{n+s-3}{2}})$&$p^n-p^{n-s}+\frac{(p-1)}{2}$&$p^n-p^{n-s}+\frac{(p-1)}{2}$\\ 
					&$(p^{n-s-1}-p^{\frac{n-s-1}{2}})$&$(p^{n-s-1}-p^{\frac{n-s-1}{2}})$\\\hline
					$\frac{(p-1)^2}{2}p^{n-2}$&$\frac{k}{p}-1$&$p^{n-s-1}-\frac{k}{p}-1$\\ \hline
					$\frac{(p-1)^2}{2}(p^{n-2}+2p^{\frac{n+s-3}{2}})$&$p^{n-s-1}-\frac{k}{p}$&$\frac{k}{p}$\\ \hline
					$\frac{(p-1)^2}{2}p^{n-2}+\frac{(p^2-1)}{2}p^{\frac{n+s-3}{2}}$&$\frac{(p-1)}{2}(\frac{k}{p}+p^{\frac{n-s-1}{2}})$&$\frac{(p-1)}{2}(p^{n-s-1}-\frac{k}{p}$\\&&$+p^{\frac{n-s-1}{2}})$\\\hline
					$\frac{(p-1)^2}{2}p^{n-2}+\frac{(p-1)(p-3)}{2}p^{\frac{n+s-3}{2}}$&$\frac{(p-1)}{2}(p^{n-s-1}-\frac{k}{p})$&$\frac{(p-1)}{2}\frac{k}{p}$\\\hline
				\end{tabular}
			\end{table}
			
			\begin{table}
				\caption{The weight distribution of $\mathcal{C}_{D_{f,SQ}}$ in Proposition 6 when $n+s$ is odd and $0\in B_-(f)$ and the weight distribution of $\mathcal{C}_{D_{f,NSQ}}$ in Proposition 6 when $n+s$ is odd and $0\in B_+(f)$}
				
				\begin{tabular}{|l|l|l|}
					\hline
					Weight &Multiplicity of  $\mathcal{C}_{D_{f,SQ}}$ &Multiplicity of $\mathcal{C}_{D_{f,NSQ}}$\\
					&($0\in B_-(f)$)&($0\in B_+(f)$)\\\hline
					0&1&1\\ \hline
					$\frac{(p-1)^2}{2}(p^{n-2}-p^{\frac{n+s-3}{2}})$&$p^n-p^{n-s}+\frac{(p-1)}{2}$&$p^n-p^{n-s}+\frac{(p-1)}{2}$\\
					&$(p^{n-s-1}+p^{\frac{n-s-1}{2}})$&$(p^{n-s-1}+p^{\frac{n-s-1}{2}})$\\ \hline
					$\frac{(p-1)^2}{2}(p^{n-2}-2p^{\frac{n+s-3}{2}})$&$\frac{k}{p}$&$p^{n-s-1}-\frac{k}{p}$\\ \hline
					$\frac{(p-1)^2}{2}p^{n-2}$&$p^{n-s-1}-\frac{k}{p}-1$&$\frac{k}{p}-1$\\ \hline
					$\frac{(p-1)^2}{2}p^{n-2}-\frac{(p-1)(p-3)}{2}p^{\frac{n+s-3}{2}}$&$\frac{(p-1)k}{2p}$&$\frac{(p-1)}{2}(p^{n-s-1}-\frac{k}{p})$\\\hline
					$\frac{(p-1)^2}{2}p^{n-2}-\frac{(p^2-1)}{2}p^{\frac{n+s-3}{2}}$&$\frac{(p-1)}{2}(p^{n-s-1}-p^{\frac{n-s-1}{2}}$&$\frac{(p-1)}{2}(\frac{k}{p}-p^{\frac{n-s-1}{2}})$\\
					&$-\frac{k}{p})$&\\\hline
				\end{tabular}
			\end{table}
			
			For a subclass of non-weakly regular $s$-plateaued functions belonging to $\mathcal{NWRF}$, we determine the weight distributions of linear codes given by Theorem 6 in the following proposition.
			\begin{proposition}
				Let $n+s$ be an odd integer with $0\le s\le n-3$, $f(x):\ \mathbb{F}_p^n\longrightarrow\mathbb{F}_p$ be a non-weakly regular $s$-plateaued function belonging to $\mathcal{NWRF}$ with $\#B_+(f)=k\ (k\ne 0\ \text{and}\ k\ne p^{n-s})
				$, and the dual $f^*(x)$ of $f(x)$ be bent relative to $\mathrm{Supp}(\widehat{\chi_f})$. Then the weight distributions of $\mathcal{C}_{D_{f,SQ}}$ and $\mathcal{C}_{D_{f,NSQ}}$ are given by Tables 21 and 22, respectively.
			\end{proposition}
			\begin{proof}
				The proof is quite straightforward from Lemma 4 and Theorem 6, so we omit it.\qed
			\end{proof}
			\begin{table}
				\caption{The weight distribution of $\widetilde{\mathcal{C}}_{\widetilde{D}_{f,SQ}}$ in Corollary 14 when $n+s$ is odd and $0\in B_+(f)$  and the weight distribution of $\widetilde{\mathcal{C}}_{\widetilde{D}_{f,NSQ}}$ in Corollary 14 when $n+s$ is odd and $0\in B_-(f)$}
				
				\begin{tabular}{|l|l|l|}
					\hline
					Weight & Multiplicity of $\widetilde{\mathcal{C}}_{\widetilde{D}_{f,SQ}}$ &Multiplicity of $\widetilde{\mathcal{C}}_{\widetilde{D}_{f,NSQ}}$ \\ 
					&($0\in B_+(f)$)&($0\in B_-(f)$)\\\hline
					0&1&1\\ \hline
					$\frac{(p-1)}{2}(p^{n-2}+p^{\frac{n+s-3}{2}})$&$p^n-p^{n-s}+\frac{(p-1)}{2}$&$p^n-p^{n-s}+\frac{(p-1)}{2}$\\
					&$(p^{n-s-1}-p^{\frac{n-s-1}{2}})$&$(p^{n-s-1}-p^{\frac{n-s-1}{2}})$\\ \hline
					$\frac{(p-1)}{2}p^{n-2}$&$\frac{k}{p}-1$&$p^{n-s-1}-\frac{k}{p}-1$\\ \hline
					$\frac{(p-1)}{2}(p^{n-2}+2p^{\frac{n+s-3}{2}})$&$p^{n-s-1}-\frac{k}{p}$&$\frac{k}{p}$\\ \hline
					$\frac{(p-1)}{2}p^{n-2}+\frac{(p+1)}{2}p^{\frac{n+s-3}{2}}$&$\frac{(p-1)}{2}(\frac{k}{p}+p^{\frac{n-s-1}{2}})$&$\frac{(p-1)}{2}(p^{n-s-1}-\frac{k}{p}$\\
					&&$+p^{\frac{n-s-1}{2}})$\\\hline
					$\frac{(p-1)}{2}p^{n-2}+\frac{(p-3)}{2}p^{\frac{n+s-3}{2}}$&$\frac{(p-1)}{2}(p^{n-s-1}-\frac{k}{p})$&$\frac{(p-1)}{2}\frac{k}{p}$\\\hline
				\end{tabular}
			\end{table}
			
			\begin{table}
				\caption{The weight distribution of $\widetilde{\mathcal{C}}_{\widetilde{D}_{f,SQ}}$ in Corollary 14 when $n+s$ is odd and $0\in B_-(f)$ and the weight distribution of $\widetilde{\mathcal{C}}_{\widetilde{D}_{f,NSQ}}$ in Corollary 14 when $n+s$ is odd and $0\in B_+(f)$}
				
				\begin{tabular}{|l|l|l|}
					\hline
					Weight &Multiplicity of  $\widetilde{\mathcal{C}}_{\widetilde{D}_{f,SQ}}$ &Multiplicity of $\widetilde{\mathcal{C}}_{\widetilde{D}_{f,NSQ}}$\\ 
					&($0\in B_-(f)$)&($0\in B_+(f)$)\\\hline
					0&1&1\\ \hline
					$\frac{(p-1)}{2}(p^{n-2}-p^{\frac{n+s-3}{2}})$&$p^n-p^{n-s}+\frac{(p-1)}{2}$&$p^n-p^{n-s}+\frac{(p-1)}{2}$\\
					&$(p^{n-s-1}+p^{\frac{n-s-1}{2}})$&$(p^{n-s-1}+p^{\frac{n-s-1}{2}})$\\ \hline
					$\frac{(p-1)}{2}(p^{n-2}-2p^{\frac{n+s-3}{2}})$&$\frac{k}{p}$&$p^{n-s-1}-\frac{k}{p}$\\ \hline
					$\frac{(p-1)}{2}p^{n-2}$&$p^{n-s-1}-\frac{k}{p}-1$&$\frac{k}{p}-1$\\ \hline
					$\frac{(p-1)}{2}p^{n-2}-\frac{(p-3)}{2}p^{\frac{n+s-3}{2}}$&$\frac{(p-1)k}{2p}$&$\frac{(p-1)}{2}(p^{n-s-1}-\frac{k}{p})$\\\hline
					$\frac{(p-1)}{2}p^{n-2}-\frac{(p+1)}{2}p^{\frac{n+s-3}{2}}$&$\frac{(p-1)}{2}(p^{n-s-1}-p^{\frac{n-s-1}{2}}$&$\frac{(p-1)}{2}(\frac{k}{p}-p^{\frac{n-s-1}{2}})$\\
					&$-\frac{k}{p})$&\\
					\hline
				\end{tabular}
			\end{table}
			Considering the punctured code $\widetilde{\mathcal{C}}_{\widetilde{D}_{f,SQ}}$ and $\widetilde{\mathcal{C}}_{\widetilde{D}_{f,NSQ}}$ defined by (6) and (7) of $\mathcal{C}_{D_{f,SQ}}$ and $\mathcal{C}_{D_{f,NSQ}}$ in Proposition 6, we have the following corollary.
			\begin{corollary}
				The punctured code $\widetilde{\mathcal{C}}_{\widetilde{D}_{f,SQ}}$ defined by (6) of $\mathcal{C}_{D_{f,SQ}}$ in Proposition 6 is a $[\frac{p^{n-1}+\epsilon_0p^{\frac{n+s-1}{2}} }{2},\ n]$ linear code and the punctured code $\widetilde{\mathcal{C}}_{\widetilde{D}_{f,NSQ}}$ defined by (7) of $\mathcal{C}_{D_{f,NSQ}}$ in Proposition 6 is a $[\frac{p^{n-1}-\epsilon_0p^{\frac{n+s-1}{2}} }{2},\ n]$ linear code. The weight distributions of $\widetilde{\mathcal{C}}_{\widetilde{D}_{f,SQ}}$ and $\widetilde{\mathcal{C}}_{\widetilde{D}_{f,NSQ}}$ are given by Tables 23 and 24, respectively. The dual code $\widetilde{\mathcal{C}}_{\widetilde{D}_{f,SQ}}^{\bot}$ is a $[\frac{p^{n-1}+\epsilon_0p^{\frac{n+s-1}{2}} }{2},\ \frac{p^{n-1}+\epsilon_0p^{\frac{n+s-1}{2}} }{2}-n,\ 3]$ linear code except for $p=3$, $n=3$ and $0\in B_-(f)$ and the dual code $\widetilde{\mathcal{C}}_{\widetilde{D}_{f,NSQ}}^{\bot}$ is a $[\frac{p^{n-1}-\epsilon_0p^{\frac{n+s-1}{2}}}{2},\ \frac{p^{n-1}-\epsilon_0p^{\frac{n+s-1}{2}}}{2}-n,\ 3]$ linear code except for $p= 3$, $n=3$ and $0\in B_+(f)$, which are almost optimal according to the sphere packing bound.
			\end{corollary}
			
			\begin{proof}
				The weight distributions of $\widetilde{\mathcal{C}}_{\widetilde{D}_{f,SQ}}$ and $\widetilde{\mathcal{C}}_{\widetilde{D}_{f,NSQ}}$ can be easily obtained by Proposition 6. The dimensions of $\widetilde{\mathcal{C}}_{\widetilde{D}_{f,SQ}}^{\bot}$ and $\widetilde{\mathcal{C}}_{\widetilde{D}_{f,NSQ}}^{\bot}$ are clearly $\frac{p^{n-1}+\epsilon_0p^{\frac{n+s-1}{2}}}{2}\\-n$ and $\frac{p^{n-1}-\epsilon_0p^{\frac{n+s-1}{2}}}{2}-n$, respectively. By the first four Pless power moments, for linear code $\widetilde{\mathcal{C}}_{\widetilde{D}_{f,SQ}}^{\bot}$, when $0\in B_+(f)$, we have $A_1^{\bot}=A_2^{\bot}=0$ and $A_3^{\bot}=\frac{1}{48}(p-1)((3p^2-6p-1)p^{n+s-2}+(4p-8)(p^{\frac{3n+s-3}{2}}-p^{\frac{n+s-1}{2}}-p^{n-1})+(p-1)^2p^{2n-3}-2(p^{n-s}-k)p^{\frac{n+3s-5}{2}}(p^2-2p-3))>0$; when $0\in B_-(f)$, we have $A_1^{\bot}=A_2^{\bot}=0$ and $A_3^{\bot}=\frac{1}{48}(p-1)((3p^2-6p-1)p^{n+s-2}-(4p-8)(p^{\frac{3n+s-3}{2}}+p^{n-1}-p^{\frac{n+s-1}{2} })+(p-1)^2p^{2n-3}+2k(p^2-2p-3)p^{\frac{n+3s-5}{2}})>0$  except for $p=3$ and $n=3$, so we can get that the minimum distance of $\widetilde{\mathcal{C}}_{\widetilde{D}_{f,SQ}}^{\bot}$ is $3$. For linear code $\widetilde{\mathcal{C}}_{\widetilde{D}_{f,NSQ}}^{\bot}$, when $0\in B_+(f)$, we have $A_1^{\bot}=A_2^{\bot}=0$ and $A_3^{\bot}=\frac{1}{48}(p-1)((3p^2-6p-1)p^{n+s-2}+(p-1)^2p^{2n-3}-(4p-8)(p^{n-1}-p^{\frac{n+s-1}{2}})-2k(p^2-2p-3)p^{\frac{n+3s-5}{2}}-(2p^2-4p+6)p^{\frac{3n+s-5}{2}})>0$ except for $p=3$ and $n=3$; when $0\in B_-(f)$, we have $A_1^{\bot}=A_2^{\bot}=0$ and $A_3^{\bot}=\frac{1}{48}(p-1)((3p^2-6p-1)p^{n+s-2}+(4p-8)(p^{\frac{3n+s-3}{2}}-p^{\frac{n+s-1}{2}}-p^{n-1})+(p-1)^2p^{2n-3}-2kp^{\frac{n+3s-5}{2}}(p^2-2p-3))>0$, so we can get that the minimum distance of $\widetilde{\mathcal{C}}_{\widetilde{D}_{f,NSQ}}^{\bot}$is $3$. According to the sphere packing bound, we have that  $\widetilde{\mathcal{C}}_{\widetilde{D}_{f,SQ}}^{\bot}$ and $\widetilde{\mathcal{C}}_{\widetilde{D}_{f,NSQ}}^{\bot}$ are almost optimal.\qed
			\end{proof}
		
			\begin{remark}
				By Remark 13, we easily have that linear codes $\widetilde{\mathcal{C}}_{\widetilde{D}_{f,SQ}}$ and $\widetilde{\mathcal{C}}_{\widetilde{D}_{f,NSQ}}$ are minimal for $0\le s\le n-5.$
			\end{remark}	
			
			Next, we verify Proposition 6 and Corollary 14 by Magma program for the following non-weakly regular plateaued functions.
		
			\begin{example}
				Consider $f(x):\ \mathbb{F}_{5}^4\longrightarrow\mathbb{F}_5,\  f(x_1,\ x_2,\ x_3,\ x_4)=x_1^2x_3^4+x_1^2+x_2x_3$, which is a non-weakly regular $1$-plateaued function belonging to $\mathcal{NWRF}$ with $0\in B_+(f)$, $k=25$, and the dual $f^*(x)$ of $f(x)$ is bent relative to $\mathrm{Supp}(\widehat{\chi_f})$. Then,  $\mathcal{C}_{D_{f,SQ}}$ constructed by (4) is a five-weight linear code with weight parameters $[300,\ 4,\ 200]$, weight enumerator $1+4z^{200}+40z^{220}+540z^{240}+20z^{260}+20z^{280}$ and the dual code $\mathcal{C}_{D_{f,SQ}}^{\bot}$ is a $[300,\ 296,\ 2]$ linear code. The punctured code $\widetilde{\mathcal{C}}_{\widetilde{D}_{f,SQ}}$ constructed by (6) is a $[75,\ 4,\ 50]$ linear code with weight enumerator $1+4z^{50}+40z^{55}+540z^{60}+20z^{65}+20z^{70}$ and the dual code $\widetilde{\mathcal{C}}_{\widetilde{D}_{f,SQ}}^{\bot}$ is a $[75,\ 71,\ 3]$ linear code, which is optimal according to the Code Table at http://www.codetables.de/.
			\end{example}
			\begin{example}
				Consider $f(x):\ \mathbb{F}_{5}^6\longrightarrow\mathbb{F}_5,\ f(x_1,\ x_2,\ x_3,\ x_4,\ x_5,\ x_6)=4x_1^2x_5^4+2x_1^2+x_2^2+x_3^2+x_4x_5$, which is a non-weakly regular $1$-plateaued function belonging to $\mathcal{NWRF}$ with $0\in B_-(f)$, $k=2500$, and the dual $f^*(x)$ of $f(x)$ is bent relative to $\mathrm{Supp}(\widehat{\chi_f})$. Then, $\mathcal{C}_{D_{f,SQ}}$ constructed by (4) is a five-weight linear code with parameters $[6000,\ 6,\ 4600]$, weight enumerator $1+500z^{4600}+200z^{4700}+13800z^{4800}+1000z^{4900}+124z^{5000}$ and the dual code $\mathcal{C}_{D_{f,SQ}}^{\bot}$ is a $[6000,\ 5994,\ 2]$ linear code. The punctured code $\widetilde{\mathcal{C}}_{\widetilde{D}_{f,SQ}}$ constructed by (6) is a $[1500,\ 6,\ 1150]$ linear code with weight enumerator $1+500z^{1150}+200z^{1175}+13800z^{1200}+1000z^{1225}+124z^{1250}$ and the dual code $\widetilde{\mathcal{C}}_{\widetilde{D}_{f,SQ}}^{\bot}$ is a $[1500,\ 1494, 3]$ linear code.
			\end{example}
			\begin{example}
				Consider $f(x):\ \mathbb{F}_{5}^6\longrightarrow\mathbb{F}_5,\  f(x_1,\ x_2,\ x_3,\ x_4,\ x_5,\ x_6)=x_1^2x_5^4+x_1^2+x_2^2+x_3^2+x_4x_5$, which is a non-weakly regular $1$-plateaued function belonging to $\mathcal{NWRF}$ with $0\in B_+(f)$, $k=625$, and the dual $f^*(x)$ of $f(x)$ is bent relative to $\mathrm{Supp}(\widehat{\chi_f})$. Then, $\mathcal{C}_{D_{f,NSQ}}$ constructed by (5) is a five-weight linear code with parameters $[6000,\ 6,\ 4600]$, weight enumerator $1+500z^{4600}+200z^{4700}+13800z^{4800}+1000z^{4900}+124z^{5000}$ and the dual code $\mathcal{C}_{D_{f,NSQ}}^{\bot}$ is a $[6000,\ 5994,\ 2]$ linear code. The punctured code $\widetilde{\mathcal{C}}_{\widetilde{D}_{f,NSQ}}$ constructed by (7) is a $[1500,\ 6,\ 1150]$ linear code with weight enumerator $1+500z^{1150}+200z^{1175}+13800z^{1200}+1000z^{1225}+124z^{1250}$ and the dual code $\widetilde{\mathcal{C}}_{\widetilde{D}_{f,NSQ}}^{\bot}$ is a $[1500,\ 1494,\ 3]$ linear code.
			\end{example}
			\begin{example}
				Consider $f(x):\ \mathbb{F}_{5}^4\longrightarrow\mathbb{F}_5,\  f(x_1,\ x_2,\ x_3,\ x_4)=4x_1^2x_3^4+2x_1^2+x_2x_3$,  which is a non-weakly regular $1$-plateaued function belonging to $\mathcal{NWRF}$ with $0\in B_-(f)$, $k=100$, and the dual $f^*(x)$ of $f(x)$ is bent relative to $\mathrm{Supp}(\widehat{\chi_f})$. Then, $\mathcal{C}_{D_{f,NSQ}}$ constructed by (5) is a five-weight linear code with parameters $[300,\ 4,\ 200]$, weight enumerator $1+4z^{200}+40z^{220}+540z^{240}+20z^{260}+20z^{280}$ and the dual code $\mathcal{C}_{D_{f,NSQ}}^{\bot}$ is a $[300,\ 296,\ 2]$ linear code. The punctured code $\widetilde{\mathcal{C}}_{\widetilde{D}_{f,NSQ}}$ constructed by (7) is a $[75,\ 4,\ 50]$ linear code with weight enumerator $1+4z^{50}+40z^{55}+540z^{60}+20z^{65}+20z^{70}$ and the dual code $\widetilde{\mathcal{C}}_{\widetilde{D}_{f,NSQ}}^{\bot}$ is a $[75,\ 71,\ 3]$ linear code, which is optimal according to the Code Table at http://www.codetables.de/.
			\end{example}

			According to Lemma 8, Theorem 6, Remark 13 and Corollary 14, we have the following results on  secret sharing schemes.
			\begin{corollary}
				Let $n+s$ be an odd integer with $0\le s\le n-5$ and $\mathcal{C}_{D_{f,SQ}}$, $\mathcal{C}_{D_{f,NSQ}}$ be the linear codes in Theorem 6 with generator matrix $G=[{g}_0,{g}_1,\cdots,\\\ {g}_{m-1}]$, where $m=\frac{(p-1)}{2}(p^{n-1}+\epsilon_0p^{\frac{n+s-1}{2}})$ for $\mathcal{C}_{D_{f,SQ}}$ and $m=\frac{(p-1)}{2}(p^{n-1}-\epsilon_0p^{\frac{n+s-1}{2}})$ for $\mathcal{C}_{D_{f,NSQ}}$. Then, in the secret sharing schemes based on $\mathcal{C}_{D_{f,SQ}}^{\bot}$ and $\mathcal{C}_{D_{f,NSQ}}^{\bot}$  with $d^{\bot}=2$, the number of participants is $m-1$ and there are $p^{n-1}$ minimal access sets. Moreover, if ${g}_j$ is a scalar multiple of ${g}_0$, $1\le j\le m-1$, then participant $P_j$ must be in every minimal access set, or else, 
				participant $P_j$ must be in $(p-1)p^{n-2}$ out of $p^{n-1}$  minimal access sets.
			\end{corollary}
			\begin{corollary}
				Let $n+s$ be an odd integer with $0\le s\le n-5$ and $\widetilde{\mathcal{C}}_{\widetilde{D}_{f,SQ}}$,  $\widetilde{\mathcal{C}}_{\widetilde{D}_{f,NSQ}}$ be the linear codes in Corollary 14 with generator matrix $G=[{g}_0,{g}_1,\cdots\\,{g}_{m-1}]$, where $m=\frac{p^{n-1}+\epsilon_0p^{\frac{n+s-1}{2}}}{2}$ for $\widetilde{\mathcal{C}}_{\widetilde{D}_{f,SQ}}$ and $m=\frac{p^{n-1}-\epsilon_0p^{\frac{n+s-1}{2}}}{2}$ for $\widetilde{\mathcal{C}}_{\widetilde{D}_{f,NSQ}}$. Then, in the secret sharing schemes based on $\widetilde{\mathcal{C}}_{\widetilde{D}_{f,SQ}}^{\bot}$ and $\widetilde{\mathcal{C}}_{\widetilde{D}_{f,NSQ}}^{\bot}$ with $d^{\bot}=3$, the number of participants is $m-1$ and there are $p^{n-1}$ minimal access sets. Moreover, every participant $P_j$ is involved in $(p-1)p^{n-2}$ out of $p^{n-1}$  minimal access sets.
			\end{corollary}

		Let $f(x)$ be a non-weakly regular bent function, i.e., $s=0$, in the following corollary, we give the parameters and weight distributions of $\mathcal{C}_{D_{f,SQ}}$, $\widetilde{\mathcal{C}}_{\widetilde{D}_{f,SQ}}$, $\mathcal{C}_{D_{f,NSQ}}$ and $\widetilde{\mathcal{C}}_{\widetilde{D}_{f,NSQ}}$ constructed by non-weakly regular bent functions.
		
		\begin{corollary}
			Let $f(x):\ \mathbb{F}_{p}^n\longrightarrow\mathbb{F}_p$ be a non-weakly regular bent function belonging to $\mathcal{NWRF}$ with $\#B_+(f)=k\ (k\ne0,\text{and}\  k\ne p^{n})$ and the dual $f^*(x)$ of $f(x)$ be bent. Then, when $n$ is even, $\mathcal{C}_{D_{f,SQ}}$ defined by (4) and $\mathcal{C}_{D_{f,NSQ}}$ defined by (5) are at most four-weight linear codes with parameters $[\frac{(p-1)}{2}(p^{n-1}-\epsilon_0p^{\frac{n}{2}-1}),\ n]$ and the weight distributions of $\mathcal{C}_{D_{f,SQ}}$ and $\mathcal{C}_{D_{f,NSQ}}$ are given by Tables 25 and 26, $\widetilde{\mathcal{C}}_{\widetilde{D}_{f,SQ}}$ defined by (6) and   $\widetilde{\mathcal{C}}_{\widetilde{D}_{f,NSQ}}$ defined by (7) are at most four-weight linear codes with parameters $[\frac{(p^{n-1}-\epsilon_0p^{\frac{n}{2}-1})}{2},\ n]$ and the weight distributions of $\widetilde{\mathcal{C}}_{\widetilde{D}_{f,SQ}}$ and   $\widetilde{\mathcal{C}}_{\widetilde{D}_{f,NSQ}}$ are given by Tables 27 and 28; when $n$ is odd, $\mathcal{C}_{D_{f,SQ}}$ defined by (4) is a  at most five-weight linear code with parameters
			$[\frac{(p-1)}{2}(p^{n-1}+\epsilon_0p^{\frac{n-1}{2}}),\ n]$, $\mathcal{C}_{D_{f,NSQ}}$ defined by (5) is a at most five-weight linear code with parameters 	$[\frac{(p-1)}{2}(p^{n-1}-\epsilon_0p^{\frac{n-1}{2}}),\ n]$, and the weight distributions of $\mathcal{C}_{D_{f,SQ}}$ and $\mathcal{C}_{D_{f,NSQ}}$ are given by Table 29 and 30, $\widetilde{\mathcal{C}}_{\widetilde{D}_{f,SQ}}$ defined by (6) is a  at most five-weight linear code with parameters 	$[\frac{(p^{n-1}+\epsilon_0p^{\frac{n-1}{2}})}{2},\ n]$,  $\widetilde{\mathcal{C}}_{\widetilde{D}_{f,NSQ}}$ defined by (7) is a  at most five-weight linear code with parameters 	$[\frac{(p^{n-1}-\epsilon_0p^{\frac{n-1}{2}})}{2},\ n]$, and the weight distributions of $\widetilde{\mathcal{C}}_{\widetilde{D}_{f,SQ}}$ and   $\widetilde{\mathcal{C}}_{\widetilde{D}_{f,NSQ}}$ are given by Tables 31 and 32. 
		\end{corollary}
		
		Next, we verify Corollary 17 by Magama program for the following non-weakly regular bent functions.
		\begin{table}
			\caption{The weight distributions of $\mathcal{C}_{D_{f,SQ}}$ and $\mathcal{C}_{D_{f,NSQ}}$  in Corollary 17 when $n$ is even and $0\in B_+(f)$}
			\begin{tabular}{|l|l|}
				\hline
				Weight &Multiplicity\\\hline
				0&1\\
				\hline
				$\frac{(p-1)^2}{2}p^{n-2}$&$\frac{(p+1)}{2}\frac{k}{p}+\frac{(p-1)}{2}p^{\frac{n}{2}-1}-1$\qquad\qquad\quad\hspace{2cm}\\
				\hline
				$\frac{(p-1)^2}{2}(p^{n-2}-2p^{\frac{n}{2}-2} )$&$\frac{(p+1)}{2}(p^{n-1}-\frac{k}{p})$\\
				\hline
				$\frac{(p-1)}{2}(p^{n-1}-p^{n-2}-2p^{\frac{n}{2}-1})$&$\frac{(p-1)}{2}(\frac{k}{p}-p^{\frac{n}{2}-1})$\\
				\hline
				$\frac{(p-1)}{2}(p^{n-1}-p^{n-2}+2p^{\frac{n}{2}-2})$&$ \frac{(p-1)}{2}(p^{n-1}-\frac{k}{p})$\\
				\hline
			\end{tabular}
		\end{table}
		\begin{table}
			\caption{The weight distributions of $\mathcal{C}_{D_{f,SQ}}$ and $\mathcal{C}_{D_{f,NSQ}}$ in Corollary 17 when $n$ is even and $0\in B_-(f)$}
			\begin{tabular}{|l|l|}
				\hline
				Weight &Multiplicity\\\hline
				0&1\\\hline
				$\frac{(p-1)^2}{2}(p^{n-2}+2p^{\frac{n}{2}-2})$&$\frac{(p+1)}{2}\frac{k}{p}$\\
				\hline
				$\frac{(p-1)^2}{2}p^{n-2}$&$\frac{(p+1)}{2}(p^{n-1}-\frac{k}{p})-\frac{(p-1)}{2}p^{\frac{n}{2}-1}-1$\\
				\hline
				$\frac{(p-1)}{2}(p^{n-1}-p^{n-2}-2p^{\frac{n}{2}-2})$&$\frac{(p-1)}{2}\frac{k}{p}$\\
				\hline
				$\frac{(p-1)}{2}(p^{n-1}-p^{n-2}+2p^{\frac{n}{2}-1})$&$\frac{(p-1)}{2}(p^{n-1}-\frac{k}{p}+p^{\frac{n}{2}-1})$\\
				\hline
			\end{tabular}
		\end{table}
		\begin{table}
			\caption{The weight distributions of $\widetilde{\mathcal{C}}_{\widetilde{D}_{f,SQ}}$ and $\widetilde{\mathcal{C}}_{\widetilde{D}_{f,NSQ}}$ in Corollary 17  when $n$ is even and $0\in B_+(f)$}
			\begin{tabular}{|l|l|}
				\hline
				Weight &Multiplicity\\\hline
				0&1\\
				\hline
				$\frac{(p-1)}{2}p^{n-2}$&$\frac{(p+1)}{2}\frac{k}{p}+\frac{(p-1)}{2}p^{\frac{n}{2}-1}-1$\qquad\qquad\quad\hspace{2cm}\\
				\hline
				$\frac{(p-1)}{2}(p^{n-2}-2p^{\frac{n}{2}-2} )$&$\frac{(p+1)}{2}(p^{n-1}-\frac{k}{p})$\\
				\hline
				$\frac{p^{n-1}-p^{n-2}-2p^{\frac{n}{2}-1}}{2}$&$\frac{(p-1)}{2}(\frac{k}{p}-p^{\frac{n}{2}-1})$\\
				\hline
				$\frac{p^{n-1}-p^{n-2}+2p^{\frac{n}{2}-2}}{2}$&$ \frac{(p-1)}{2}(p^{n-1}-\frac{k}{p})$\\
				\hline
			\end{tabular}
		\end{table}
		\begin{table}
			\caption{The weight distributions of $\widetilde{\mathcal{C}}_{\widetilde{D}_{f,SQ}}$ and $\widetilde{\mathcal{C}}_{\widetilde{D}_{f,NSQ}}$  in Corollary 17 when $n$ is even and $0\in B_-(f)$}
			\begin{tabular}{|l|l|}
				\hline
				Weight &Multiplicity\\\hline
				0&1\\\hline
				$\frac{(p-1)}{2}(p^{n-2}+2p^{\frac{n}{2}-2})$&$\frac{(p+1)}{2}\frac{k}{p}$\\
				\hline
				$\frac{(p-1)}{2}p^{n-2}$&$\frac{(p+1)}{2}(p^{n-1}-\frac{k}{p})-\frac{(p-1)}{2}p^{\frac{n}{2}-1}-1$\\
				\hline
				$\frac{p^{n-1}-p^{n-2}-2p^{\frac{n}{2}-2}}{2}$&$\frac{(p-1)}{2}\frac{k}{p}$\\
				\hline
				$\frac{p^{n-1}-p^{n-2}+2p^{\frac{n}{2}-1}}{2}$&$\frac{(p-1)}{2}(p^{n-s-1}-\frac{k}{p}+p^{\frac{n}{2}-1})$\\
				\hline
			\end{tabular}
		\end{table}
		\begin{table}
			\caption{The weight distribution of $\mathcal{C}_{D_{f,SQ}}$ in Corollary 17 when $n$ is odd and $0\in B_+(f)$ and the weight distribution of $\mathcal{C}_{D_{f,NSQ}}$ in Corollary 17 when $n$ is odd and $0\in B_-(f)$}
			\begin{tabular}{|l|l|l|}
				\hline
				Weight & Multiplicity of $\mathcal{C}_{D_{f,SQ}}$ &Multiplicity of $\mathcal{C}_{D_{f,NSQ}}$ \\ 
				&($0\in B_+(f))$&($0\in B_-(f))$\\\hline
				0&1&1\\ \hline
				$\frac{(p-1)^2}{2}(p^{n-2}+p^{\frac{n-3}{2}})$&$\frac{(p-1)}{2}(p^{n-1}-p^{\frac{n-1}{2}})$&$\frac{(p-1)}{2}(p^{n-1}-p^{\frac{n-1}{2}})$
				\\\hline
				$\frac{(p-1)^2}{2}p^{n-2}$&$\frac{k}{p}-1$&$p^{n-1}-\frac{k}{p}-1$\\ \hline
				$\frac{(p-1)^2}{2}(p^{n-2}+2p^{\frac{n-3}{2}})$&$p^{n-1}-\frac{k}{p}$&$\frac{k}{p}$\\ \hline
				$\frac{(p-1)^2}{2}p^{n-2}+\frac{(p^2-1)}{2}p^{\frac{n-3}{2}}$&$\frac{(p-1)}{2}(\frac{k}{p}+p^{\frac{n-1}{2}})$&$\frac{(p-1)}{2}(p^{n-1}-\frac{k}{p}+p^{\frac{n-1}{2}})$\\\hline
				$\frac{(p-1)^2}{2}p^{n-2}+\frac{(p-1)(p-3)}{2}p^{\frac{n-3}{2}}$&$\frac{(p-1)}{2}(p^{n-1}-\frac{k}{p})$&$\frac{(p-1)}{2}\frac{k}{p}$\\\hline
			\end{tabular}
		\end{table}
		\begin{table}
			\caption{The weight distribution of $\mathcal{C}_{D_{f,SQ}}$ in Corollary 17 when $n$ is odd and $0\in B_-(f)$ and the weight distribution of $\mathcal{C}_{D_{f,NSQ}}$ in Corollary 17 when $n$ is odd and $0\in B_+(f)$}
			
			\begin{tabular}{|l|l|l|}
				\hline
				Weight &Multiplicity of  $\mathcal{C}_{D_{f,SQ}}$ &Multiplicity of $\mathcal{C}_{D_{f,NSQ}}$\\
				&($0\in B_-(f)$)&($0\in B_+(f)$)\\\hline
				0&1&1\\ \hline
				$\frac{(p-1)^2}{2}(p^{n-2}-p^{\frac{n-3}{2}})$&$\frac{(p-1)}{2}(p^{n-1}+p^{\frac{n-1}{2}})$&$\frac{(p-1)}{2}(p^{n-1}+p^{\frac{n-1}{2}})$\\
				\hline
				$\frac{(p-1)^2}{2}(p^{n-2}-2p^{\frac{n-3}{2}})$&$\frac{k}{p}$&$p^{n-1}-\frac{k}{p}$\\ \hline
				$\frac{(p-1)^2}{2}p^{n-2}$&$p^{n-1}-\frac{k}{p}-1$&$\frac{k}{p}-1$\\ \hline
				$\frac{(p-1)^2}{2}p^{n-2}-\frac{(p-1)(p-3)}{2}p^{\frac{n-3}{2}}$&$\frac{(p-1)k}{2p}$&$\frac{(p-1)}{2}(p^{n-1}-\frac{k}{p})$\\\hline
				$\frac{(p-1)^2}{2}p^{n-2}-\frac{(p^2-1)}{2}p^{\frac{n-3}{2}}$&$\frac{(p-1)}{2}(p^{n-1}-p^{\frac{n-1}{2}}-\frac{k}{p})$&$\frac{(p-1)}{2}(\frac{k}{p}-p^{\frac{n-1}{2}})$\\
				\hline
			\end{tabular}
		\end{table}
		
		\begin{table}
			\caption{The weight distribution of $\widetilde{\mathcal{C}}_{\widetilde{D}_{f,SQ}}$ in Corollary 17 when $n$ is odd and $0\in B_+(f)$ and the weight distribution of $\widetilde{\mathcal{C}}_{\widetilde{D}_{f,NSQ}}$ in Corollary 17 when $n$ is odd and $0\in B_-(f)$}
			
			\begin{tabular}{|l|l|l|}
				\hline
				Weight & Multiplicity of $\widetilde{\mathcal{C}}_{\widetilde{D}_{f,SQ}}$ &Multiplicity of $\widetilde{\mathcal{C}}_{\widetilde{D}_{f,NSQ}}$ \\ 
				&($0\in B_+(f)$)&($0\in B_-(f)$)\\\hline
				0&1&1\\ \hline
				$\frac{(p-1)}{2}(p^{n-2}+p^{\frac{n-3}{2}})$&$\frac{(p-1)}{2}(p^{n-1}-p^{\frac{n-1}{2}})$&$\frac{(p-1)}{2}(p^{n-1}-p^{\frac{n-1}{2}})$\\
				\hline
				$\frac{(p-1)}{2}p^{n-2}$&$\frac{k}{p}-1$&$p^{n-1}-\frac{k}{p}-1$\\ \hline
				$\frac{(p-1)}{2}(p^{n-2}+2p^{\frac{n-3}{2}})$&$p^{n-1}-\frac{k}{p}$&$\frac{k}{p}$\\ \hline
				$\frac{(p-1)}{2}p^{n-2}+\frac{(p+1)}{2}p^{\frac{n-3}{2}}$&$\frac{(p-1)}{2}(\frac{k}{p}+p^{\frac{n-1}{2}})$&$\frac{(p-1)}{2}(p^{n-1}-\frac{k}{p}+p^{\frac{n-1}{2}})$\\\hline
				$\frac{(p-1)}{2}p^{n-2}+\frac{(p-3)}{2}p^{\frac{n-3}{2}}$&$\frac{(p-1)}{2}(p^{n-1}-\frac{k}{p})$&$\frac{(p-1)}{2}\frac{k}{p}$\\\hline
			\end{tabular}
		\end{table}
		\begin{table}
			\caption{The weight distribution of $\widetilde{\mathcal{C}}_{\widetilde{D}_{f,SQ}}$ in Corollary 17 when $n$ is odd and $0\in B_-(f)$ and the weight distribution of $\widetilde{\mathcal{C}}_{\widetilde{D}_{f,NSQ}}$ in Corollary 17 when $n$ is odd and $0\in B_+(f)$}
			\begin{tabular}{|l|l|l|}
				\hline
				Weight &Multiplicity of  $\widetilde{\mathcal{C}}_{\widetilde{D}_{f,SQ}}$ &Multiplicity of $\widetilde{\mathcal{C}}_{\widetilde{D}_{f,NSQ}}$\\ 
				&($0\in B_-(f)$)&($0\in B_+(f)$)\\\hline
				0&1&1\\ \hline
				$\frac{(p-1)}{2}(p^{n-2}-p^{\frac{n-3}{2}})$&$\frac{(p-1)}{2}(p^{n-1}+p^{\frac{n-1}{2}})$&$\frac{(p-1)}{2}(p^{n-1}+p^{\frac{n-1}{2}})$\\ \hline
				$\frac{(p-1)}{2}(p^{n-2}-2p^{\frac{n-3}{2}})$&$\frac{k}{p}$&$p^{n-1}-\frac{k}{p}$\\ \hline
				$\frac{(p-1)}{2}p^{n-2}$&$p^{n-1}-\frac{k}{p}-1$&$\frac{k}{p}-1$\\ \hline
				$\frac{(p-1)}{2}p^{n-2}-\frac{(p-3)}{2}p^{\frac{n-3}{2}}$&$\frac{(p-1)k}{2p}$&$\frac{(p-1)}{2}(p^{n-1}-\frac{k}{p})$\\\hline
				$\frac{(p-1)}{2}p^{n-2}-\frac{(p+1)}{2}p^{\frac{n-3}{2}}$&$\frac{(p-1)}{2}(p^{n-1}-p^{\frac{n-1}{2}}-\frac{k}{p})$&$\frac{(p-1)}{2}(\frac{k}{p}-p^{\frac{n-1}{2}})$\\
				\hline
			\end{tabular}
		\end{table}
		\begin{example}
			Consider $f(x):\ \mathbb{F}_{3}^4\longrightarrow\mathbb{F}_3$, $f(x_1, x_2, x_3, x_4)=2x_1^2x_4^2+2x_1^2+x_2^2+x_3x_4$,  which is a non-weakly regular bent function belonging to $\mathcal{NWRF}$ with $0\in B_+(f)$, $k=27$, and the dual $f^*(x)$ of $f(x)$ is bent. Then, $\mathcal{C}_{D_{f,SQ}}$ constructed by (4) and $\mathcal{C}_{D_{f,NSQ}}$ constructed by (5) are four-weight linear codes with parameters $[24,\ 4,\ 12]$, weight enumerator $1+6z^{12}+36z^{14}+20z^{18}+18z^{20}$. The dual codes $\mathcal{C}_{D_{f,SQ}}^{\bot}$ and $\mathcal{C}_{D_{f,NSQ}}^{\bot}$ are $[24,\ 20,\ 2]$ linear codes which are almost optimal according to the Code Table at http://www.codetables.de/. Moreover, the punctured codes $\widetilde{\mathcal{C}}_{\widetilde{D}_{f,SQ}}$ constructed by (6) and   $\widetilde{\mathcal{C}}_{\widetilde{D}_{f,NSQ}}$ constructed by (7) are $[12,\ 4,\ 6]$ linear codes with weight enumerator $1+6z^{6}+36z^{7}+20z^{9}+18z^{10}$, which are optimal according to the Code Table at http://ww\\w.codetables.de/. The dual codes $\widetilde{\mathcal{C}}_{\widetilde{D}_{f,SQ}}^{\bot}$ and   $\widetilde{\mathcal{C}}_{\widetilde{D}_{f,NSQ}}^{\bot}$ are $[12,\ 8,\ 3]$ linear codes which are optimal according to the Code Table at http://www.codetables.de/.
			\end{example}
		\begin{example}
			Consider $f(x):\ \mathbb{F}_{3}^4\longrightarrow\mathbb{F}_3$, $f(x_1, x_2, x_3, x_4)=x_1^2x_4^2+x_1^2+x_2^2+x_3x_4$, which is a non-weakly regular bent function belonging to $\mathcal{NWRF}$ with $0\in B_-(f)$, $k=54$, and the dual $f^*(x)$ is bent. Then, $\mathcal{C}_{D_{f,SQ}}$ constructed by (4) and $\mathcal{C}_{D_{f,NSQ}}$ constructed by (5) are four-weight linear codes with parameters $[30,\ 4,\ 16]$ , weight enumerator $1+18z^{16}+14z^{18}+36z^{22}+12z^{24}$. The dual codes $\mathcal{C}_{D_{f,SQ}}^{\bot}$ and $\mathcal{C}_{D_{f,NSQ}}^{\bot}$ are $[30,\ 26,\ 2]$ linear codes which are almost optimal according to the Code Table at http://www.codetables.de/. Moreover, the punctured codes $\widetilde{\mathcal{C}}_{\widetilde{D}_{f,SQ}}$ constructed by (6) and   $\widetilde{\mathcal{C}}_{\widetilde{D}_{f,NSQ}}$ constructed by (7) are $[15,\ 4,\ 8]$ linear codes with weight enumerator $1+18z^{8}+14z^{9}+36z^{11}+12z^{12}$, which are almost optimal according to the Code Table at http://www.codetables.de/. The dual codes $\widetilde{\mathcal{C}}_{\widetilde{D}_{f,SQ}}^{\bot}$ and   $\widetilde{\mathcal{C}}_{\widetilde{D}_{f,NSQ}}^{\bot}$ are $[15,\ 11,\ 3]$ linear codes which are optimal according to the Code Table at http://www.codetables.de/.
		\end{example}
	\begin{example}
		Consider $f(x):\ \mathbb{F}_{5}^5\longrightarrow\mathbb{F}_5$, $f(x_1, x_2, x_3,x_4, x_5)=x_1^2x_5^4+x_1^2+x_2^2+x_3^2+x_4x_5$, which is a non-weakly regular bent function belonging to $\mathcal{NWRF}$ with $0\in B_+(f)$, $k=625$, and the dual $f^*(x)$ is bent. Then, $\mathcal{C}_{D_{f,SQ}}$ constructed by (4) is a five-weight linear code with parameters $[1300,\ 5,\ 1000]$, weight enumerator $1+124z^{1000}+1000z^{1020}+1200z^{1040}+300z^{1060}+500z^{1080}$ and the dual code $\mathcal{C}_{D_{f,SQ}}^{\bot}$ is a $[1300,\ 1295, 2]$ linear code. The punctured code $\widetilde{\mathcal{C}}_{\widetilde{D}_{f,SQ}}$ constructed by (6) is a $[325,\ 5,\ 250]$ linear code with weight enumerator $1+124z^{250}+1000z^{255}+1200z^{260}+300z^{265}+500z^{270}$ and the dual code $\widetilde{\mathcal{C}}_{\widetilde{D}_{f,SQ}}^{\bot}$ is a $[325, 320,\ 3]$ linear code.
		\end{example}
	\begin{example}
		Consider $f(x):\ \mathbb{F}_{5}^5\longrightarrow\mathbb{F}_5$, $f(x_1, x_2, x_3,x_4, x_5)=4x_1^2x_5^4+2x_1^2+x_2^2+x_3^2+x_4x_5$, which is a non-weakly regular bent function belonging to $\mathcal{NWRF}$ with $0\in B_-(f)$, $k=2500$, and the dual $f^*(x)$ is bent. Then, $\mathcal{C}_{D_{f,SQ}}$ constructed by (4) is a five-weight linear code with parameters $[1200,\ 5,\ 920]$, weight enumerator $1+500z^{920}+200z^{940}+1300z^{960}+1000z^{980}+124z^{1000}$ and the dual code $\mathcal{C}_{D_{f,SQ}}^{\bot}$ is a $[1200,\ 1195, 2]$ linear code. The punctured code $\widetilde{\mathcal{C}}_{\widetilde{D}_{f,SQ}}$ constructed by (6) is a $[300,\ 5,\ 230]$ linear code with weight enumerator $1+500z^{230}+200z^{235}+1300z^{240}+1000z^{245}+124z^{250}$and the dual code $\widetilde{\mathcal{C}}_{\widetilde{D}_{f,SQ}}^{\bot}$ is a $[300, 295,\ 3]$ linear code.
	\end{example}
	\begin{example}
		Consider $f(x):\ \mathbb{F}_{5}^5\longrightarrow\mathbb{F}_5$, $f(x_1, x_2, x_3,x_4, x_5)=x_1^2x_5^4+x_1^2+x_2^2+x_3^2+x_4x_5$, which is a non-weakly regular bent function belonging to $\mathcal{NWRF}$ with $0\in B_+(f)$, $k=625$, and the dual $f^*(x)$ is bent. Then, $\mathcal{C}_{D_{f,NSQ}}$ constructed by (5) is a five-weight linear code with parameters $[1200,\ 5,\ 920]$, weight enumerator $1+500z^{920}+200z^{940}+1300z^{960}+1000z^{980}+124z^{1000}$ and the dual code $\mathcal{C}_{D_{f,NSQ}}^{\bot}$ is a $[1200,\ 1195, 2]$ linear code. The punctured code $\widetilde{\mathcal{C}}_{\widetilde{D}_{f,NSQ}}$ constructed by (7) is a $[300,\ 5,\ 230]$ linear code with weight enumerator $1+500z^{230}+200z^{235}+1300z^{240}+1000z^{245}+124z^{250}$and the dual code $\widetilde{\mathcal{C}}_{\widetilde{D}_{f,NSQ}}^{\bot}$ is a $[300, 295,\ 3]$ linear code.
	\end{example}
\begin{example}
	Consider $f(x):\ \mathbb{F}_{5}^5\longrightarrow\mathbb{F}_5$, $f(x_1, x_2, x_3,x_4, x_5)=4x_1^2x_5^4+2x_1^2+x_2^2+x_3^2+x_4x_5$, which is a non-weakly regular bent function belonging to $\mathcal{NWRF}$ with $0\in B_-(f)$, $k=2500$, and the dual $f^*(x)$ is bent. Then, $\mathcal{C}_{D_{f,NSQ}}$ constructed by (5) is a five-weight linear code with parameters $[1300,\ 5,\ 1000]$, weight enumerator $1+124z^{1000}+1000z^{1020}+1200z^{1040}+300z^{1060}+500z^{1080}$ and the dual code $\mathcal{C}_{D_{f,NSQ}}^{\bot}$ is a $[1300,\ 1295, 2]$ linear code. The punctured code $\widetilde{\mathcal{C}}_{\widetilde{D}_{f,NSQ}}$ constructed by (7) is a $[325,\ 5,\ 250]$ linear code with weight enumerator $1+124z^{250}+1000z^{255}+1200z^{260}+300z^{265}+500z^{270}$ and the dual code $\widetilde{\mathcal{C}}_{\widetilde{D}_{f,NSQ}}^{\bot}$ is a $[325, 320,\ 3]$ linear code.
\end{example}
\begin{table}
	\caption{Known $[n,\ k,\ d]$ linear codes with four weights over $\mathbb{F}_{p^t}$}
	
	\begin{tabular}{|l|l|l|l|l|}
		\hline
		p &n&k&d&Remarks\\ \hline
		odd&$q^{2n}$&$3n+1$&$(q-1)q^{n-1}(q^n-1)$&$t\ge1$, $q=p^t$\cite{Carlet1}\\\hline
		odd&$q^n$&$2n+1$&$(q-1)q^{n-1}-q^{\frac{n-1}{2}}$&$n>1$, $n$ is odd, \\&&&&$t\ge1$, $q=p^t$\cite{Li}\\\hline
		odd&$p^{n-1}+p^{m-1}$&$n+1$&$(p-1)p^{n-2}$&$m\ge1,\ n\ge2,$\\&&&&$ t=1$\cite{Wang}\\\hline
		odd&$p^{n-1}(1-p^{-\frac{r}{2}})$&$n$&$p^{n-2}(p-1-$&$0\le r\le n$, $r$ is \\&-1&&$2p^{-\frac{r}{2}+1})$&even, $t=1$\cite{Ta}\\
		&&&&\\\hline
		odd&$p^{n-1}(1+p^{-\frac{r}{2}})$&$n$&$p^{n-2}(p-1)$&$0\le r\le n$, $r$ is \\&-1&&&even, $t=1$\cite{Ta}\\\hline
		odd&$p^{n-1}(1+(p-1)$&$n$&$p^{n-2}(p-1)$&$0\le r\le n$, $r$ is \\
		&$p^{-\frac{r}{2}})-1$&&&even, $t=1$\cite{Ta}\\	\hline
		odd&$p^{n-1}(1-(p-1)$&$n$&$p^{n-2}(p-1)$& $r$ is even, $t=1$\\
		&$p^{-\frac{r}{2}})-1$&&$(1-p^{1-\frac{r}{2}})$&\cite{Ta}\\	\hline
		
		$2$&$2^n$&$2n+1$&$2^{n-1}-2^{\frac{n-1}{2}}$&$n\ge5$, $n$ is odd,\\&&&&$t=1$\cite{Xiangcan}\\\hline
		odd&$p^2+p$&4&$p(p-1)$&$t=1$\cite{ZX}\\\hline
		odd&$p^6+p^5$&8&$(p-1)p^5$&$t=1$\cite{ZX}\\\hline
		odd, $p>3$&$p^{n-1}-(p-1)$&$n$&$(p-1)(p^{n-2}-p^{\frac{n}{2}-1})$&$n>2$, $n$ is even, \\
		&$p^{\frac{n}{2}-1}-1$&&&$t=1$\cite{XU}\\\hline
		$p\equiv3\ (\text{mod}\ 4)$&$\frac{p-1}{2}p^{n-2}$&$n$&$\frac{p-1}{2}(p^{n-2}-p^{n-3}
		$&$p|n$, $n\ge4$, $n$ \\
		&&&$-p^{\frac{n}{2}-2})$&is even,$t=1$\cite{Duan}\\\hline
		prime&$q^{2n}-q^n$&$3n$&$(q-1)(q^{2n-1}$&$n\ge1$,$q=p^t$,\\&&&$-2q^{n-1})$&$t\ge1$,$q\ne2$\cite{ZHU}\\\hline			
		
		prime&$(q^n-1)^2$&$3n$&$(q-1)(q^{2n-1}$&$n\ge1$,$q=p^t$,\\&&&$-3q^{n-1})$&$t\ge1$,$q\ne2$\cite{ZHU}\\\hline		
		prime&$\frac{(q^n-1)^2}{q-1}$&$3n$&$q^{2n-1}-q^n-q^{n-1}$&$n\ge1$,$q=p^t$,\\&&&&$t\ge1$,$q\ne2$\cite{ZHU}\\\hline
	\end{tabular}
\end{table}
\begin{table}
	\caption{Known $[n,\ k,\ d]$ linear codes with five weights over $\mathbb{F}_{p^t}$}
	
	\begin{tabular}{|l|l|l|l|l|}
		\hline
		p &n&k&d&Remarks\\ \hline
		odd &$q^n-1$&$2n$ &$(q-1)(q^{n-1}-q^{\frac{n-2}{2}})$&$n$ is even,$t\ge 1$, \\&&&& $q=p^t$ \cite{Li}\\\hline
		odd &$p^{2n-1}-(p-1)p^{n-1}$&$2n+1$&$(p-1)(p^{2n-2}$&$n\ge 2$, $t=1$ \\
		&$-1$&&$-p^{n-1}-1)$&\cite {Wang}\\\hline
		odd&$p^{2n-1}+p^{n+d-1}$&$2n+1$&$(p-1)p^{2n-2}$& $n/ d\equiv 0 (\mathrm{mod}\ 2)$,\\&&&& $t=1$ \cite {Wang}\\ \hline
		odd &$p^{n-1}(1+p^{-\frac{r-1}{2}})$&$n$& $(p-1)p^{n-2}$& $0\le r \le n$ , $r$ \\&$-1$&&&is odd, $t=1$ \cite{Ta}\\\hline
		odd& $p^{n-1}(1-p^{-\frac{r-1}{2}})$&$n$&$p^{n-2}((p-1)$& $0\le r \le n$, $r$ \\&$-1$&&$-(p+1)p^{-\frac{r-1}{2}})$&is odd, $t=1$ \cite{Ta}\\\hline
		odd & $p^{n-1}-1$&$n$&$p^{n-2}(p-1)(1-p^{-\frac{r-1}{2}})$& $0\le r \le n$, $r$ \\&&&&is odd, $t=1$ \cite{Ta}\\\hline
		odd &$p^{n-1}-1$& $n$& $(p-1)(p^{n-2}-p^{\frac{n-3}{2}})$& $n$ is odd, \\&&&&$p\mid n$, $t=1$ \cite{li}\\\hline
		odd & $p^{n-1}+p^{\frac{n+1}{2}-1}-1$&$n$& $(p-1)p^{n-2}$&  $n$ is odd, \\&&&&$p\nmid n$, $t=1$ \cite{li}\\\hline
		odd  & $p^{n-1}-p^{\frac{n+1}{2}-1}-1$&$n$& $(p-1)p^{n-2}$&  $n$ is odd, \\&&&$-(p+1)p^{\frac{n+1}{2}-2}$&$p\nmid n$, $t=1$ \cite{li} \\\hline
		odd &$p^{n-1}(p^n-p-1)+1$&$2n-1$&$p^{n-1}(p^n-p^{n-1}-p+1)$&$n\ge 2$ is even,\\
		& & & $-(p-1)p^{\frac{n}{2}-1}(p^{n-1}-1)$& $t=1$ \cite{Xiang}\\\hline
		2&$2^n-1$&$2n$&$2^{n-1}-2^{\frac{n}{2}}$&$n\ge 4$ is even,\\&&&& $t=1$ \cite{Xiangcan}\\\hline
		2&$2^n-2$&$2n-1$&$2^{n-1}-2^{\frac{n}{2}}$&$n\ge 4$ is even,\\&&&& $t=1$ \cite{Xiangcan}\\\hline
		2&$2^n-3$&$2n-2$&$2^{n-1}-2^{\frac{n}{2}}$&$n\ge 4$ is even,\\&&&& $t=1$ \cite{Xiangcan}\\\hline
		2&$2^n-4$&$2n-3$&$2^{n-1}-2^{\frac{n}{2}}$&$n\ge 4$ is even,\\&&&& $t=1$ \cite{Xiangcan}\\\hline
		odd &$p^n-2$&$2n-1$&$(p-1)
		(p^{n-1}-p^{\frac{n-2}{2}})$&$n\ge 2$ is even, \\
		&&&&$t=1$ \cite{Xiangcan}\\\hline
		odd &$p^n-1$&$2n$&$(p-1)
		(p^{n-1}-p^{\frac{n-2}{2}})$&$n\ge 2$ is even, \\
		&&&&$t=1$ \cite{Xiangcan}\\\hline
		odd &$p^n-p$&$2n-2$&$(p-1)
		(p^{n-1}-p^{\frac{n-2}{2}})$&$n\ge 2$ is even,\\&&&& $t=1$ \cite{Xiangcan}\\\hline
		odd &$\frac{p-1}{2}(p^2+p)$&4&$\frac{(p-1)^2}{2}p$&$p\equiv 3(\mathrm{mod}\ 4)$, \\&&&&$t=1$ \cite{ZhangX}\\\hline
		odd &$\frac{p-1}{2}(p^6+p^5)$&8&$\frac{(p-1)^2}{2}p^5$&$p\equiv 3(\mathrm{mod}\ 4)$, \\&&&&$t=1$ \cite{ZhangX}\\\hline
		odd &$p^{4n-2}+p^{3n-2}$&$4n$&$p^{4n-3}(p-1)$&$n>1$,\\&&&& $t=1$ \cite{ZX}\\\hline
		odd&$p^{4n-2}+p^{3n+d-2}$&$4n$&$p^{4n-3}(p-1)$&$n> d+1$, \\
		&&&&$n/d\equiv 0(\mathrm{mod}\ 2)$, \\&&&&$t=1$\cite{ZX}\\\hline
		odd&$p^{2n-1}$&$2n$&$(p-1)(p^{2n-2}$&$0\le d,l\le n$ \\
		$p>3$&&&$-p^{\frac{2n+d+l-3}{2}})$&$d+l$ is odd\\&&&&$t=1$ \cite{Yang}\\\hline
	\end{tabular}
\end{table}
\section{Conclusion}
			\qquad
			In this paper, for the first time, we constructed some three-weight and at most five-weight linear codes from non-weakly regular plateaued functions based on the first and the second generic constructions. Specially, using non-weakly regular bent functions, i.e., non-weakly regular $0$-plateaued functions, we obtained some three-weight, at most four-weight and at most five-weight linear codes. The parameters and weight distributions of the codes constructed from two special classes of non-weakly regular plateaued functions were completely determined. In Remark 5 and Remark 10, we showed that the three-weight linear codes constructed in this paper have the same parameters with the codes in \cite{Mesnager2,Mesnager3}. In Table 33, we list the known four-weight linear codes over $\mathbb{F}_{p^t}$, where $p$ is a prime and $t$ is a positive integer. By comparing the four-weight linear codes in Table 33 with ours, we have that when $n$ is even and $r=n$ the $[p^{n-1}(1+(p-1)p^{-\frac{r}{2}})-1,\ n,\ p^{n-2}(p-1)]$ linear code in \cite{Ta} have the same parameters with $\mathcal{C}_{D_f}$ in Corollary 10. In Table 34, we list the known five-weight linear codes over $\mathbb{F}_{p^t}$, where $p$ is a prime and $t$ is a positive integer. By comparing the five-weight linear codes in Table 34 with ours, we have that when $s=0$ and $r=n$, the code in Theorem 1 have the same parameters with $[p^{n-1}-1,\ n, \ p^{n-2}(p-1)(1-p^{-\frac{r-1}{2}})]$ linear code in \cite{Ta} and $[p^{n-1},\ n,\ (p-1)(p^{n-2}-p^{\frac{n-3}{2}})]$ linear code in \cite{li}. For other codes in Table 33 and Table 34, our codes are different from them.
			We also obtained some optimal and almost optimal linear codes which are listed below:
			\\
			(\romannumeral 1) A family of AMDS codes with parameters $[p^2-1,\ p^2-4,\ 3]$.\\
			(\romannumeral 2) Some MDS codes with parameters $[p-1,\ 2,\ p-2]$, $[p-1,\ p-3,\ 3]$, $[p+1,\ 3,\ p-1]$ and $[p+1,\ p-2,\ 4]$.\\
			(\romannumeral 3) Some almost optimal codes, according to the sphere packing bound, with parameters $[\frac{p^{n-1}\pm(p-1)p^{\frac{n+s}{2}-1}-1}{p-1},\ \frac{p^{n-1}\pm(p-1)p^{\frac{n+s}{2}-1}-1}{p-1}-n,\ 3]$, $[\frac{p^{n-1}-1}{p-1},\frac{p^{n-1}-1}{p-1}\\-n,\ 3]$, $[\frac{p^{n-1} \pm p^{\frac{n+s}{2}-1}}{2}, \frac{p^{n-1}\pm p^{\frac{n+s}{2}-1}}{2}-n,\ 3]$, and  $[\frac{p^{n-1}\pm p^{\frac{n+s-1}{2}} }{2},\ \frac{p^{n-1}\pm p^{\frac{n+s-1}{2}} }{2}\\-n,\ 3]$ .\\
			Finally, based on the dual codes of the constructed linear codes, we presented some secret sharing schemes with nice access structures.

			\section*{Acknowledgements}               \quad This research is supported by the National Key Research and Development Program of China (Grant Nos. 2022YFA1005000
			and 2018YFA0704703), the National Natural Science Foundation of China (Grant Nos. 12141108, 62371259, 12226336), the Fundamental Research Funds for the Central Universities of China (Nankai University), the Nankai Zhide Foundation.

			
			
		\end{document}